\documentclass[journal,10pt]{IEEEtran}

\usepackage{amsmath,amssymb,amsfonts}
\usepackage{bm} 
\usepackage{amsthm}
\usepackage{bbm, dsfont}
\usepackage{mathtools}
\usepackage{cite} 
\usepackage{graphics} 
\usepackage{float}
\usepackage{booktabs}
\usepackage{epstopdf}
\usepackage{url} 
\usepackage{paralist}
\usepackage{lipsum}
\usepackage{mathtools}
\usepackage{breqn}
\usepackage{cuted}
\usepackage{graphicx}
\usepackage[caption=false,font=footnotesize]{subfig}

\usepackage{optidef}
\usepackage{algorithm,algorithmic}
\usepackage{xcolor}
\usepackage[framemethod=TikZ]{mdframed}
\usepackage{balance}
\usepackage{comment}
\usepackage{cleveref}

\allowdisplaybreaks 

\newtheorem{lem}{Lemma}
\newtheorem{prop}{Proposition}

\newcounter{tempEquationCounter}
\newcounter{thisEquationNumber}
\newenvironment{floatEq}
{\setcounter{thisEquationNumber}{\value{equation}}\addtocounter{equation}{1}
\begin{figure*}[!t]
\normalsize\setcounter{tempEquationCounter}{\value{equation}}
\setcounter{equation}{\value{thisEquationNumber}}
}
{\setcounter{equation}{\value{tempEquationCounter}}
\hrulefill\vspace*{4pt}
\end{figure*}}

\begin{document}
%
\title{Robust Secure ISAC: How RSMA and Active RIS Manage Eavesdropper's Spatial Uncertainty}


\author{A. Abdelaziz Salem, 
        Saeed Abdallah*,~\IEEEmembership{Member,~IEEE}, 
        Mohamed Saad,~\IEEEmembership{Senior Member,~IEEE},\\
        Khawla Alnajjar,~\IEEEmembership{Member,~IEEE},  
        and Mahmoud A. Albreem,~\IEEEmembership{Senior Member,~IEEE}
\thanks{A. A. Salem is with the Smart Automation and Communication Technologies (SACT) Research Center,  Research Institute of Sciences and Engineering (RISE), University of Sharjah, PO Box 26666, Sharjah, UAE. A. A. Salem is also with the Department of Electronics and Electrical Communications Engineering, Faculty of Electronic Engineering, Menoufia University, Menouf 32952, Menoufia, Egypt, e-mails: \texttt{ahmed.abdalaziz40@el-eng.menofia.edu.eg}.}
\thanks{\noindent Mahmoud A. Albreem, Khawla Alnajjar and Saeed Abdallah are with the Department of Electrical Engineering, University of Sharjah, Sharjah, UAE, e-mail: \texttt{\{malbreem,kalnajjar,sabdallah\}@sharjah.ac.ae}.}
\thanks{\noindent Mohamed Saad is with the Department of Computer Engineering, University of Sharjah, Sharjah, UAE, e-mail: \texttt{msaad@sharjah.ac.ae}.}
\thanks{$^*$ Corresponding author.}
\thanks{This work was supported in part by the Distributed and Networked Systems (DNS) Research Group, associated with the SACT Research Center, RISE, Operating Grant number 150410, University of Sharjah.}
}

{}

\maketitle

\begin{abstract}
Incorporating rate splitting multiple access (RSMA) into integrated sensing and communication (ISAC) presents a significant security challenge, particularly in scenarios where the location of a potential eavesdropper (Eve) is unidentified. Splitting users' messages into common and private streams exposes them to eavesdropping, with the common stream dedicated for sensing and accessible to multiple users. In response to this challenge, this paper proposes a novel approach that leverages active reconfigurable intelligent surface (RIS) aided beamforming and artificial noise (AN) to enhance the security of RSMA-enabled ISAC. Specifically, we first derive the ergodic private secrecy rate (EPSR) based on mathematical approximation of the average Eve channel gain. An optimization problem is then formulated to maximize the minimum EPSR, while satisfying the minimum required thresholds on ergodic common secrecy rate, radar sensing and RIS power budget. To address this non-convex problem, a novel optimization strategy is developed, whereby we alternatively optimize the transmit beamforming matrix for the common and private streams, rate splitting, AN, RIS reflection coefficient matrix, and radar receive beamformer. Successive convex approximation (SCA) and Majorization-Minimization (MM) are employed to convexify the beamforming and RIS sub-problems. Simulations  are conducted to showcase the effectiveness of the proposed framework against established benchmarks.
\end{abstract}

\begin{IEEEkeywords}
Active reconfigurable intelligent surfaces (ARIS), Integrated sensing and communication (ISAC), Rate splitting multiple access (RSMA), Robust physical layer security.

\end{IEEEkeywords}

\IEEEpeerreviewmaketitle

\section{Introduction}\label{sec1}
As researchers develop the sixth generation (6G) of wireless networks, numerous promising applications like autonomous vehicles, connected robotics, and augmented reality are driving the integration of sensing and communication. By exploiting a fully shared spectrum and dual-functional beamforming, integrated sensing and communication (ISAC) systems can achieve significantly better spectral and energy efficiency \cite{zhang2021enabling}. Generally, the design concept of ISAC is classified into three main categories: communication-centric, radar-centric, and joint design. The communication-centric approach focuses on integrating sensing capabilities into existing communication platforms, utilizing tailored communication waveforms to meet sensing requirements. The radar-centric scheme incorporates communication functions into existing radar systems by embedding communication information into sensing waveforms for data transmission. The third category involves developing new ISAC waveforms instead of utilizing pre-existing communication or sensing signals. The latter approach offers the advantage of providing greater flexibility in simultaneously supporting data transmission and target sensing, resulting in a higher degree of freedom \cite{zhang2021overview}.

 Although ISAC offers significant performance improvements, there remain several critical challenges that require attention. ISAC signals may potentially suffer from poor propagation, especially in the complex wireless environment where the availability of line-of-sight (LoS) paths becomes rare. The use of relays and unmanned aerial vehicles (UAVs) can potentially boost ISAC performance in such environments, but it also entails additional hardware costs and energy consumption. These challenges will be exacerbated in 6G systems, which envision ultra-dense deployment and data rates in the range of terabits per second (Tbps) \cite{you2021towards}. To tackle these issues, considerable attention has been given to the development of reconfigurable intelligent surfaces (RIS), also known as intelligent reflecting surfaces (IRS). The RIS is a two-dimensional array of small passive (or active) reflectors, controlled by a smart controller. Each reflector is capable of adjusting the phase shift and/or amplitude of incoming signals to enhance signal propagation in favorable directions. Due to its ability to establish strong virtual LoS links and manipulate propagation conditions with minimal power consumption and deployment cost, the RIS is viewed as a highly promising technology for assisting wireless communication, user localization, and environment sensing \cite{yu2023active} in 6G wireless systems.

Recently, RISs have been exploited to optimize spectral/energy efficiency \cite{yu2022ris} or outage performance \cite{guo2020outage}, transmit passive information \cite{yan2019passive}, optimize wireless power transfer \cite{yang2021reconfigurable}, and enhance physical-layer security against eavesdropping \cite{guan2020intelligent, dong2020enhancing}, etc. Moreover, RISs have been integrated into recent applications of millimeter-wave localization \cite{elzanaty2021reconfigurable} and radar systems \cite{buzzi2021radar} to mitigate user positioning errors or enhance target detection capabilities. Given the notable performance enhancements achieved by RISs in the communication and sensing domains, researchers have focused on exploring RIS-enabled ISAC systems. Initially, the work in \cite{jiang2021intelligent} employed a RIS to assist a single-user dual-function radar and communication (DFRC) system, where joint optimization of the RIS phase shifts and transmit beamforming was proposed. Subsequently, various works extended the RIS-enabled ISAC designs to multi-user multi-target scenarios \cite{tong2021joint, wang2021joint, luo2022joint}, ISAC-aided dual RISs \cite{he2022ris}, hybrid RIS models \cite{sankar2022beamforming}, discrete phase shift\cite{wang2021joint_discrete} models, and practical-sized targets \cite{xing2023joint}. Notably, most existing RIS-enabled ISAC designs aim to optimize either communication or sensing performance while ensuring the target performance requirements of the other. This allows for a trade-off between communication and sensing by adjusting performance thresholds color{red}within constraints.\color{black}

While RISs have proven effective in enhancing ISAC performance across various applications, the challenges of spectrum scarcity and interference management continue to be significant bottlenecks impacting performance. In fact, the 4th generation (4G) and 5th generation (5G) wireless communication systems interfere  with radar systems in the S-band, C-band, and possibly the millimeter-wave (mmWave) band. This may lead to spectrum misuse and hinder the achievement of higher data rates in future wireless systems. Although ISAC provides a long-term solution via enabling spectrum sharing between communication and radar systems \cite{he2022energy, xiong2023fundamental}, resorting to orthogonal access schemes continues to constrain ISAC performance. To unleash the communication and sensing capabilities in ISAC systems, emerging technologies of rate splitting multiple access (RSMA) and
RIS can be integrated. 

RSMA is recognized as a robust interference management strategy in multi-antenna wireless networks, which offers substantial gains in spectral efficiency through its non-orthogonal transmission capabilities \cite{mishra2022rate}. The main concept of RSMA is to divide user messages into private and common components. The common components are encoded  into shared streams for multiple users, while the private components are independently encoded into individual streams for each corresponding user. This approach allows RSMA to partially decode interference and treat the remaining interference as noise, unlike space division multiple access (SDMA), which treats all interference as noise, and non-orthogonal multiple access (NOMA), which fully decodes interference \cite{mao2022rate}.

In the context of RSMA-aided ISAC, the authors in \cite{xu2021rate} proposed weighted sum rate maximization and mean square error minimization of the approximated radar beampattern problem, where they jointly design message splits and precoders. This approach effectively manages interference from radar to communication without requiring an additional radar sequence. Motivated by designing waveforms for an RSMA-aided ISAC system, the work in \cite{yin2022rate} maximizes the minimum fairness rate and minimizes the Cram\'er--Rao bound for target estimation, while satisfying the power constraint of each transmitting antenna. Besides, RIS is deployed in a RSMA-assisted ISAC system in \cite{chen2023joint}, where the radar SNR is maximized by jointly optimizing the rate splitting and beamforming vectors at the BS and the RIS. 
While the above-mentioned studies investigate the effectiveness of RIS and RSMA in ISAC system, significant security challenges associated with rate splitting remain unresolved. Firstly, dividing users' messages into common and private components increases susceptibility to eavesdropping, since the common stream rate is intended for decoding by multiple users. This is in contrast to SDMA, where the lack of message splitting reduces the sum-rate but makes the system more resilient to eavesdropping. Furthermore, the presence of unidentified terminals with unknown positions as potential eavesdroppers (Eves) poses a significant challenge to ISAC systems. Therefore, there remains an open question regarding how to ensure secure RIS-aided ISAC while accounting for the potential target sensing in the presence of these unknown Eves.

In order to leverage the benefits of both RIS and RSMA, recent research has explored their integration from communication systems perspective, such as maximizing the energy efficiency \cite{niu2023active, narottama2024quantum}, improving the outage performance \cite{bansal2021rate}, and enhancing physical layer security (PLS) \cite{gao2022rate}. In the context of RIS-assisted secure ISAC, the work in \cite{hua2023secure} has regarded the targets as potential Eves and introduced a framework for the joint optimization of the communication, sensing, and secrecy rates, where the sensing beampattern gain is maximized. The work in \cite{xing2023reconfigurable} aims to maximize the approximate ergodic achievable secrecy rate, while simultaneously ensuring the minimum communication performance for the legitimate user and the minimum sensing performance for the target. However, none of the prior studies have investigated the integration of active RIS and RSMA to contend with the unknown location of Eve in the ISAC system. This gap in research serves as a key motivation for our work presented in this paper. Therefore, we propose a novel joint communication, sensing, and security optimization framework for an active RIS-enabled RSMA-ISAC system, wherein the exact position and the channel associated with the Eves are assumed unknown. Our contributions are summarized as follows.

\begin{itemize}
    \item We consider active RIS-enabled RSMA for a robust secure ISAC system, in which  an Eve is assumed to potentially appear in a certain region with an unknown position, and artificial noise (AN) is employed to enhance system security. A joint secure communication and sensing optimization problem is formulated to maximize the minimum ergodic private secrecy rate (MaxMin-EPSR) over all possible RIS-Eve channels, while ensuring the minimum requirement of ergodic common secrecy rate (ECSR), the minimum radar output SNR requirement for target sensing, and active RIS as well as BS power budget constraints. 
    \item The formulated problem is non-convex and intractable. To this end, the MaxMin-EPSR and ECSR are firstly approximated into a tractable form by deriving the mathematical expressions of the ergodic rate of the common and private streams at Eve based on the approximated mean of the Eve channel gain.
    \item Accordingly, we develop a novel convexification analysis based on the alternative optimization approach to jointly optimize the transmit beamforming, AN, active RIS reflection coefficients, common rate splitting, and radar receive beamforming. Specifically, the common rate splitting, transmit beamforming, and AN sub-problem is solved via successive convex approximation (SCA). Then, the active RIS reflection matrix and common rate splitting sub-problem is solved by exploiting Majorization-Minimization (MM) and SCA. Moreover, the radar receive beamformer sub-problem is categorized as a Rayleigh-quotient problem, resulting in a solution expressed in closed form.  
    \item Simulations are conducted to examine the effectiveness of the proposed framework against active RIS-enabled SDMA (ARIS-SDMA), passive RIS-enabled RSMA (RRIS-RSMA), and passive RIS-enabled SDMA (PRIS-SDMA). Our findings showcase the superiority of our suggested framework in terms of ergodic private secrecy rate (EPSR) and ECSR.
\end{itemize}

The subsequent sections of this paper are organized as follows. Section \ref{sec2} presents the system model and the formulation of the RSMA-based MaxMin EPSR problem. In Section \ref{sec3}, we provide a closed form approximation for the Eve channel, which is exploited to reformulate EPSR and ECSR in tractable form. Following this, Section \ref{sec4} showcases the performance evaluation of the proposed scheme compared to several benchmarks. Finally, concluding remarks are provided in Section \ref{sec5}.

\textit{Notations:} Vector notations are denoted by lowercase boldface letters, while uppercase boldface represents  matrices. Scalars are expressed by normal face letters. $\mathbb{C}^{M \times N}$ denotes the complex space with dimensions $M \times N$. Moreover, $\mathrm{Tr}(.)$, $\mathrm{diag}(.)$, $\mathrm{vec(.)}$, $(.)^*$, and $(.)^H$ stand for the trace operator, diagonal matrix, vectorization operator, conjugate, and conjugate transpose, respectively. The symbol $\otimes$ gives the Kronecker product. $\Re \left\{\cdot \right\}$ and $\Im \left\{ \cdot \right\}$ refer to the real and imaginary parts, respectively. ${{\mathbf{I}}_{M}}$ signifies the $M\times M$ identity matrix. Besides, $\mathbf{x}\sim\mathcal{C}\mathcal{N}\left( \pmb{\mu},{{\mathbf{R}}} \right)$ defines a circularly symmetric complex Gaussian random vector with mean $\pmb{\mu}$ and covariance matrix $\mathbf{R}$. $\left\| . \right\|_{F}$ and $\left\| . \right\|_2$ denote
the Frobenius norm and the norm of a vector, respectively.

\section{System Model and problem formulation}\label{sec2}
In this paper, we consider a RSMA-enabled active RIS for a robust secure ISAC system, where the base station (BS) is equipped with a uniform linear array (ULA) composed of $M$ antennas. In Fig.~\ref{sys_model}, an active RIS with $N$ elements is dedicated for sensing radar target and securing the $K$ single-antenna communication users against the Eve's spatial ambiguity. In practical terms, we assume that the Eve's position is not precisely known to the BS because the Eve is typically an uncooperative user who intercepts information passively. It is reasonable to assume that the Eve's location follows a uniform distribution within a region ${{\mathcal{R}}_{E}}\triangleq \left\{ {{d}_{RE}},{{\theta }_{E}} \right\}$, where ${{d}_{RE}}= \left[ {{d}_{1}},{{d}_{2}} \right]$ represents the range of the distances where the Eve can be located with respect to (w.r.t) the RIS, and ${{\theta }_{E}}= \left[ {{\theta }_{1}},{{\theta }_{2}} \right]$ gives the range of the azimuth angle-of-departures (AODs) of the Eve w.r.t the RIS. On the other hand, we assume that the channels between the BS and the RIS, as well as between the RIS and the user equipments (UEs), are known. Estimates of these channels can be obtained using existing techniques proposed for RIS-aided systems \cite{wei2021channel}.
Let ${{\varphi }_{n}}={{\beta }_{n}}{{e}^{j{{\xi }_{n}}}},n=1,...,N,$ be the reflection coefficient of the $n$-th active RIS element, where ${{\beta }_{n}}\le {{\beta }^{\max }}$ and ${{\xi }_{n}}\in \left[ 0,2\pi  \right)$ denote the corresponding amplitude and phase shift. Then, the RIS reflection coefficient matrix is expressed by $\mathbf{\Phi }=\mathrm{diag}\left( \pmb{\theta } \right)\in {{\mathbb{C}}^{N\times N}}$, where $\pmb{\theta }={{\left[ {{\varphi }_{1}},...,{{\varphi }_{N}} \right]}^{T}}\in {{\mathbb{C}}^{N\times 1}}$. 
\begin{figure}[t]
	\centering{\includegraphics[width=\columnwidth]{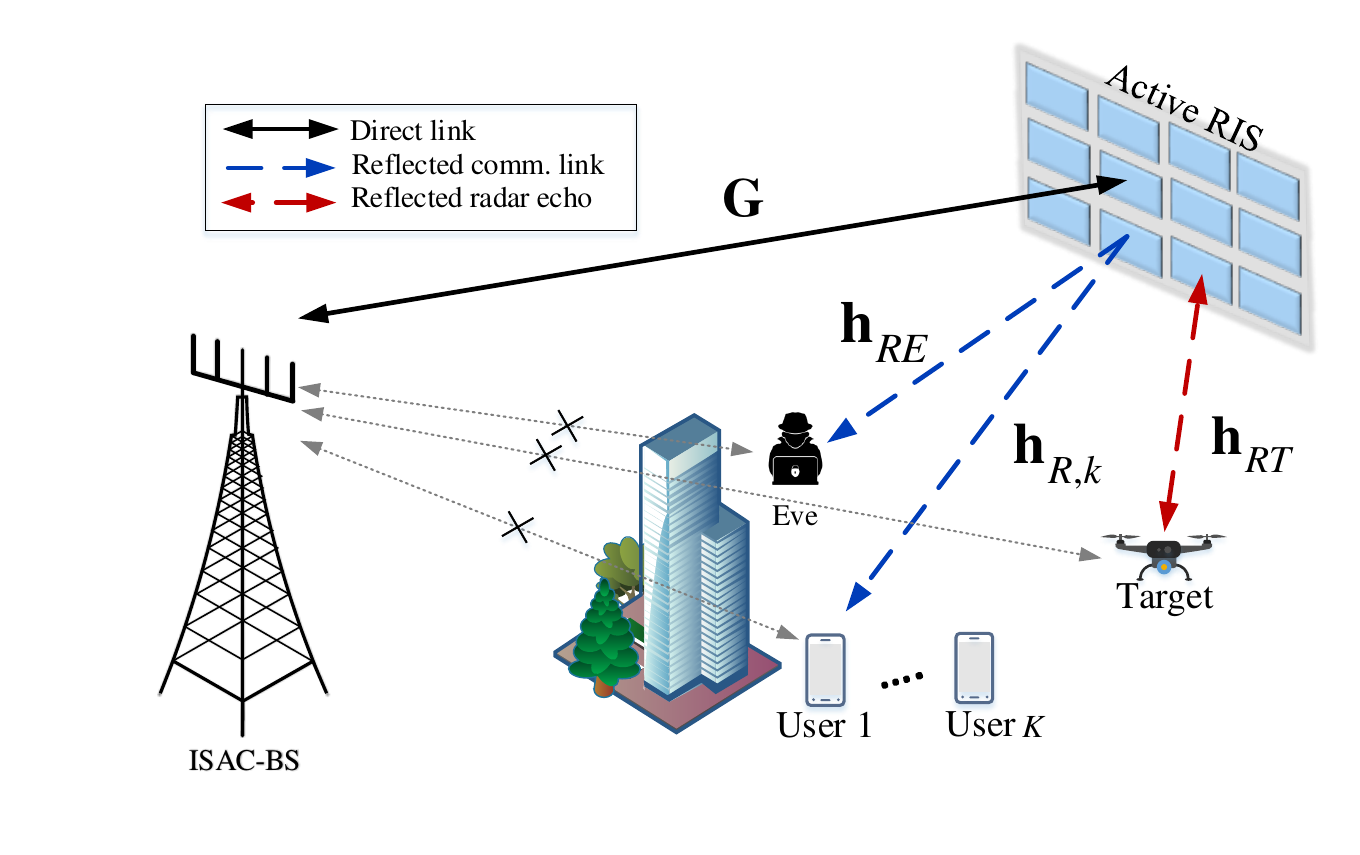}}
    \captionsetup{skip=3pt}
	\caption{Active RIS-enabled secure RSMA-ISAC system.\label{sys_model}}
\end{figure}

At the BS side, common and private streams are precoded along with the artificial noise\footnote{Dividing users' messages into common and private components increases susceptibility to eavesdropping, where the common stream is intended for decoding by multiple users. Thus, the transmitted signal is customized to guard against eavesdropping by integrating the AN.} (AN). In specific, the message ${{x}_{k}}$ of UE $k$ is split into common and private parts, denoted as ${{x}_{\mathrm{c},k}}$ and ${{x}_{\mathrm{p},k}}$, respectively. The common messages of the users, $\left\{ {{x}_{\mathrm{c},k}} \right\}_{k=1}^{K}$, are encoded into one single common stream ${{s}_{0}}$, whereas the private messages are separately encoded into the private streams $\left\{ {{s}_{k}} \right\}_{k=1}^{K}$. To this end, the transmitted signal is given as 
\begin{equation} \label{Trans_RSMA_AN}
    \mathbf{x}={{\mathbf{w}}_{0}}{{s}_{0}}+\sum\limits_{k=1}^{K}{{{\mathbf{w}}_{k}}{{s}_{k}}}+\mathbf{z}=\mathbf{Ws}+\mathbf{z},
\end{equation}
where the data streams and the transmit precoding matrix are denoted as $\mathbf{s}={{[{{s}_{0}},{{s}_{1}},...,{{s}_{K}}]}^{T}}\in {{\mathbb{C}}^{\left( K+1 \right)\times 1}}$ and $\mathbf{W}=[{{\mathbf{w}}_{0}},{{\mathbf{w}}_{1}},...,{{\mathbf{w}}_{K}}]\in {{\mathbb{C}}^{M\times \left( K+1 \right)}}$, respectively. Besides, the AN is denoted as $\mathbf{z}\in {{\mathbb{C}}^{M\times 1}}$, where $\mathbf{z}\sim\mathcal{C}\mathcal{N}\left( \mathbf{0},{{\mathbf{R}}} \right)$, where ${{\mathbf{R}}}\underline{\succ }\mathbf{0}$ represents the corresponding covariance matrix. At the UE side, successive interference cancellation (SIC) is employed to sequentially eliminate the common stream, allowing the users to extract their private streams. Furthermore, we assume that the direct paths of the BS-UE, BS-Eve, and BS-target are severely obstructed, so that the RIS steps in as a crucial player to shape alternative virtual LoS links. 


\subsection{Security and communication model} \label{sec2-1}
The received signals by $k$-th user and Eve are respectively expressed as 
\begin{align}
{{y}_{k}}=&\mathbf{h}_{R,k}^{H}{{\mathbf{\Phi }}^{H}}\mathbf{G}\mathbf{x}+ \mathbf{h}_{R,k}^{H}{{\mathbf{\Phi }}^{H}}{{\mathbf{n}}_{R}}+{{n}_{k}}, \label{Rx_userk} \\ 
{{y}_{E}}=&\mathbf{h}_{RE}^{H}{{\mathbf{\Phi }}^{H}}\mathbf{G}\mathbf{x}+ \mathbf{h}_{RE}^{H}{{\mathbf{\Phi }}^{H}}{{\mathbf{n}}_{R}}+{{n}_{E}}, \label{Rx_Eve} 
\end{align}
where ${{n}_{k}}\sim\mathcal{C}\mathcal{N}\left( 0,\sigma _{k}^{2} \right)$, ${{n}_{E}}\sim\mathcal{C}\mathcal{N}\left( 0,\sigma _{E}^{2} \right)$, ${{\mathbf{n}}_{R}}\sim\mathcal{C}\mathcal{N}\left( \mathbf{0},\sigma _{R}^{2}{{\mathbf{I}}_{N}} \right)$ represent the additive white Gaussian noise (AWGN) at the $k$-th user, the Eve and the RIS amplification noise, respectively. Furthermore,  $\left\{ {{\mathbf{h}}_{R,k}},{{\mathbf{h}}_{RE}} \right\}\in {{\mathbb{C}}^{N\times 1}}$ and $\mathbf{G}\in {{\mathbb{C}}^{N\times M}}$ represent the links of RIS-user$k$, RIS-Eve, and BS-RIS, respectively. Without loss of generality, the channel coefficients are assumed Rician distributed as in \cite{wu2019intelligent}.  
The cascaded channel is denoted as $\mathbf{h}_{i}^{H}{{\mathbf{\Phi }}^{H}}\mathbf{G}={{\pmb{\theta }}^{H}} \mathrm{diag}\left( \mathbf{h}_{i}^{H} \right)\mathbf{G}$. Thus, the decoding Signal-to-Interference-plus-Noise Ratio (SINR) of common stream at the  $k$-th UE and Eve are expressed as 
\begin{equation} \label{com_SINRs_eqs}
    {\gamma }_{{{s}_{0}},k} =\frac{{{\left| {{\pmb{\theta }}^{H}}{{{\mathbf{\bar{H}}}}_{k}}{{\mathbf{w}}_{0}} \right|}^{2}}}{\mathcal{I}_{s_0,k} \left( \mathbf{W}, \mathbf{\Phi}, \mathbf{z} \right)},  \quad 
    {\gamma }_{{{s}_{0}},E}=\frac{{{\left| {{\pmb{\theta }}^{H}}{{{\mathbf{\bar{H}}}}_{E}}{{\mathbf{w}}_{0}} \right|}^{2}}}{\mathcal{I}_{s_0,E} \left( \mathbf{W}, \mathbf{\Phi}, \mathbf{z} \right)}, 
\end{equation}
where the corresponding interference power on the received common stream is expressed  by $\mathcal{I}_{s_0,k} = \sum\limits_{{k}'=1}^{K}{{{\left| {{\pmb{\theta }}^{H}}{{{\mathbf{\bar{H}}}}_{k}}{{\mathbf{w}}_{{{k}'}}} \right|}^{2}}}+{{\left| {{\pmb{\theta }}^{H}}{{{\mathbf{\bar{H}}}}_{k}}{{\mathbf{z}}} \right|}^{2}}{{\bar{\sigma }}_{R,k}}$ with  ${{\bar{\sigma }}_{R,k}} = \sigma _{R}^{2}{{\left\| {{\pmb{\theta }}^{H}}\mathrm{diag}\left( \mathbf{\bar{h}}_{R,k}^{H} \right) \right\|}^{2}}+1$, and $\mathcal{I}_{s_0,E} = \sum\limits_{k=1}^{K}{{{\left| {{\pmb{\theta }}^{H}}{{{\mathbf{\bar{H}}}}_{E}}{{\mathbf{w}}_{k}} \right|}^{2}}}+{{\left| {{\pmb{\theta }}^{H}}{{{\mathbf{\bar{H}}}}_{E}}{{\mathbf{z}}_{\mathrm{AN}}} \right|}^{2}}+{{\bar{\sigma }}_{E}}$ with ${{\bar{\sigma }}_{E}}=1+ \sigma _{R}^{2}{{\left\| {{\left( {{\pmb{\theta }}^{\left( t \right)}} \right)}^{H}}\mathrm{diag}\left( \mathbf{\bar{h}}_{RE}^{H} \right) \right\|}^{2}}$, respectively.  Moreover, we define  in which ${{\mathbf{\bar{H}}}_{k}}=\sigma _{k}^{-1}\text{diag}\left( \mathbf{h}_{R,k}^{H} \right)\mathbf{G}, \mathbf{\bar{h}}_{R,k}^{H}=\sigma _{k}^{-1}\mathbf{h}_{R,k}^{H}, {{\mathbf{\bar{H}}}_{E}}=\sigma _{E}^{-1}\text{diag}\left( \mathbf{h}_{RE}^{H} \right)\mathbf{G}$, and $\mathbf{\bar{h}}_{RE}^{H}=\sigma _{E}^{-1}\mathbf{h}_{RE}^{H}$. Then, the achievable common rate rates are given by ${{R}_{{{s}_{0}},k}}={{\log }_{2}}\left( 1+{{\gamma }_{{{s}_{0}},k}} \right)$ and ${{R}_{{{s}_{0}},E}}={{\log }_{2}}\left( 1+{{\gamma }_{{{s}_{0}},E}} \right)$, respectively. To protect the common stream for multiple users and prevent Eve from decoding it, the following constraint must be met. 
\begin{equation} \label{ergodic_com_rate}
     \underset{k}{\textrm{Min}}\hspace{0.05cm}{{R}_{{{s}_{0}},k}}-{{\mathbb{E}}_{{{\mathbf{h}}_{RE}}}}\left\{ {{R}_{{{s}_{0}},E}} \right\}  \ge \sum\limits_{k=1}^{K}{{{r}_{k}}},
\end{equation}
where ${{r}_{k}}\ge 0$ represents the partial common rate allocated to UE $k$.

Following that, the common stream is subtracted from the received signal to decode the private streams. Thus, the achievable SINRs of the $k$-th private stream at the UE and the Eve are respectively denoted as 
\begin{equation} \label{Priv_SINRs_eqs} 
    {{\gamma }_{k}}=\frac{{{\left| {{\pmb{\theta }}^{H}}{{{\mathbf{\bar{H}}}}_{k}}{{\mathbf{w}}_{k}} \right|}^{2}}}{\mathcal{I}_{k} \left( \mathbf{W}, \mathbf{\Phi}, \mathbf{z} \right)}, \quad
	{{\gamma }_{k,E}}=\frac{{{\left| {{\pmb{\theta }}^{H}}{{{\mathbf{\bar{H}}}}_{E}}{{\mathbf{w}}_{k}} \right|}^{2}}}{\mathcal{I}_{k,E} \left( \mathbf{W}, \mathbf{\Phi}, \mathbf{z} \right)},  
\end{equation}
where $\mathcal{I}_{k} =\sum\limits_{i=1,i\ne k}^{K}{{{\left| {{\pmb{\theta }}^{H}}{{{\mathbf{\bar{H}}}}_{k}}{{\mathbf{w}}_{i}} \right|}^{2}}}+{{\left| {{\pmb{\theta }}^{H}}{{{\mathbf{\bar{H}}}}_{k}}\mathbf{z} \right|}^{2}}+{{\bar{\sigma }}_{R,k}}$ and  $\mathcal{I}_{k,E} = {{\left| {{\pmb{\theta }}^{H}}{{{\mathbf{\bar{H}}}}_{E}}{{\mathbf{w}}_{0}} \right|}^{2}}+\sum\limits_{i=1,i\ne k}^{K}{{{\left| {{\pmb{\theta }}^{H}}{{{\mathbf{\bar{H}}}}_{E}}{{\mathbf{w}}_{i}} \right|}^{2}}}+{{\left| {{\pmb{\theta }}^{H}}{{{\mathbf{\bar{H}}}}_{E}}\mathbf{z} \right|}^{2}}+{{\bar{\sigma }}_{E}}$.
According to (\ref{Priv_SINRs_eqs}), the achievable $k$-th private rates at user $k$ and the Eve are given as ${{R}_{k}}={{\log }_{2}}\left( 1+{{\gamma }_{k}} \right)$ and ${{R}_{k,E}}={{\log }_{2}}\left( 1+{{\gamma }_{k,E}} \right)$, respectively. Moreover, the EPSR of stream $k$ can be expressed as $\mathrm{SR}_k={{r}_{k}}+{{\left[ {{R}_{k}}-{{\mathbb{E}}_{{{\mathbf{h}}_{RE}}}}\left\{ {{R}_{k,E}} \right\} \right]}^{+}}$, where ${{\left[ x \right]}^{+}}=\max \left( 0,x \right)$.

\subsection{Sensing model} \label{sec2-2}
In addition to providing security, the signal $\mathbf{x}$ is also employed for target detection. The received echo signal at the BS is mathematically expressed as
\begin{equation} \label{RX_radar}
    \mathbf{y}=\varsigma {{\mathbf{G}}^{H}}\mathbf{\Phi }{{\mathbf{H}}_{RT}} \mathbf{\Phi} \left( \mathbf{ Gx} + {{\mathbf{n}}_{R}} \right)  + {{\mathbf{G}}^{H}}\mathbf{\Phi }{{\mathbf{\bar{n}}}_{R}}+\mathbf{n},
\end{equation}
where $\mathbf{n}\sim\mathcal{C}\mathcal{N}\left( 0,{{\sigma }^{2}}{{\mathbf{I}}_{M}} \right)$ and $\varsigma $, with $\mathbb{E}\left\{ {{\left| \varsigma  \right|}^{2}} \right\}={{\zeta }^{2}}$, denotes the AWGN at the radar receiver and the radar cross section (RCS). Moreover, ${{\mathbf{\bar{n}}}_{R}}\sim\mathcal{C}\mathcal{N}\left( \mathbf{0},\sigma _{R}^{2}{{\mathbf{I}}_{N}} \right)$ represents the RIS dynamic noise, which is inserted back with the radar echo. Also,    ${{\mathbf{{n}}}_{R}}$ denotes the forward active RIS noise towards the target, which is independent and identically distributed (i.i.d) with  ${{\mathbf{\bar{n}}}_{R}}$. Besides, ${{\mathbf{h}}_{RT}}\in {{\mathbb{C}}^{N\times 1}}$ and ${{\mathbf{H}}_{RT}} \in {{\mathbb{C}}^{N\times N}}$ refer to RIS-target channel and the target response, which is modelled as \cite{song2024cramer}.

Then, the received echo (\ref{RX_radar}) is post-processed with the receive filter $\mathbf{u}\in {{\mathbb{C}}^{M\times 1}}$. Given that the BS transmits a distinct pattern of AN, it is logical to infer that BS can detect changes in AN caused by target interaction, rendering it useful for target sensing. Consequently, the radar output SNR can be calculated as follows.
\begin{equation} \label{Radar_outSNR}
    \gamma =\frac{{{\zeta }^{2}}{{\left\| {{\mathbf{u}}^{H}}{{\mathbf{H}}_{T}}\mathbf{W} \right\|}^{2}}+{{\zeta }^{2}}{{\left| {{\mathbf{u}}^{H}}{{\mathbf{H}}_{T}}\mathbf{z} \right|}^{2}}}{{{{\bar{\sigma }}}_{R}}},
\end{equation}
where ${{\bar{\sigma }}_{R}}={{\zeta }^{2}}\sigma _{R}^{2}{{\left\| {{\mathbf{u}}^{H}}{{\mathbf{H}}_{0}} \right\|}^{2}}+\sigma _{R}^{2}{{\left\| {{\mathbf{u}}^{H}}{{\mathbf{H}}_{1}} \right\|}^{2}}+{{\sigma }^{2}}{{\left\| {{\mathbf{u}}^{H}} \right\|}^{2}}$, ${{\mathbf{H}}_{T}}={{\mathbf{G}}^{H}}\mathbf{\Phi }{{\mathbf{H}}_{RT}}\mathbf{\Phi G}$, ${{\mathbf{H}}_{0}}={{\mathbf{G}}^{H}}\mathbf{\Phi }{{\mathbf{H}}_{RT}}\mathbf{\Phi }$, and ${{\mathbf{H}}_{1}}={{\mathbf{G}}^{H}}\mathbf{\Phi }$.

\subsection{Problem formulation} \label{sec2-3}
In this section, we aim to optimize the communication, security and sensing for the active RIS-aided RSMA ISAC system in the presence of spatial ambiguity regarding the location of the Eve, i.e., unknown $\mathbf{h}_{RE}$. Therefore, the optimization problem is formulated to maximize the minimal EPSR, while ensuring that the common stream is undetected by the Eve, meeting the sensing performance requirements, and considering the power budget at the BS and active RIS. In this context, the common rate allocation $\mathbf{r} \triangleq \left\{r_1,...,r_K\right\}$, AN $\mathbf{z}$, transmit precoding matrix $\mathbf{W}$, active RIS reflection coefficients $\mathbf{\Phi }$, and radar receive filter $\mathbf{u}$ are jointly designed. Mathematically, the overall optimization problem is formulated as follows,
\begin{subequations}\label{opt_problem_orginal}
	\begin{align}
		\underset{\textbf{r}\ge\textbf{0},\textbf{u},{\textbf{z}},\pmb{\Phi },\textbf{W}}{\textrm{Maximize}}\hspace{0.15cm} & \underset{k}{\textrm{Min}}\hspace{0.15cm} {{r}_{k}}+{{\left[ {{R}_{k}}-{{\mathbb{E}}_{{{\mathbf{h}}_{RE}}}}\left\{ {{R}_{k,E}} \right\} \right]}^{+}}, \label{opt1_obj} \\
		\textrm{s.t.: } 
		& {{R}_{{{s}_{0}},k}}-{{\mathbb{E}}_{{{\mathbf{h}}_{RE}}}}\left\{ {{R}_{{{s}_{0}},E}} \right\}\ge \sum\limits_{k=1}^{K}{{{r}_{k}}}, \forall k \label{opt1_c1},  \\
		& \gamma \ge {{\Gamma }_{r}} \label{opt1_c3}, \\
  	& \sum\limits_{i=0}^{K}{{{\left\| {{\mathbf{w}}_{i}} \right\|}^{2}}}+{{\left\| \mathbf{z} \right\|}^{2}}\le {{P}^{\mathrm{BS}}} \label{opt1_c4}, \\
		& \mathcal{P}\left( \mathbf{W},\mathbf{\Phi },\mathbf{z} \right)\le {{P}^{\mathrm{RIS}}} \label{opt1_c5}, \\
        & {{\beta }_{n}}\le {{\beta }^{\max }},\quad \forall n\label{opt1_c6},
		\end{align}
\end{subequations}
where $\mathcal{P} = {{\zeta }^{2}}\left\| \mathbf{\Phi }{{\mathbf{H}}_{RT}}\mathbf{\Phi GW} \right\|_{F}^{2} + \left\| \mathbf{\Phi GW} \right\|_{F}^{2} + \left\| \mathbf{\Phi G}\mathbf{z} \right\|_{2}^{2} + {{\zeta }^{2}}\left\| \mathbf{\Phi }{{\mathbf{H}}_{RT}}\mathbf{\Phi G}\mathbf{z} \right\|_{2}^{2} + 2\sigma _{R}^{2}\left\| \mathbf{\Phi } \right\|_{F}^{2} + {{\zeta }^{2}}\sigma _{R}^{2}\left\| \mathbf{\Phi }{{\mathbf{H}}_{RT}}\mathbf{\Phi } \right\|_{F}^{2}$.
The constraint (\ref{opt1_c1}) ensures the security of the ergodic common rate. (\ref{opt1_c3}) defines the radar sensing requirement, while (\ref{opt1_c4}) and (\ref{opt1_c5}) limit the power consumption at the BS and the active RIS, respectively. In addition, (\ref{opt1_c6}) controls the limit on the amplification coefficients of the active RIS elements. Unfortunately, the optimization problem in (\ref{opt_problem_orginal}) poses significant challenges due to its highly non-convex nature and the coupling among multiple optimization variables. Moreover, the presence of the expectation operator over logarithmic functions in (\ref{opt1_obj}) and (\ref{opt1_c1}) introduces further complexity, where incorporating the unknown Eve channel results in an intractable formulation.

\section{Proposed solution}\label{sec3}
To tackle the aforementioned challenges, we derive closed-form expressions for the ergodic rate terms ${{\mathbb{E}}_{{{\mathbf{h}}_{RE}}}}\left\{ {{R}_{k,E}} \right\}$ and ${{\mathbb{E}}_{{{\mathbf{h}}_{RE}}}}\left\{ {{R}_{{{s}_{0}},E}} \right\}$. Accordingly, we recast (\ref{opt_problem_orginal}) into convex sub-problems by adopting the MM and SCA methods, along with matrix inversion theory and matrix fractional approximation \cite{dong2020secure} to transform (\ref{opt1_obj}), (\ref{opt1_c1}), and (\ref{opt1_c3}) into quadratic forms. Then, an alternating optimization approach is invoked to solve the resultant problems.  

Before delving into the problem solution, we determine ${{\widehat{\mathbf{H}}}_{RE}} \triangleq \mathbb{E}\left\{ {{\mathbf{h}}_{RE}}\mathbf{h}_{RE}^{H} \right\}$ to facilitate deriving the ergodic rate expressions through the following proposition.  
\begin{prop} \label{prop1}
    Given that the Eve's location is uniformly distributed within region $\mathcal{R}_E$ and the RIS-Eve channel coefficients follow the Rician distribution, ${{\widehat{\mathbf{H}}}_{RE}}$ can be approximated as    
    \begin{multline} \label{eve_esti_CH_ver3}
        {{\widehat{\mathbf{H}}}_{RE}} = {{\left( \frac{1}{{{d}_{0}}} \right)}^{-{{\alpha }_{RE}}}}\frac{2\ell \left( d_{1}^{2-{{\alpha }_{RE}}}-d_{2}^{2-{{\alpha }_{RE}}} \right)}{\left( 1+\kappa  \right)\left( d_{2}^{2}-d_{1}^{2} \right)\left( {{\alpha }_{RE}}-2 \right)} \\
        \times \Biggl( \frac{\Delta \theta_{E}\kappa}{2(\theta_{2}-\theta_{1})} \sum_{i=0}^{N_{\theta_{E}}-1} \mathbf{A}(\theta_{1}+i\Delta \theta_{E}) \\
         + \mathbf{A}(\theta_{1}+(i+1)\Delta \theta_{E}) + \mathbf{I}_{N} \Biggr),
    \end{multline}
where $\mathbf{A}\left( {{\theta }_{E}} \right)=\mathbf{a}\left( {{\theta }_{E}} \right){{\mathbf{a}}^{H}}\left( {{\theta }_{E}} \right)$.  
\end{prop}
\begin{proof}
	See Appendix \ref{appendix1}.
\end{proof}

\subsection{Transmit beamforming, AN, and common rate optimization} \label{sec3-1}
Here, we consider that the values of $\mathbf{\Phi }$ and $\mathbf{u}$ are fixed. Therefore, the optimization problem (\ref{opt_problem_orginal}) is reformulated as 
\begin{subequations}\label{BF1_opt}
	\begin{align}
		\underset{\textbf{r}\ge\textbf{0},{\textbf{z}},\textbf{W}}{\textrm{Maximize}}\hspace{0.15cm} & \tau, \label{BF1_obj} \\
		\textrm{s.t.: } 
        & {{r}_{k}}+ {{R}_{k}}-{{\mathbb{E}}_{{{\mathbf{h}}_{RE}}}}\left\{ {{R}_{k,E}} \right\} \ge \tau , \quad \forall k \label{BF1_c1},  \\
		& {{R}_{{{s}_{0}},k}}-{{\mathbb{E}}_{{{\mathbf{h}}_{RE}}}}\left\{ {{R}_{{{s}_{0}},E}} \right\}\ge \sum\limits_{k=1}^{K}{{{r}_{k}}}, \quad \forall k \label{BF1_c2},  \\
		& \gamma \ge {{\Gamma }_{r}} \label{BF1_c4}, \\
  	& \sum\limits_{i=0}^{K}{{{\left\| {{\mathbf{w}}_{i}} \right\|}^{2}}}+{{\left\| \mathbf{z} \right\|}^{2}}\le {{P}^{\mathrm{BS}}} \label{BF1_c5}, \\
		& \bar{\mathcal{P}}\left( \mathbf{W},\mathbf{z}\right)\le {{\overline{P}}^{\mathrm{RIS}}} \label{BF1_c6}, 
		\end{align}
\end{subequations}
where $\bar{\mathcal{P}}\left( \mathbf{W},\mathbf{z}\right)=\left\| \mathbf{\Phi GW} \right\|_{F}^{2}\text{+ }{{\zeta }^{2}}\left\| \mathbf{\Phi }{{\mathbf{H}}_{RT}}\mathbf{\Phi GW} \right\|_{F}^{2}+{{\zeta }^{2}}\left\| \mathbf{\Phi }{{\mathbf{H}}_{RT}}\mathbf{\Phi G}\mathbf{z}\right\|_{2}^{2}+\left\| \mathbf{\Phi G}\mathbf{z}\right\|_{2}^{2}$ and ${{\overline{P}}^{\mathrm{RIS}}}={{P}^{\mathrm{RIS}}}  - {{\zeta }^{2}}\sigma _{R}^{2}\left\| \mathbf{\Phi }{{\mathbf{H}}_{RT}}\mathbf{\Phi } \right\|_{F}^{2} - 2\sigma _{R}^{2}\left\| \mathbf{\Phi } \right\|_{F}^{2}$. 

From (\ref{Radar_outSNR}), the radar SNR constraint (\ref{BF1_c4}) is rewritten as 
\begin{equation}\label{radar_constraint_BF_v1}
    \frac{{{\zeta }^{2}}}{{{{\bar{\sigma }}}_{R}}}\left( {{\mathbf{w}}^{H}}{{{\mathbf{\tilde{H}}}}_{T}}\mathbf{w}+{{\left| {{{\mathbf{\bar{h}}}}_{T}}\mathbf{z} \right|}^{2}} \right)\ge {{\Gamma }_{r}},
\end{equation}
where ${{\mathbf{\bar{h}}}_{T}}={{\mathbf{u}}^{H}}{{\mathbf{H}}_{T}},\text{~}{{\mathbf{\bar{H}}}_{T}}=\mathbf{\bar{h}}_{T}^{H}{{\mathbf{\bar{h}}}_{T}},\text{~} \mathbf{w}=\mathrm{vec}\left( \mathbf{W} \right)$, and ${{\mathbf{\tilde{H}}}_{T}}={{\mathbf{I}}_{K+1}}\otimes {{\mathbf{\bar{H}}}_{T}}$.  
However, the constraint (\ref{radar_constraint_BF_v1}) is non-convex due to the convex left hand side (LHS) terms, which can be approximated as 
\begin{align}
    {{\mathbf{w}}^{H}}{{{\mathbf{\tilde{H}}}}_{T}}\mathbf{w}\ge & 2\Re \left\{ {{\left( {{\mathbf{w}}^{\left( t \right)}} \right)}^{H}}{{{\mathbf{\tilde{H}}}}_{T}}\mathbf{w} \right\}-{{\left( {{\mathbf{w}}^{\left( t \right)}} \right)}^{H}}{{{\mathbf{\tilde{H}}}}_{T}}{{\mathbf{w}}^{\left( t \right)}}, \nonumber \\ 
    {{\left| {{{\mathbf{\bar{h}}}}_{T}}\mathbf{z} \right|}^{2}}\ge & 2\Re \left\{ {{\left( \mathbf{z}^{\left( t \right)} \right)}^{H}}{{{\mathbf{\bar{H}}}}_{T}}\mathbf{z} \right\}-{{\left( \mathbf{z}^{\left( t \right)} \right)}^{H}}{{{\mathbf{\bar{H}}}}_{T}}\mathbf{z}^{\left( t \right)}, \label{radar_constr_BF1} 
\end{align}
where ${{\mathbf{w}}^{\left( t \right)}}$ and $\mathbf{z}^{\left( t \right)}$ represent solutions at iteration $t$. 
By substituting (\ref{radar_constr_BF1}) into (\ref{radar_constraint_BF_v1}), the radar constraint becomes 
\begin{equation} \label{radar_constr_BF_final}
    2\Re \left\{ {{\left( {{\mathbf{w}}^{\left( t \right)}} \right)}^{H}}{{{\mathbf{\tilde{H}}}}_{T}}\mathbf{w}+{{\left( \mathbf{z}^{\left( t \right)} \right)}^{H}}{{{\mathbf{\bar{H}}}}_{T}}\mathbf{z} \right\}\ge {{\bar{\Gamma }}_{r}},
\end{equation}
where ${{\bar{\Gamma }}_{r}}=\frac{{{\Gamma }_{r}}{{{\bar{\sigma }}}_{R}}}{{{\zeta }^{2}}}+{{\left( {{\mathbf{w}}^{\left( t \right)}} \right)}^{H}}{{\mathbf{\tilde{H}}}_{T}}{{\mathbf{w}}^{\left( t \right)}}+{{\left( \mathbf{z}^{\left( t \right)} \right)}^{H}}{{\mathbf{\bar{H}}}_{T}}\mathbf{z}^{\left( t \right)}$.

In the RIS power budget constraint, (\ref{BF1_c6}), the first and second terms can be re-expressed as $\left\| \mathbf{\Phi GW} \right\|_{F}^{2}={{\mathbf{w}}^{H}}{{\mathbf{G}}_{BR}}\mathbf{w}$ and ${{\zeta }^{2}}\left\| \mathbf{\Phi }{{\mathbf{H}}_{RT}}\mathbf{\Phi GW} \right\|_{F}^{2}={{\mathbf{w}}^{H}}{{\mathbf{\bar{G}}}_{RT}}\mathbf{w}$, respectively, where ${{\mathbf{G}}_{BR}}={{\mathbf{G}}^{H}}{{\mathbf{\Phi }}^{H}}\mathbf{\Phi G}\otimes {{\mathbf{I}}_{K+1}}$ and ${{\mathbf{\bar{G}}}_{RT}}={{\zeta }^{2}}\mathbf{G}_{RT}^{H}{{\mathbf{G}}_{RT}}\otimes {{\mathbf{I}}_{K+1}}$ with ${{\mathbf{G}}_{RT}}=\mathbf{\Phi }{{\mathbf{H}}_{RT}}\mathbf{\Phi G}$. Then, the constraint (\ref{BF1_c6}) can be written in convex form as  
\begin{equation}\label{RIS_buget_final}
    \resizebox{0.89\columnwidth}{!}{$
    {{\mathbf{w}}^{H}}{{\mathbf{H}}_{BRT}}\mathbf{w}+{{\zeta }^{2}}\left\| \mathbf{\Phi }{{\mathbf{H}}_{RT}}\mathbf{\Phi G}\mathbf{z} \right\|_{2}^{2} 
    +\left\| \mathbf{\Phi G}\mathbf{z} \right\|_{2}^{2}\le {{\overline{P}}^{\mathrm{RIS}}},
    $}
\end{equation}
where ${{\mathbf{H}}_{BRT}}={{\mathbf{G}}_{BR}}+{{\mathbf{\bar{G}}}_{RT}}$.

\subsubsection{EPSR reformulation in terms of $\left\{ \mathbf{W, z} \right\}$} \label{sec3-1-1}
The LHS of (\ref{BF1_c1}) is handled by employing the SCA method. First, $R_k$ is approximated around given points $\left\{ {{\mathbf{W}}^{\left( t \right)}},\mathbf{z}^{\left( t \right)} \right\}$ by exploiting the rate approximation result of \cite{nasir2016secrecy}.
Then, the achievable rate of the $k$-th private stream is  approximated as 
\begin{multline} \label{PrivRate_k}
    {{R}_{k}}={{\varepsilon }_{0,k}}+\frac{2\Re \left\{ \mathbf{w}_{k}^{H}\mathbf{\bar{H}}_{k}^{H}{{\pmb{\theta }}^{\left( t \right)}}{{\left( {{\pmb{\theta }}^{\left( t \right)}} \right)}^{H}}{{{\mathbf{\bar{H}}}}_{k}}\mathbf{w}_{k}^{\left( t \right)} \right\}}{\beta _{k}^{\left( t \right)}} \\
    -\frac{{{\left| \alpha _{k}^{\left( t \right)} \right|}^{2}}\left( \sum\limits_{i=1}^{K}{{{\left| {{\left( {{\pmb{\theta }}^{\left( t \right)}} \right)}^{H}}{{{\mathbf{\bar{H}}}}_{k}}{{\mathbf{w}}_{i}} \right|}^{2}}}+{{\left| {{\left( {{\pmb{\theta }}^{\left( t \right)}} \right)}^{H}}{{{\mathbf{\bar{H}}}}_{k}}\mathbf{z} \right|}^{2}} \right)}{\beta _{k}^{\left( t \right)}\left( \beta _{k}^{\left( t \right)}+{{\left| \alpha _{k}^{\left( t \right)} \right|}^{2}} \right)},
\end{multline}
where 
\begin{align*}
    {{\varepsilon }_{0,k}}=&\ln \left( 1+\frac{{{\left| \alpha _{k}^{\left( t \right)} \right|}^{2}}}{\beta _{k}^{\left( t \right)}} \right)-\frac{{{\left| \alpha _{k}^{\left( t \right)} \right|}^{2}}}{\beta _{k}^{\left( t \right)}} \nonumber \\
    &\quad -\frac{{{\left| \alpha _{k}^{\left( t \right)} \right|}^{2}}\left( 1+\sigma _{R}^{2}{{\left\| {{\left( {{\pmb{\theta }}^{\left( t \right)}} \right)}^{H}}\mathrm{diag}\left( \mathbf{\bar{h}}_{R,k}^{H} \right) \right\|}^{2}} \right)}{\beta _{k}^{\left( t \right)}\left( \beta _{k}^{\left( t \right)}+{{\left| \alpha _{k}^{\left( t \right)} \right|}^{2}} \right)}, \\
    \alpha _{k}^{\left( t \right)}=&{{\left( {{\pmb{\theta }}^{\left( t \right)}} \right)}^{H}}{{{\mathbf{\bar{H}}}}_{k}}\mathbf{w}_{k}^{\left( t \right)}, \\ 
    \beta _{k}^{\left( t \right)}=&\sum\limits_{i=1,i\ne k}^{K}{{{\left| {{\left( {{\pmb{\theta }}^{\left( t \right)}} \right)}^{H}}{{{\mathbf{\bar{H}}}}_{k}}\mathbf{w}_{i}^{\left( t \right)} \right|}^{2}}}+{{\left| {{\left( {{\pmb{\theta }}^{\left( t \right)}} \right)}^{H}}{{{\mathbf{\bar{H}}}}_{k}}\mathbf{z}^{\left( t \right)} \right|}^{2}}+{{\bar{\sigma }}_{R,k}}. 
\end{align*}

Next, the ergodic rate of private stream $k$ at the Eve can be rewritten through the following lemma. 
\begin{lem} \label{lem1}
    The achievable ergodic rate at the Eve while decoding private steam $k$ is denoted by
	\begin{multline} \label{ergo_PrivRate_Eve_final}
    -{{\mathbb{E}}_{{{\mathbf{h}}_{RE}}}}\left\{ {{\log }_{2}}\left( 1+{{\gamma }_{k,E}} \right) \right\}={{{\bar{\varepsilon }}}_{11,k}} -\frac{\pmb{\omega }_{E,k}^{H}\mathbf{\bar{Q}}_{E,k}^{\left( t \right)}{{\pmb{\omega }}_{E,k}}}{1-q_{E,k}^{\left( t \right)}} \\ 
    +\frac{2\Re \left\{ \bar{\sigma }_{E}^{-1}{{\left( \pmb{\omega }_{E,k}^{\left( t \right)} \right)}^{H}}\widehat{\pmb{\Omega }}_{E,k}^{H}{{\left( \mathbf{Q}_{E,k}^{\left( t \right)} \right)}^{-1}}{{\widehat{\pmb{\Omega }}}_{E,k}}{{\pmb{\omega }}_{E,k}} \right\}}{1-q_{E,k}^{\left( t \right)}} \\
    -\frac{\sum\limits_{i=0}^{K}{\mathbf{w}_{i}^{H}{{\widehat{\mathbf{G}}}_{E}}{{\mathbf{w}}_{i}}}+\mathbf{z}^{H}{{\widehat{\mathbf{G}}}_{E}}\mathbf{z}}{1+\ell _{E}^{\left( t \right)}},  
\end{multline}
where 
\begin{align*} 
    &{{\bar{\varepsilon }}_{11,k}}={{\varepsilon }_{11,k}}-\frac{{{\varepsilon }_{E,k}}}{1-q_{E,k}^{\left( t \right)}}-\frac{{{\bar{\sigma }}_{E}}}{1+\ell _{E}^{\left( t \right)}},\\
    &{{\varepsilon }_{E,k}}= \bar{\sigma }_{E}^{-1}{{\left( \pmb{\omega }_{E,k}^{\left( t \right)} \right)}^{H}}\widehat{\pmb{\Omega }}_{E,k}^{H}{{\left( \mathbf{Q}_{E,k}^{\left( t \right)} \right)}^{-1}}{{\left( \mathbf{Q}_{E,k}^{\left( t \right)} \right)}^{-1}}{{\widehat{\pmb{\Omega }}}_{E,k}}\pmb{\omega }_{E,k}^{\left( t \right)}, \\
    &{{\varepsilon }_{11,k}}={{\varepsilon }_{1,k}} -\frac{\Re \left\{ \bar{\sigma }_{E}^{-1}{{\left( \pmb{\omega }_{E,k}^{\left( t \right)} \right)}^{H}}\widehat{\pmb{\Omega }}_{E,k}^{H}{{\left( \mathbf{Q}_{E,k}^{\left( t \right)} \right)}^{-1}}{{\widehat{\pmb{\Omega }}}_{E,k}}\pmb{\omega }_{E,k}^{\left( t \right)} \right\}}{1-q_{E,k}^{\left( t \right)}}, \\
    &{{\varepsilon }_{1,k}}=1-\ln \left( \frac{1-q_{E,k}^{\left( t \right)}}{{{{\bar{\sigma }}}_{E}}} \right)-\ln \left( 1+\ell _{E}^{\left( t \right)} \right), \\
    &\ell _{E}^{\left( t \right)}=\sum\limits_{i=0}^{K}{{{\left( \mathbf{w}_{i}^{\left( t \right)} \right)}^{H}}{{\widehat{\mathbf{G}}}_{E}}\mathbf{w}_{i}^{\left( t \right)}}+{{\left( \mathbf{z}^{\left( t \right)} \right)}^{H}}{{\widehat{\mathbf{G}}}_{E}}\mathbf{z}^{\left( t \right)}+{{\bar{\sigma }}_{E}}-1, \\   
    &q_{E,k}^{\left( t \right)}=\bar{\sigma }_{E}^{-1}{{\left( \pmb{\omega }_{E,k}^{\left( t \right)} \right)}^{H}}\widehat{\pmb{\Omega }}_{E,k}^{H}{{\left( \mathbf{Q}_{E,k}^{\left( t \right)} \right)}^{-1}}{{\widehat{\pmb{\Omega }}}_{E,k}}\pmb{\omega }_{E,k}^{\left( t \right)}, \\
    &\mathbf{Q}_{E,k}^{\left( t \right)}={{\mathbf{I}}_{N\left( K+1 \right)}}+\bar{\sigma }_{E}^{-1}{{\widehat{\pmb{\Omega }}}_{E,k}}\pmb{\omega }_{E,k}^{\left( t \right)}{{\left( \pmb{\omega }_{E,k}^{\left( t \right)} \right)}^{H}}\widehat{\pmb{\Omega }}_{E}^{H}, \nonumber \\
    &\mathbf{\bar{Q}}_{E,k}^{\left( t \right)}=\bar{\sigma }_{E}^{-2}\widehat{\pmb{\Omega }}_{E,k}^{H}{{\left( \mathbf{Q}_{E,k}^{\left( t \right)} \right)}^{-1}}{{\widehat{\pmb{\Omega }}}_{E,k}}\pmb{\omega }_{E,k}^{\left( t \right)}{{\left( \pmb{\omega }_{E,k}^{\left( t \right)} \right)}^{H}}\widehat{\pmb{\Omega }}_{E,k}^{H} \nonumber \\
    & \quad \times {{\left( \mathbf{Q}_{E,k}^{\left( t \right)} \right)}^{-1}}{{\widehat{\pmb{\Omega }}}_{E,k}}
\end{align*}
\end{lem}
\begin{proof}
	See Appendix \ref{appendix2}.
\end{proof} 
Based on the results of Lemma \ref{lem1} and (\ref{PrivRate_k}), the EPSR constraint (\ref{BF1_c1}) can be formulated as 
\begin{equation} \label{priv_rate_k_final}
    \mathcal{F}^{\mathrm{p}}\left( {{\pmb{\omega }}_{E,k}} \right) + r_k \ge \tau,
\end{equation}
where $\mathcal{F}^{\mathrm{p}}\left( {{\pmb{\omega }}_{E,k}} \right)$ is given by (\ref{priv_sec_rate_BF_constraint}) at the top of next page, in which  
\begin{align*}
    {{\bar{\varepsilon }}_{E,k}}=&{{\varepsilon }_{0,k}}+{{\varepsilon }_{1,k}}-\frac{{{\bar{\sigma }}_{E}}}{1+\ell _{E}^{\left( t \right)}}-\frac{{{\varepsilon }_{E,k}}}{1-q_{E,k}^{\left( t \right)}} \\
    &-\frac{\Re \left\{ \bar{\sigma }_{E}^{-1}{{\left( \pmb{\omega }_{E,k}^{\left( t \right)} \right)}^{H}}\widehat{\pmb{\Omega }}_{E,k}^{H}{{\left( \mathbf{Q}_{E,k}^{\left( t \right)} \right)}^{-1}}{{\widehat{\pmb{\Omega }}}_{E,k}}\pmb{\omega }_{E,k}^{\left( t \right)} \right\}}{1-q_{E,k}^{\left( t \right)}}.
\end{align*}
        
\begin{floatEq}
	\begin{subequations}\begin{align} 
            & \mathcal{F}^{\mathrm{p}}\left( {{\pmb{\omega }}_{E,k}} \right) =  {{{\bar{\varepsilon }}}_{E,k}} +  \frac{2\Re \left\{ \mathbf{w}_{k}^{H}\mathbf{\bar{H}}_{k}^{H}{{\pmb{\theta }}^{\left( t \right)}}{{\left( {{\pmb{\theta }}^{\left( t \right)}} \right)}^{H}}{{{\mathbf{\bar{H}}}}_{k}}\mathbf{w}_{k}^{\left( t \right)} \right\}}{\beta _{k}^{\left( t \right)}}+\frac{2\Re \left\{ \bar{\sigma }_{E}^{-1}{{\left( \pmb{\omega }_{E,k}^{\left( t \right)} \right)}^{H}}\widehat{\pmb{\Omega }}_{E,k}^{H}{{\left( \mathbf{Q}_{E,k}^{\left( t \right)} \right)}^{-1}}{{\widehat{\pmb{\Omega }}}_{E,k}}{{\pmb{\omega }}_{E,k}} \right\}-\pmb{\omega }_{E,k}^{H}\mathbf{\bar{Q}}_{E,k}^{\left( t \right)}{{\pmb{\omega }}_{E,k}}}{1-q_{E,k}^{\left( t \right)}} \nonumber \\ 
            & \qquad  \qquad \quad -\frac{\sum\limits_{i=0}^{K}{\mathbf{w}_{i}^{H}{{\widehat{\mathbf{G}}}_{E}}{{\mathbf{w}}_{i}}}+\mathbf{z}^{H}{{\widehat{\mathbf{G}}}_{E}}\mathbf{z}}{1+\ell _{E}^{\left( t \right)}} -\frac{{{\left| \alpha _{k}^{\left( t \right)} \right|}^{2}}\left( \sum\limits_{i=1}^{K}{{{\left| {{\left( {{\pmb{\theta }}^{\left( t \right)}} \right)}^{H}}{{{\mathbf{\bar{H}}}}_{k}}{{\mathbf{w}}_{i}} \right|}^{2}}}+{{\left| {{\left( {{\pmb{\theta }}^{\left( t \right)}} \right)}^{H}}{{{\mathbf{\bar{H}}}}_{k}}\mathbf{z} \right|}^{2}} \right)}{\beta _{k}^{\left( t \right)}\left( \beta _{k}^{\left( t \right)}+{{\left| \alpha _{k}^{\left( t \right)} \right|}^{2}} \right)},  
                \label{priv_sec_rate_BF_constraint} \\
            & {{R}_{{{s}_{0}},k}}= {{\varepsilon }_{2,k}} 
            -\frac{{{\left| \alpha _{0k}^{\left( t \right)} \right|}^{2}}\left( {{\left| {{\left( {{\pmb{\theta }}^{\left( t \right)}} \right)}^{H}}{{{\mathbf{\bar{H}}}}_{k}}{{\mathbf{w}}_{0}} \right|}^{2}}+\sum\limits_{{k}'=1}^{K}{{{\left| {{\left( {{\pmb{\theta }}^{\left( t \right)}} \right)}^{H}}{{{\mathbf{\bar{H}}}}_{k}}{{\mathbf{w}}_{{{k}'}}} \right|}^{2}}}+{{\left| {{\left( {{\pmb{\theta }}^{\left( t \right)}} \right)}^{H}}{{{\mathbf{\bar{H}}}}_{k}}\mathbf{z} \right|}^{2}} \right)}{\beta _{0k}^{\left( t \right)}\left( \beta _{0k}^{\left( t \right)}+{{\left| \alpha _{0k}^{\left( t \right)} \right|}^{2}} \right)}  +\frac{2\Re \left\{ \mathbf{w}_{0}^{H}\mathbf{\bar{H}}_{k}^{H}{{\pmb{\theta }}^{\left( t \right)}}{{\left( {{\pmb{\theta }}^{\left( t \right)}} \right)}^{H}}{{{\mathbf{\bar{H}}}}_{k}}\mathbf{w}_{0}^{\left( t \right)} \right\}}{\beta _{0k}^{\left( t \right)}}. 
                \label{com_rate_UE}          
	\end{align}\end{subequations}
\end{floatEq}

\subsubsection{ECSR reformulation in terms of $\left\{ \mathbf{W, z} \right\}$} \label{sec3-1-2}
The rate difference ${{R}_{{{s}_{0}},k}}-{{\mathbb{E}}_{{{\mathbf{h}}_{RE}}}}\left\{ {{R}_{{{s}_{0}},E}} \right\}$ can be processed following the same procedure in section \ref{sec3-1-1}. Specifically, the decoding rate ${{R}_{{{s}_{0}},k}}$ is expressed by (\ref{com_rate_UE}) at the top of the next page, where 
\begin{align*}
    {{\varepsilon }_{2,k}}=&\ln \left( 1+\frac{{{\left| \alpha _{0k}^{\left( t \right)} \right|}^{2}}}{\beta _{0k}^{\left( t \right)}} \right)-\frac{{{\left| \alpha _{0k}^{\left( t \right)} \right|}^{2}}}{\beta _{0k}^{\left( t \right)}} \\
    &\quad -\frac{{{\left| \alpha _{0k}^{\left( t \right)} \right|}^{2}}\left( \sigma _{R}^{2}{{\left\| {{\left( {{\pmb{\theta }}^{\left( t \right)}} \right)}^{H}}\mathrm{diag}\left( \mathbf{\bar{h}}_{R,k}^{H} \right) \right\|}^{2}}+1 \right)}{\beta _{0k}^{\left( t \right)}\left( \beta _{0k}^{\left( t \right)}+{{\left| \alpha _{0k}^{\left( t \right)} \right|}^{2}} \right)},\\
    \alpha _{0k}^{\left( t \right)}=&{{\left( {{\pmb{\theta }}^{\left( t \right)}} \right)}^{H}}{{{\mathbf{\bar{H}}}}_{k}}\mathbf{w}_{0}^{\left( t \right)}, \\ 
    \beta _{0k}^{\left( t \right)}=&\sum\limits_{{k}'=1}^{K}{{{\left| {{\left( {{\pmb{\theta }}^{\left( t \right)}} \right)}^{H}}{{{\mathbf{\bar{H}}}}_{k}}\mathbf{w}_{{{k}'}}^{\left( t \right)} \right|}^{2}}}+{{\left| {{\left( {{\pmb{\theta }}^{\left( t \right)}} \right)}^{H}}{{{\mathbf{\bar{H}}}}_{k}}\mathbf{z}^{\left( t \right)} \right|}^{2}}+{{\bar{\sigma }}_{R,k}}.  
\end{align*}


Following the procedure of Lemma \ref{lem1}, the ergodic rate of the common stream at the Eve can be formulated as 
\begin{multline} \label{com_rate_Eve} 
  -{{\mathbb{E}}_{{{\mathbf{h}}_{RE}}}}\left\{ {{\log }_{2}}\left( 1+{{\gamma }_{{{s}_{0}},E}} \right) \right\}={{{\bar{\varepsilon }}}_{12,k}} -\frac{\pmb{\omega }_{0E,k}^{H}\mathbf{\bar{Q}}_{0E,k}^{\left( t \right)}{{\pmb{\omega }}_{0E,k}}}{1-q_{0E,k}^{\left( t \right)}} \\
  +\frac{2\Re \left\{ \bar{\sigma }_{E}^{-1}{{\left( \pmb{\omega }_{0E,k}^{\left( t \right)} \right)}^{H}}\widehat{\pmb{\Omega }}_{E,k}^{H}{{\left( \mathbf{Q}_{0E,k}^{\left( t \right)} \right)}^{-1}}{{\widehat{\pmb{\Omega }}}_{E,k}}{{\pmb{\omega }}_{0E,k}} \right\}}{1-q_{0E,k}^{\left( t \right)}} \\ 
  -\frac{\sum\limits_{i=0}^{K}{\mathbf{w}_{i}^{H}{{\widehat{\mathbf{G}}}_{E}}{{\mathbf{w}}_{i}}}+\mathbf{z}^{H}{{\widehat{\mathbf{G}}}_{E}}\mathbf{z}}{1+\ell _{E}^{\left( t \right)}},  
\end{multline}
where ${{\pmb{\omega }}_{0E,k}}$ is attained by removing only the first column from the matrix ${{\left[ \begin{matrix}
   \mathbf{w}_{0}^{T} & \mathbf{w}_{1}^{T} & \cdots  & \mathbf{w}_{K}^{T} & \mathbf{z}^{T}  \\
\end{matrix} \right]}^{T}}$. Moreover, we have 
\begin{align*}
    &{{\bar{\varepsilon }}_{12,k}}=1-\frac{\Re \left\{ \bar{\sigma }_{E}^{-1}{{\left( \pmb{\omega }_{0E,k}^{\left( t \right)} \right)}^{H}}\widehat{\pmb{\Omega }}_{E,k}^{H}{{\left( \mathbf{Q}_{0E,k}^{\left( t \right)} \right)}^{-1}}{{\widehat{\pmb{\Omega }}}_{E,k}}\pmb{\omega }_{0E,k}^{\left( t \right)} \right\}}{1-q_{0E,k}^{\left( t \right)}} \\
    &-\frac{{{\varepsilon }_{0E,k}}}{1-q_{0E,k}^{\left( t \right)}}-\frac{{{\bar{\sigma }}_{E}}}{1+\ell _{E}^{\left( t \right)}}-\ln \left( \frac{1-q_{0E,k}^{\left( t \right)}}{{{{\bar{\sigma }}}_{E}}} \right)-\ln \left( 1+\ell _{E}^{\left( t \right)} \right), \\
    &\mathbf{\bar{Q}}_{0E,k}^{\left( t \right)}=\bar{\sigma }_{E}^{-2}\widehat{\pmb{\Omega }}_{E,k}^{H} {{\left( \mathbf{Q}_{0E,k}^{\left( t \right)} \right)}^{-1}}{{\widehat{\pmb{\Omega }}}_{E,k}}\pmb{\omega }_{0E,k}^{\left( t \right)}{{\left( \pmb{\omega }_{0E,k}^{\left( t \right)} \right)}^{H}}\widehat{\pmb{\Omega }}_{E,k}^{H}\\
    & \quad \times {{\left( \mathbf{Q}_{0E,k}^{\left( t \right)} \right)}^{-1}}{{\widehat{\pmb{\Omega }}}_{E,k}},\\
    &{{\varepsilon }_{0E,k}}=\bar{\sigma }_{E}^{-1}{{\left( \pmb{\omega }_{0E,k}^{\left( t \right)} \right)}^{H}}\widehat{\pmb{\Omega }}_{E,k}^{H}{{\left( \mathbf{Q}_{0E,k}^{\left( t \right)} \right)}^{-1}} {{\left( \mathbf{Q}_{0E,k}^{\left( t \right)} \right)}^{-1}} {{\widehat{\pmb{\Omega }}}_{E,k}}\pmb{\omega }_{0E,k}^{\left( t \right)}, \\
    &q_{0E,k}^{\left( t \right)}=\bar{\sigma }_{E}^{-1}{{\left( \pmb{\omega }_{0E,k}^{\left( t \right)} \right)}^{H}}\widehat{\pmb{\Omega }}_{E,k}^{H}{{\left( \mathbf{Q}_{0E,k}^{\left( t \right)} \right)}^{-1}}{{\widehat{\pmb{\Omega }}}_{E,k}}\pmb{\omega }_{0E,k}^{\left( t \right)}.
\end{align*}
According to (\ref{com_rate_UE}) and (\ref{com_rate_Eve}), the ECSR constraint (\ref{BF1_c2}) can be reformulated as 
\begin{equation} \label{com_rate_k_final}
    {{\mathcal{F}}^{\text{c}}}\left( {{\pmb{\omega }}_{0E,k}} \right)\ge \sum\limits_{k=1}^{K}{{{r}_{k}}},
\end{equation}
where ${{\mathcal{F}}^{\text{c}}}\left( {{\pmb{\omega }}_{0E,k}} \right)$ is given in (\ref{com_rate_Eve_expression}) at the top of the next page, with 
\small
\begin{align*}
    &{{\bar{\varepsilon }}_{0E,k}}=1+{{\varepsilon }_{2,k}}-\ln \left( \frac{1-q_{0E,k}^{\left( t \right)}}{{{{\bar{\sigma }}}_{E}}} \right)-\ln \left( 1+\ell _{E}^{\left( t \right)} \right)-\frac{{{\bar{\sigma }}_{E}}}{1+\ell _{E}^{\left( t \right)}} \\
    &  -\frac{\Re \left\{ \bar{\sigma }_{E}^{-1}{{\left( \pmb{\omega }_{0E,k}^{\left( t \right)} \right)}^{H}}\widehat{\pmb{\Omega }}_{E,k}^{H}{{\left( \mathbf{Q}_{0E,k}^{\left( t \right)} \right)}^{-1}}{{\widehat{\pmb{\Omega }}}_{E,k}}\pmb{\omega }_{0E,k}^{\left( t \right)} \right\}}{1-q_{0E,k}^{\left( t \right)}}-\frac{{{\varepsilon }_{0E,k}}}{1-q_{0E,k}^{\left( t \right)}}.
\end{align*} 
\normalsize
\begin{floatEq}
	 \begin{align} 
         &{{\mathcal{F}}^{\text{c}}}\left( {{\pmb{\omega }}_{0E,k}} \right)= \frac{2\Re \left\{ \mathbf{w}_{0}^{H}\mathbf{\bar{H}}_{k}^{H}{{\pmb{\theta }}^{\left( t \right)}}{{\left( {{\pmb{\theta }}^{\left( t \right)}} \right)}^{H}}{{{\mathbf{\bar{H}}}}_{k}}\mathbf{w}_{0}^{\left( t \right)} \right\}}{\beta _{0k}^{\left( t \right)}}+\frac{2\Re \left\{ \bar{\sigma }_{E}^{-1}{{\left( \pmb{\omega }_{0E,k}^{\left( t \right)} \right)}^{H}}\widehat{\pmb{\Omega }}_{E,k}^{H}{{\left( \mathbf{Q}_{0E,k}^{\left( t \right)} \right)}^{-1}}{{\widehat{\pmb{\Omega }}}_{E,k}}{{\pmb{\omega }}_{0E,k}} \right\} - \pmb{\omega }_{0E,k}^{H}\mathbf{\bar{Q}}_{0E,k}^{\left( t \right)}{{\pmb{\omega }}_{0E,k}}}{1-q_{0E,k}^{\left( t \right)}} \nonumber \\ 
        & \quad -\frac{{{\left| \alpha _{0k}^{\left( t \right)} \right|}^{2}}\left( {{\left| {{\left( {{\pmb{\theta }}^{\left( t \right)}} \right)}^{H}}{{{\mathbf{\bar{H}}}}_{k}}{{\mathbf{w}}_{0}} \right|}^{2}}+\sum\limits_{{k}'=1}^{K}{{{\left| {{\left( {{\pmb{\theta }}^{\left( t \right)}} \right)}^{H}}{{{\mathbf{\bar{H}}}}_{k}}{{\mathbf{w}}_{{{k}'}}} \right|}^{2}}}+{{\left| {{\left( {{\pmb{\theta }}^{\left( t \right)}} \right)}^{H}}{{{\mathbf{\bar{H}}}}_{k}}\mathbf{z} \right|}^{2}} \right)}{\beta _{0k}^{\left( t \right)}\left( \beta _{0k}^{\left( t \right)}+{{\left| \alpha _{0k}^{\left( t \right)} \right|}^{2}} \right)} -\frac{\sum\limits_{i=0}^{K}{\mathbf{w}_{i}^{H}{{\widehat{\mathbf{G}}}_{E}}{{\mathbf{w}}_{i}}}+\mathbf{z}^{H}{{\widehat{\mathbf{G}}}_{E}}\mathbf{z}}{1+\ell _{E}^{\left( t \right)}} + {{{\bar{\varepsilon }}}_{0E,k}}. \label{com_rate_Eve_expression}
	\end{align} 
\end{floatEq} 

According to the above analysis, the beamforming and rate splitting problem can be recast as 
\begin{subequations}\label{BF1_opt_final}
	\begin{align}
		\underset{\textbf{r} \ge \textbf{0},{\textbf{z}},\textbf{W}}{\textrm{Maximize}}\hspace{0.15cm} & \tau, \label{BF_obj_final} \\
		\textrm{s.t.: } 
        & (\ref{priv_rate_k_final}), (\ref{com_rate_k_final}), (\ref{radar_constr_BF_final}), (\ref{BF1_c5}) \text{ and } (\ref{RIS_buget_final})  \label{BF_constraints}, 
		\end{align}
\end{subequations}
which represents a convex optimization problem and can be tackled using standard optimization tools such as CVX \cite{grant2014cvx}.

\subsection{Active reflection coefficients and common rate optimization} \label{sec3-2}
Given the variables $\mathbf{W}$, $\mathbf{z}$, $\mathbf{r}$, and $\mathbf{u}$, the optimization problem of the active RIS reflection matrix is given as  
\begin{subequations}\label{RIS1_opt}
	\begin{align}
		\underset{\textbf{r} \ge \textbf{0}, \pmb{\Phi }}
        {\textrm{Maximize}}\hspace{0.15cm} & \tau, \label{RIS1_obj} \\
		\textrm{s.t.: } 
            & {{r}_{k}}+ {{R}_{k}}-{{\mathbb{E}}_{{{\mathbf{h}}_{RE}}}}\left\{ {{R}_{k,E}} \right\} \ge \tau , \quad \forall k \label{RIS1_c1},  \\
	    & {{R}_{{{s}_{0}},k}}-{{\mathbb{E}}_{{{\mathbf{h}}_{RE}}}}\left\{                  {{R}_{{{s}_{0}},E}} \right\}\ge \sum\limits_{k=1}^{K}{{{r}_{k}}}, \quad \forall k \label{RIS1_c2},  \\
	    & \gamma \ge {{\Gamma }_{r}} \label{RIS1_c4}, \\
  	    & \beta_{n} \le \beta^{\mathrm{max}} \label{RIS1_c5}, \\
	    & \bar{\mathcal{P}}\left( \mathbf{\Phi }\right)\le {P}^{\mathrm{RIS}} \label{RIS1_c6}, 
	\end{align}
\end{subequations}
where $\bar{\mathcal{P}}\left( \mathbf{\Phi }\right) = \left\| \mathbf{\Phi GW} \right\|_{F}^{2} + {{\zeta }^{2}}\left\| \mathbf{\Phi }{{\mathbf{H}}_{RT}}\mathbf{\Phi GW} \right\|_{F}^{2} + {{\zeta }^{2}}\left\| \mathbf{\Phi }{{\mathbf{H}}_{RT}}\mathbf{\Phi G}\mathbf{z} \right\|_{2}^{2} + {{\zeta }^{2}}\sigma _{R}^{2}\left\| \mathbf{\Phi }{{\mathbf{H}}_{RT}}\mathbf{\Phi } \right\|_{F}^{2} + \left\| \mathbf{\Phi G}\mathbf{z} \right\|_{2}^{2} + 2\sigma _{R}^{2}\left\| \mathbf{\Phi } \right\|_{F}^{2} $.

The primary challenge with the constraint (\ref{RIS1_c4}) lies in managing the fractional form, which involves implicit functions of $\mathbf{\Phi }$, as outlined below.
\begin{equation}\label{radar_RIS}
    \gamma \triangleq \frac{f\left( \mathbf{\Phi } \right)}{q\left( \mathbf{\Phi } \right)}=\frac{{{\zeta }^{2}}{{\mathbf{u}}^{H}}{{\mathbf{H}}_{T}}\mathbf{\Pi }\mathbf{H}_{T}^{H}\mathbf{u}}{{{\mathbf{u}}^{H}}\left( {{\zeta }^{2}}\sigma _{R}^{2}{{\mathbf{H}}_{0}}\mathbf{H}_{0}^{H}+\sigma _{R}^{2}{{\mathbf{H}}_{1}}\mathbf{H}_{1}^{H}+{{\sigma }^{2}}{{\mathbf{I}}_{M}} \right)\mathbf{u}},
\end{equation}
where $\mathbf{\Pi }=\mathbf{W}{{\mathbf{W}}^{H}}+\mathbf{R}$ and   $\mathbf{R}=\mathbf{z}\mathbf{z}^{H}$. We attempt to re-express $f\left( \mathbf{\Phi } \right)$ and $q\left( \mathbf{\Phi } \right)$ as explicit functions of $\pmb{\theta }$. By exploiting the transformations $\mathbf{\Phi }{{\mathbf{h}}_{xy}}=\mathrm{diag}\left\{ {{\mathbf{h}}_{xy}} \right\}\pmb{\theta }$, and $\mathbf{h}_{xy}^{H}\mathbf{\Phi }={{\pmb{\theta }}^{T}}\mathrm{diag}\left\{ \mathbf{h}_{xy}^{*} \right\}$, the channels ${{\mathbf{H}}_{T}}$ and ${{\mathbf{H}}_{0}}$ are respectively denoted by ${{\mathbf{H}}_{T}}= {{{\mathbf{\tilde{G}}}}_{RT}}\pmb{\theta }{{\pmb{\theta }}^{T}}\mathbf{\tilde{G}}_{RT}^{H}$ and ${{\mathbf{H}}_{0}}={{{\mathbf{\tilde{G}}}}_{RT}}\pmb{\theta }{{\pmb{\theta }}^{T}}\mathrm{diag}\left\{ \mathbf{h}_{RT}^{*} \right\}$, where ${{\mathbf{\tilde{G}}}_{RT}}={{\mathbf{G}}^{H}}\mathrm{diag}\left\{ {{\mathbf{h}}_{RT}} \right\}$, and $\mathbf{\tilde{G}}_{RT}^{H}=\mathrm{diag}\left\{ \mathbf{h}_{RT}^{*} \right\}\mathbf{G}$. Thus, we can rewrite $f\left( \mathbf{\Phi } \right)$ and $q\left( \mathbf{\Phi } \right)$ as   
\begin{align} \label{radar_num_RIS}
   & f\left( \mathbf{\Phi } \right)={{\zeta }^{2}}\mathrm{Tr}\left( \underbrace{\mathbf{\tilde{G}}_{RT}^{H}\mathbf{u}{{\mathbf{u}}^{H}}{{{\mathbf{\tilde{G}}}}_{RT}}}_{\mathbf{K}}\underbrace{\pmb{\theta }{{\pmb{\theta }}^{T}}}_{\mathbf{L}}\underbrace{\mathbf{\tilde{G}}_{RT}^{H}\mathbf{\Pi }{{{\mathbf{\tilde{G}}}}_{RT}}}_{\mathbf{P}}\underbrace{\pmb{\theta }{{\pmb{\theta }}^{T}}}_{\mathbf{Q}} \right) \nonumber \\ 
 & \overset{\left( a \right)}{\mathop{=}}\,{{\zeta }^{2}}\mathrm{vec}^{H}\left( \pmb{\theta }{{\pmb{\theta }}^{T}} \right)\left[ \left( \mathbf{\tilde{G}}_{RT}^{T}{\mathbf{\Pi }^{T}}\mathbf{\tilde{G}}_{RT}^{*} \right)\otimes \left( \mathbf{\tilde{G}}_{RT}^{H}\mathbf{u}{{\mathbf{u}}^{H}}{{{\mathbf{\tilde{G}}}}_{RT}} \right) \right] \nonumber \\
 & \quad \times \mathrm{vec}\left( \pmb{\theta }{{\pmb{\theta }}^{T}} \right) = {{\mathbf{v}}^{H}}\mathbf{Av},  
\end{align}
where $\mathbf{A}={{\zeta }^{2}} \mathbf{\tilde{G}}_{RT}^{T}{\mathbf{\Pi }^{T}}\mathbf{\tilde{G}}_{RT}^{*} \otimes  \mathbf{\tilde{G}}_{RT}^{H}\mathbf{u}{{\mathbf{u}}^{H}}{{{\mathbf{\tilde{G}}}}_{RT}} $ and $\mathbf{v}=\mathrm{vec}\left( \pmb{\theta }{{\pmb{\theta }}^{T}} \right)=\pmb{\theta }\otimes \pmb{\theta }$. 
The equality $\left( a \right)$ follows the trace property $\mathrm{Tr}\left( \mathbf{KLPQ} \right)=\mathrm{vec}^{H}\left( {{\mathbf{Q}}^{H}} \right)\left( {{\mathbf{P}}^{T}}\otimes \mathbf{K} \right)\mathrm{vec}\left( \mathbf{L} \right)$. Similar to (\ref{radar_num_RIS}), $q\left( \mathbf{\Phi } \right)$ can be equivalently formulated as
\begin{equation} \label{radar_den_RIS}
    q\left( \mathbf{\Phi } \right)={{\mathbf{v}}^{H}}\mathbf{Bv}+{{\pmb{\theta }}^{H}}\mathbf{C} \pmb{\theta }+{{\sigma }^{2}}\left\| \mathbf{u} \right\|_{2}^{2},
\end{equation}
where $\mathbf{B}={{\zeta }^{2}}\sigma _{R}^{2} \mathrm{diag}\left\{ {{\mathbf{h}}_{RT}} \right\} \mathrm{diag}\left\{ \mathbf{h}_{RT}^{H} \right\} \otimes \mathbf{\tilde{G}}_{RT}^{H}\mathbf{u}{{\mathbf{u}}^{H}}{{{\mathbf{\tilde{G}}}}_{RT}} $ and $\mathbf{C}=\sigma _{R}^{2}\mathrm{diag}\left\{ \mathbf{Gu} \right\}\mathrm{diag}\left\{ {{\mathbf{u}}^{H}}{{\mathbf{G}}^{H}} \right\}$. By combining (\ref{radar_num_RIS}) and (\ref{radar_den_RIS}), the constraint (\ref{RIS1_c4}) can be re-expressed as
\begin{equation} \label{radar_RIS_v1}
    {{\mathbf{v}}^{H}}\mathbf{Dv}+{{\Gamma }_{r}}{{\pmb{\theta }}^{H}}\mathbf{C} \pmb{\theta }+{{c}_{0}}\le 0,
\end{equation}
where $\mathbf{D}={{\Gamma }_{r}}\mathbf{B}-\mathbf{A}$ and ${{c}_{0}}={{\Gamma }_{r}}{{\sigma }^{2}}\left\| \mathbf{u} \right\|_{2}^{2}$. However, the quartic term in (\ref{radar_RIS_v1}) is non-convex  w.r.t $\pmb{\theta }$. To address this, the MM approach is utilized to identify an upper-bound on ${{\mathbf{v}}^{H}}\mathbf{Dv}$ via the second-order Taylor expansion (STE) as follows 
\begin{multline} \label{1st_term_approx_radarRIS}
    {{\mathbf{v}}^{H}}\mathbf{Dv}\le {{\lambda }_{\mathbf{D}}}{{\mathbf{v}}^{H}}\mathbf{v}+2\Re \left\{ {{\mathbf{v}}^{H}}\left( \mathbf{D}-{{\lambda }_{\mathbf{D}}}{{\mathbf{I}}_{{{N}^{2}}}} \right){{\mathbf{v}}^{\left( t \right)}} \right\} \\
    +{{\left( {{\mathbf{v}}^{\left( t \right)}} \right)}^{H}}\left( {{\lambda }_{\mathbf{D}}}{{\mathbf{I}}_{{{N}^{2}}}}-\mathbf{D} \right){{\mathbf{v}}^{\left( t \right)}},
\end{multline}
where ${{\lambda }_{\mathbf{D}}}$ gives the maximum eigenvalue of $\mathbf{D}$. Given that ${{\beta }_{n}}\le {{\beta }^{\max }}$, the first term of (\ref{1st_term_approx_radarRIS}) is relaxed as  
${{\lambda }_{\mathbf{D}}}{{\mathbf{v}}^{H}}\mathbf{v}={{\lambda }_{\mathbf{D}}}\left( {{\pmb{\theta }}^{H}}\pmb{\theta } \right)\otimes \left( {{\pmb{\theta }}^{H}}\pmb{\theta } \right)  
    \le {{\lambda }_{\mathbf{D}}}{{N}^{2}}{{\left( {{\beta }^{\max }} \right)}^{4}}$.
Then, (\ref{1st_term_approx_radarRIS}) becomes
\begin{equation} \label{1st_term_approx_radarRIS_V2}
    {{\mathbf{v}}^{H}}\mathbf{Dv}=\Re \left\{ {{\mathbf{v}}^{H}}\mathbf{d} \right\}+{{c}_{1}}=\Re \left\{ {{\pmb{\theta }}^{H}}\mathbf{\ddot{D}}{{\pmb{\theta }}^{*}} \right\}+{{c}_{1r}},
\end{equation}
where $\mathbf{d}=2\left( \mathbf{D}-{{\lambda }_{\mathbf{D}}}{{\mathbf{I}}_{{{N}^{2}}}} \right){{\mathbf{v}}^{\left( t \right)}}, \mathbf{d}=\mathrm{vec}\left( {\mathbf{\ddot{D}}} \right)$ and ${{c}_{1r}}={{\left( {{\mathbf{v}}^{\left( t \right)}} \right)}^{H}}\left( {{\lambda }_{\mathbf{D}}}{{\mathbf{I}}_{{{N}^{2}}}}-\mathbf{D} \right){{\mathbf{v}}^{\left( t \right)}}+{{\lambda }_{\mathbf{D}}}{{N}^{2}}{{\left( {{\beta }^{\max }} \right)}^{4}}$. Since $\Re \left\{ {{\pmb{\theta }}^{H}}\mathbf{\ddot{D}}{{\pmb{\theta }}^{*}} \right\}$ is non-convex expression, it can be converted into real-valued quadratic form ${\pmb{\bar{\theta }}^{T}}\pmb{\bar{D}\bar{\theta }}$, where $\mathbf{\bar{D}}=\left[ \Re \left\{ {\mathbf{\ddot{D}}} \right\}, \Im \left\{ {\mathbf{\ddot{D}}} \right\} ;  \Im \left\{ {\mathbf{\ddot{D}}} \right\}, -\Re \left\{ {\mathbf{\ddot{D}}} \right\}  
 \right]$ and $\pmb{\bar{\theta }}={{\left[ \Re \left\{ {{\theta }^{T}} \right\},\Im \left\{ {{\theta }^{T}} \right\} \right]}^{T}}$. Again, we re-invoke the STE to precisely convexify ${{\pmb{\bar{\theta }}}^{T}}\pmb{\bar{D}\bar{\theta }}$ as
\begin{align} \label{1st_term_approx_radarRIS_V3}
  & {{{\pmb{\bar{\theta }}}}^{T}}\pmb{\bar{D}\bar{\theta }}\le  {{\left( {{{\pmb{\bar{\theta }}}}^{\left( t \right)}} \right)}^{T}}\mathbf{\bar{D}}{{{\pmb{\bar{\theta }}}}^{\left( t \right)}}+{{\left( {{{\pmb{\bar{\theta }}}}^{\left( t \right)}} \right)}^{T}}\left( \mathbf{\bar{D}}+{{{\mathbf{\bar{D}}}}^{T}} \right)\left( \pmb{\bar{\theta }}-{{{\pmb{\bar{\theta }}}}^{\left( t \right)}} \right) \nonumber \\
  & +\frac{{{\lambda }_{{\mathbf{\bar{D}}}}}}{2}{{\left( \pmb{\bar{\theta }}-{{{\pmb{\bar{\theta }}}}^{\left( t \right)}} \right)}^{T}}\left( \pmb{\bar{\theta }}-{{{\pmb{\bar{\theta }}}}^{\left( t \right)}} \right) =\frac{{{\lambda }_{{\mathbf{\bar{D}}}}}}{2}{{\pmb{\theta }}^{H}}\pmb{\theta }+\Re \left( {{\pmb{\theta }}^{H}}\mathbf{\tilde{d}} \right)+{{c}_{2r}},  
\end{align}
where $\mathbf{\tilde{d}}=\left[ \begin{matrix}
   {{\mathbf{I}}_{N}} & \sqrt{-1}  \\ \end{matrix}{{\mathbf{I}}_{N}} \right]\left( \mathbf{\bar{D}}+{{{\mathbf{\bar{D}}}}^{T}}-{{\lambda }_{{\mathbf{\bar{D}}}}}{{\mathbf{I}}_{2N}} \right){{\pmb{\bar{\theta }}}^{\left( t \right)}}$, ${{\lambda }_{{\mathbf{\bar{D}}}}}$ defines the maximum eigenvalue of $\mathbf{\bar{D}}+{{\mathbf{\bar{D}}}^{T}}$, and ${{c}_{2r}}=\frac{{{\lambda }_{{\mathbf{\bar{D}}}}}}{2}{{\left( {{{\pmb{\bar{\theta }}}}^{\left( t \right)}} \right)}^{T}}{{\pmb{\bar{\theta }}}^{\left( t \right)}}-{{\left( {{{\pmb{\bar{\theta }}}}^{\left( t \right)}} \right)}^{T}}{{\mathbf{\bar{D}}}^{T}}{{\pmb{\bar{\theta }}}^{\left( t \right)}}$. Therefore, (\ref{1st_term_approx_radarRIS_V2}) is transformed as  
\begin{equation} \label{approx2}
   {{\mathbf{v}}^{H}}\mathbf{Dv}\approx \frac{{{\lambda }_{{\mathbf{\bar{D}}}}}}{2}{{\pmb{\theta }}^{H}}\pmb{\theta }+\Re \left( {{\pmb{\theta }}^{H}}\mathbf{\tilde{d}} \right)+{{c}_{1}}, 
\end{equation}
where ${{c}_{1}}={{c}_{1r}}+{{c}_{2r}}$.  Substituting (\ref{approx2}) into (\ref{radar_RIS_v1}) leads to the following radar constraint 
\begin{equation} \label{radar_constr_RIS_final}
    \Re \left( {{\pmb{\theta }}^{H}}\mathbf{\tilde{d}} \right)+{{\pmb{\theta }}^{H}}\mathbf{\tilde{C}} \pmb{\theta }+{{\tilde{c}}_{0}}\le 0,
\end{equation}
where $\mathbf{\tilde{C}}={{\Gamma }_{r}}\mathbf{C}+\frac{{{\lambda }_{{\mathbf{\bar{D}}}}}}{2}{{\mathbf{I}}_{N}}$ and ${{\tilde{c}}_{0}}={{c}_{0}}+{{c}_{1}}$.

Regarding the active RIS power consumption constraint (\ref{RIS1_c6}), the first, second and third terms of the expression $\bar{\mathcal{P}}\left( \mathbf{\Phi }\right)$  can be rewritten as $\left\| \mathbf{\Phi GW} \right\|_{F}^{2}\ = {{\pmb{\theta }}^{H}}{{\mathbf{G}}_{C}}\pmb{\theta }$, $\left\| \mathbf{\Phi G}\mathbf{z} \right\|_{2}^{2}={{\pmb{\theta }}^{H}}{{\mathbf{G}}_{AN}}\pmb{\theta }$, and $\left\| \mathbf{\Phi } \right\|_{F}^{2}={{\pmb{\theta }}^{H}}{{\mathbf{I}}_{N}}\pmb{\theta }$, respectively, where ${{\mathbf{G}}_{AN}}=\mathrm{diag}\left( \mathbf{z}^{H}{{\mathbf{G}}^{H}} \right)\mathrm{diag}\left( \mathbf{G}\mathbf{z} \right)$ and ${{\mathbf{G}}_{C}}=\sum\limits_{i=0}^{K}{\mathrm{diag}\left( \mathbf{w}_{i}^{H}{{\mathbf{G}}^{H}} \right) \mathrm{diag}\left( \mathbf{G}{{\mathbf{w}}_{i}}  \right)}$. 
In addition, the remaining terms are transformed as: $\left\| \mathbf{\Phi }{{\mathbf{H}}_{RT}}\mathbf{\Phi G}\mathbf{z} \right\|_{2}^{2}={{\mathbf{v}}^{H}}{{\mathbf{J}}_{AN}}\mathbf{v}$, $\left\| \mathbf{\Phi }{{\mathbf{H}}_{RT}}\mathbf{\Phi } \right\|_{F}^{2}={{\mathbf{v}}^{H}}{{\mathbf{J}}_{R}}\mathbf{v}$, and $ \left\| \mathbf{\Phi }{{\mathbf{H}}_{RT}}\mathbf{\Phi GW} \right\|_{F}^{2}={{\mathbf{v}}^{H}}{{\mathbf{J}}_{RT}}\mathbf{v}$, in which ${{\mathbf{J}}_{RT}}=\mathrm{diag}\left( \mathbf{h}_{RT}^{T} \right)\mathrm{diag}\left( \mathbf{h}_{RT}^{*} \right) \otimes \mathbf{\tilde{G}}_{RT}^{H}\mathbf{W}{{\mathbf{W}}^{H}}{{\mathbf{\tilde{G}}}_{RT}}$,  ${{\mathbf{J}}_{AN}}=\mathrm{diag}\left( \mathbf{h}_{RT}^{T} \right)\mathrm{diag}\left( \mathbf{h}_{RT}^{*} \right) \otimes \mathbf{\tilde{G}}_{RT}^{H}\mathbf{R}{{\mathbf{\tilde{G}}}_{RT}}$, and ${{\mathbf{J}}_{R}}=\mathrm{diag}\left( \mathbf{h}_{RT}^{T} \right)\mathrm{diag}\left( \mathbf{h}_{RT}^{*} \right) \otimes \mathrm{diag}\left( \mathbf{h}_{RT}^{*} \right)\mathrm{diag}\left( \mathbf{h}_{RT}^{T} \right)$. As a result, the constraint (\ref{RIS1_c6}) can be formulated as 
\begin{equation} \label{RIS_pow_buget_v1}
    {{\mathbf{v}}^{H}}\mathbf{\tilde{J}v}+{{\pmb{\theta }}^{H}}\mathbf{\tilde{G}} \pmb{\theta }\le {{P}^{\mathrm{RIS}}},
\end{equation}
where $\mathbf{\tilde{J}}={{\zeta }^{2}}\left( {{\mathbf{J}}_{RT}}+{{\mathbf{J}}_{AN}}+\sigma _{R}^{2}{{\mathbf{J}}_{R}} \right)$, and $\mathbf{\tilde{G}}={{\mathbf{G}}_{C}}+{{\mathbf{G}}_{AN}}+2\sigma _{R}^{2}{{\mathbf{I}}_{N}}$. Similar to (\ref{1st_term_approx_radarRIS})--(\ref{1st_term_approx_radarRIS_V3}), the quartic term ${{\mathbf{v}}^{H}}\mathbf{\tilde{J}v}$ is approximated at the feasible point ${{\mathbf{v}}^{\left( t \right)}}$ via the STE. To this end, and after some algebraic manipulations, the constraint (\ref{RIS1_c6}) is denoted as 
\begin{equation} \label{RIS_pow_budget_final}
    {{\pmb{\theta }}^{H}}\overline{\overline{\mathbf{G}}}\pmb{\theta }+\Re \left( {{\pmb{\theta }}^{H}}\mathbf{\tilde{b}} \right)\le {{P}^{\mathrm{RIS}}}-{{c}_{2}}-{{c}_{3}},
\end{equation}
where $\overline{\overline{\mathbf{G}}}=\mathbf{\tilde{G}}+\frac{{{\lambda }_{{\mathbf{\bar{B}}}}}}{2}{{\mathbf{I}}_{N}}$, $ {{c}_{2}}={{\left( {{\mathbf{v}}^{\left( t \right)}} \right)}^{H}}\left( {{\lambda }_{{\mathbf{\tilde{J}}}}}{{\mathbf{I}}_{{{N}^{2}}}}-\mathbf{\tilde{J}} \right){{\mathbf{v}}^{\left( t \right)}}+{{\lambda }_{{\mathbf{\tilde{J}}}}}{{N}^{2}}{{\left( {{\beta }^{\max }} \right)}^{4}}$, ${{c}_{3}}=\frac{{{\lambda }_{{\mathbf{\bar{B}}}}}}{2}{{\left( {{{\pmb{\bar{\theta }}}}^{\left( t \right)}} \right)}^{T}}{{\pmb{\bar{\theta }}}^{\left( t \right)}}-{{\left( {{{\pmb{\bar{\theta }}}}^{\left( t \right)}} \right)}^{T}}{{\mathbf{\bar{B}}}^{T}}{{\pmb{\bar{\theta }}}^{\left( t \right)}}$, $\mathbf{\tilde{b}}=\left[{{\mathbf{I}}_{N}},  \sqrt{-1}{{\mathbf{I}}_{N}} \right]\left( \mathbf{\bar{B}}+{{{\mathbf{\bar{B}}}}^{T}}-{{\lambda }_{{\mathbf{\bar{B}}}}}{{\mathbf{I}}_{2N}} \right){{\pmb{\bar{\theta }}}^{\left( t \right)}}$. Moreover, ${\lambda }_{{\mathbf{\bar{B}}}}$ and ${\lambda }_{{\mathbf{\tilde{J}}}}$ define the maximum eigenvalue of the matrix $\mathbf{\bar{B}}+{{\mathbf{\bar{B}}}^{T}}$ and ${\mathbf{\tilde{J}}}$, respectively, where $\mathbf{\bar{B}}=\left[ 
   \Re \left\{ {\mathbf{\ddot{B}}} \right\}, \Im \left\{ {\mathbf{\ddot{B}}} \right\};
   \Im \left\{ {\mathbf{\ddot{B}}} \right\}, -\Re \left\{ {\mathbf{\ddot{B}}} \right\}  \right]$, and $\mathbf{\ddot{B}}$ represents the reconstructed matrix of the vector $\mathbf{b}$, i.e., $\mathbf{b}=\mathrm{vec}\left( {\mathbf{\ddot{B}}} \right)=2\left( \mathbf{\tilde{J}}-{{\lambda }_{{\mathbf{\tilde{J}}}}}{{\mathbf{I}}_{{{N}^{2}}}} \right){{\mathbf{v}}^{\left( t \right)}}$.

\subsubsection{EPSR reformulation in terms of $\pmb{\theta }$} \label{sec3-2-1}
Following (\ref{PrivRate_k}), the LHS of the constraint (\ref{RIS1_c1}) can be approximated around given points $\left\{ {{\mathbf{\Phi }}^{\left( t \right)}} \right\}$. Firstly, we rewrite the interference power on the private stream as  
\begin{multline}\label{privk_def1_RIS}
     \sum\limits_{i=1}^{K}{{{\left| {{\pmb{\theta }}^{H}}{{{\mathbf{\bar{H}}}}_{k}}\mathbf{w}_{i}^{\left( t \right)} \right|}^{2}}}+{{\left| {{\pmb{\theta }}^{H}}{{{\mathbf{\bar{H}}}}_{k}}\mathbf{z}^{\left( t \right)} \right|}^{2}}+\sigma _{R}^{2}{{\left\| {{\pmb{\theta }}^{H}}\mathrm{diag}\left( \mathbf{\bar{h}}_{R,k}^{H} \right) \right\|}^{2}} \\ 
    ={{\pmb{\theta }}^{H}}{{{\mathbf{\bar{H}}}}_{k}}\mathbf{\bar{A}\bar{H}}_{k}^{H}\pmb{\theta }+{{\pmb{\theta }}^{H}}{{{\mathbf{\bar{B}}}}_{k}}\pmb{\theta },  
\end{multline} 
where $\mathbf{\bar{A}}=\sum\limits_{i=1}^{K}{\mathbf{w}_{i}^{\left( t \right)}{{\left( \mathbf{w}_{i}^{\left( t \right)} \right)}^{H}}}+\mathbf{z}^{\left( t \right)}{{\left( \mathbf{z}^{\left( t \right)} \right)}^{H}}$ and ${{\mathbf{\bar{B}}}_{k}}=\sigma _{R}^{2}\mathrm{diag}\left( \mathbf{\bar{h}}_{R,k}^{H} \right)\mathrm{diag}\left( {{{\mathbf{\bar{h}}}}_{R,k}} \right)$.
Then, by employing the rate approximation of \cite{nasir2016secrecy}, the achievable rate of private stream $k$ can be represented as 
\begin{multline}  \label{priv_ratek_RIS}
    {{R}_{k}}={{\varepsilon }_{4,k}}+\frac{2\Re \left\{ {{\pmb{\theta }}^{H}}{{{\mathbf{\bar{H}}}}_{k}}\mathbf{w}_{k}^{\left( t \right)}{{\left( \mathbf{w}_{k}^{\left( t \right)} \right)}^{H}}\mathbf{\bar{H}}_{k}^{H}{{\pmb{\theta }}^{\left( t \right)}} \right\}}{\beta _{k}^{\left( t \right)}} \\
    -\frac{{{\left| \alpha _{k}^{\left( t \right)} \right|}^{2}}\left( {{\pmb{\theta }}^{H}} {{{\mathbf{\bar{H}}}}_{k}}\mathbf{\bar{A}\bar{H}}_{k}^{H}\pmb{\theta }+{{\pmb{\theta }}^{H}}{{{\mathbf{\bar{B}}}}_{k}}\pmb{\theta } \right)}{\beta _{k}^{\left( t \right)}\left( \beta _{k}^{\left( t \right)}+{{\left| \alpha _{k}^{\left( t \right)} \right|}^{2}} \right)},
\end{multline}
where 
\begin{align*}
    {{\varepsilon }_{4,k}}=&\ln \left( 1+\frac{{{\left| \alpha _{k}^{\left( t \right)} \right|}^{2}}}{\beta _{k}^{\left( t \right)}} \right)-\frac{{{\left| \alpha _{k}^{\left( t \right)} \right|}^{2}}}{\beta _{k}^{\left( t \right)}}-\frac{{{\left| \alpha _{k}^{\left( t \right)} \right|}^{2}}}{\beta _{k}^{\left( t \right)}\left( \beta _{k}^{\left( t \right)}+{{\left| \alpha _{k}^{\left( t \right)} \right|}^{2}} \right)}, \\
   \alpha _{k}^{\left( t \right)}=&{{\left( {{\pmb{\theta }}^{\left( t \right)}} \right)}^{H}}{{{\mathbf{\bar{H}}}}_{k}}\mathbf{w}_{k}^{\left( t \right)}, \\ 
    \beta _{k}^{\left( t \right)}=&\sum\limits_{i=1,i\ne k}^{K}{{{\left| {{\left( {{\pmb{\theta }}^{\left( t \right)}} \right)}^{H}}{{{\mathbf{\bar{H}}}}_{k}}\mathbf{w}_{i}^{\left( t \right)} \right|}^{2}}}+{{\left| {{\left( {{\pmb{\theta }}^{\left( t \right)}} \right)}^{H}}{{{\mathbf{\bar{H}}}}_{k}}\mathbf{z}^{\left( t \right)} \right|}^{2}}+{{\bar{\sigma }}_{R,k}}.  
\end{align*}
 
Next, we drive the closed form expression for the Eve's ergodic rate decoding private stream $k$ using the following lemma.  
\begin{lem} \label{lem2}
    The achievable ergodic rate of private steam $k$ at the Eve, in terms of $\pmb{\theta}$, can reformulated as
    \begin{multline} \label{ergo_PrivRate_Eve_finalRIS}
        -{{\mathbb{E}}_{{{\mathbf{h}}_{RE}}}}\left\{ {{\log }_{2}}\left( 1+{{\gamma }_{k, E}} \right) \right\}={{\varepsilon }_{13,Ek}}-\frac{{{\pmb{\theta }}^{H}}{{\mathbf{\Gamma }}_{E}}\pmb{\theta }}{1+\mu _{E}^{\left( t \right)}}-\frac{{{\pmb{\theta }}^{H}}{{\mathbf{C}}_{k}}\pmb{\theta }}{1-u_{E,k}^{\left( t \right)}}\\
        +\frac{2\Re \left\{ {{\left( {{\pmb{\theta }}^{\left( t \right)}} \right)}^{H}}{{\widehat{\mathbf{\Psi }}}_{E,k}}{{\left( \mathbf{\Xi }_{E,k}^{\left( t \right)} \right)}^{-1}}\widehat{\mathbf{\Psi }}_{E,k}^{H}\pmb{\theta } \right\}}{1-u_{E,k}^{\left( t \right)}},  
    \end{multline}
where 
\begin{align*}
    &{{\varepsilon }_{13,Ek}}=1-\ln \left( 1-u_{E,k}^{\left( t \right)} \right)-\ln \left( 1+\mu _{E}^{\left( t \right)} \right)-\frac{1}{1+\mu _{E}^{\left( t \right)}}\\
    &-\frac{{{\varepsilon }_{5,k}}}{1-u_{E,k}^{\left( t \right)}}-\frac{\Re \left\{ {{\left( {{\pmb{\theta }}^{\left( t \right)}} \right)}^{H}}{{\widehat{\mathbf{\Psi }}}_{E,k}}{{\left( \mathbf{\Xi }_{E,k}^{\left( t \right)} \right)}^{-1}}\widehat{\mathbf{\Psi }}_{E,k}^{H}{{\pmb{\theta }}^{\left( t \right)}} \right\}}{1-u_{E,k}^{\left( t \right)}},\\
    &u_{E,k}^{\left( t \right)}={{\left( {{\pmb{\theta }}^{\left( t \right)}} \right)}^{H}}{{\widehat{\mathbf{\Psi }}}_{E,k}}{{\left( \mathbf{\Xi }_{E,k}^{\left( t \right)} \right)}^{-1}}\widehat{\mathbf{\Psi }}_{E,k}^{H}{{\pmb{\theta }}^{\left( t \right)}}, \\ 
    &\mathbf{\Xi }_{E,k}^{\left( t \right)}={{\mathbf{I}}_{{\bar{N}}}}+\widehat{\mathbf{\Psi }}_{E,k}^{H}{{\pmb{\theta }}^{\left( t \right)}}{{\left( {{\pmb{\theta }}^{\left( t \right)}} \right)}^{H}}{{\widehat{\mathbf{\Psi }}}_{E,k}}, \text{~} \mu _{E}^{\left( t \right)}={{\left( {{\pmb{\theta }}^{\left( t \right)}} \right)}^{H}}{{\mathbf{\Gamma }}_{E}}{{\pmb{\theta }}^{\left( t \right)}}, \\
    &{{\varepsilon }_{5,k}}={{\left( {{\pmb{\theta }}^{\left( t \right)}} \right)}^{H}}{{\widehat{\mathbf{\Psi }}}_{E,k}}{{\left( \mathbf{\Xi }_{E,k}^{\left( t \right)} \right)}^{-1}}{{\left( \mathbf{\Xi }_{E,k}^{\left( t \right)} \right)}^{-1}}\widehat{\mathbf{\Psi }}_{E,k}^{H}{{\pmb{\theta }}^{\left( t \right)}}, \\
    &{{\mathbf{C}}_{k}}={{\widehat{\mathbf{\Psi }}}_{E,k}}{{\left( \mathbf{\Xi }_{E,k}^{\left( t \right)} \right)}^{-1}}\widehat{\mathbf{\Psi }}_{E,k}^{H}{{\pmb{\theta }}^{\left( t \right)}}{{\left( {{\pmb{\theta }}^{\left( t \right)}} \right)}^{H}}{{\widehat{\mathbf{\Psi }}}_{E,k}}{{\left( \mathbf{\Xi }_{E,k}^{\left( t \right)} \right)}^{-1}}\widehat{\mathbf{\Psi }}_{E,k}^{H},
\end{align*}
\end{lem}
\begin{proof}
	See Appendix \ref{appendix3}.
\end{proof} 

According to (\ref{priv_ratek_RIS}) and (\ref{ergo_PrivRate_Eve_finalRIS}), the constraint (\ref{RIS1_c1}) can be recast as 
\begin{equation} \label{priv_rate_RISconstraint_final}
    {{\bar{\varepsilon }}_{5,k}}+2\Re \left\{ {{\pmb{\theta }}^{H}}{{\mathbf{d}}_{E,k}} \right\}-{{\pmb{\theta }}^{H}}{{\mathbf{D}}_{E,k}}\pmb{\theta }\ge \tau -{{r}_{k}},
\end{equation}
where 
\begin{align*}
    {{\mathbf{d}}_{E,k}}=&\frac{{{{\mathbf{\bar{H}}}}_{k}}\mathbf{w}_{k}^{\left( t \right)}{{\left( \mathbf{w}_{k}^{\left( t \right)} \right)}^{H}}\mathbf{\bar{H}}_{k}^{H}{{\pmb{\theta }}^{\left( t \right)}}}{\beta _{k}^{\left( t \right)}}+\frac{{{\widehat{\mathbf{\Psi }}}_{E,k}}{{\left( \mathbf{\Xi }_{E,k}^{\left( t \right)} \right)}^{-1}}\widehat{\mathbf{\Psi }}_{E,k}^{H}{{\pmb{\theta }}^{\left( t \right)}}}{1-u_{E,k}^{\left( t \right)}}, \\ 
    {{\mathbf{D}}_{E,k}}=&\frac{{{\left| \alpha _{k}^{\left( t \right)} \right|}^{2}}\left( {{{\mathbf{\bar{H}}}}_{k}}\mathbf{\bar{A}\bar{H}}_{k}^{H}+{{{\mathbf{\bar{B}}}}_{k}} \right)}{\beta _{k}^{\left( t \right)}\left( \beta _{k}^{\left( t \right)}+{{\left| \alpha _{k}^{\left( t \right)} \right|}^{2}} \right)}+\frac{{{\mathbf{C}}_{k}}}{1-u_{E,k}^{\left( t \right)}}+\frac{{{\mathbf{\Gamma }}_{E}}}{1+\mu _{E}^{\left( t \right)}}, \\ 
    {{{\bar{\varepsilon }}}_{5,k}}=&1+{{\varepsilon }_{4,k}}-\ln \left( 1-u_{E,k}^{\left( t \right)} \right)-\ln \left( 1+\mu _{E}^{\left( t \right)} \right) -\frac{1}{1+\mu _{E}^{\left( t \right)}} \\
    &-\frac{\Re \left\{ {{\left( {{\pmb{\theta }}^{\left( t \right)}} \right)}^{H}}{{\widehat{\mathbf{\Psi }}}_{E,k}}{{\left( \mathbf{\Xi }_{E,k}^{\left( t \right)} \right)}^{-1}}\widehat{\mathbf{\Psi }}_{E,k}^{H}{{\pmb{\theta }}^{\left( t \right)}} \right\}+{\varepsilon }_{5,k}}{1-u_{E,k}^{\left( t \right)}}.
\end{align*}

\subsubsection{ECSR reformulation in terms of $\pmb{\theta }$} \label{sec3-2-2}
Similarly, the rate difference ${{R}_{{{s}_{0}}, k}}-{{\mathbb{E}}_{{{\mathbf{h}}_{RE}}}}\left\{ {{R}_{{{s}_{0}}, E}} \right\}$ can be analyzed as in the following. The achievable rate of the common stream at UE $k$ can be expressed as   
\begin{multline} \label{com_rate_userk_RIS}
    {{R}_{{{s}_{0}},k}}={{{\varepsilon }}_{6,k}}+\frac{2\Re \left\{ {{\pmb{\theta }}^{H}}{{{\mathbf{\bar{H}}}}_{k}}\mathbf{w}_{0}^{\left( t \right)}{{\left( \mathbf{w}_{0}^{\left( t \right)} \right)}^{H}}\mathbf{\bar{H}}_{k}^{H}{{\pmb{\theta }}^{\left( t \right)}} \right\}}{\beta _{0k}^{\left( t \right)}} \\
    -\frac{{{\left| \alpha _{0k}^{\left( t \right)} \right|}^{2}}\left( {{\pmb{\theta }}^{H}}{{{\mathbf{\bar{H}}}}_{k}}{{{\mathbf{\bar{A}}}}_{0}}\mathbf{\bar{H}}_{k}^{H}\pmb{\theta }+{{\pmb{\theta }}^{H}}{{{\mathbf{\bar{B}}}}_{k}}\pmb{\theta } \right)}{\beta _{0k}^{\left( t \right)}\left( \beta _{0k}^{\left( t \right)}+{{\left| \alpha _{0k}^{\left( t \right)} \right|}^{2}} \right)},
\end{multline}
where 
\begin{align*}
    {{{\varepsilon }}_{6,k}}=&\ln \left( 1+\frac{{{\left| \alpha _{0k}^{\left( t \right)} \right|}^{2}}}{\beta _{0k}^{\left( t \right)}} \right)-\frac{{{\left| \alpha _{0k}^{\left( t \right)} \right|}^{2}}}{\beta _{0k}^{\left( t \right)}}-\frac{{{\left| \alpha _{0k}^{\left( t \right)} \right|}^{2}}}{\beta _{0k}^{\left( t \right)}\left( \beta _{0k}^{\left( t \right)}+{{\left| \alpha _{0k}^{\left( t \right)} \right|}^{2}} \right)}, \\
    \beta _{0k}^{\left( t \right)}=&\sum\limits_{{k}'=1}^{K}{{{\left| {{\left( {{\pmb{\theta }}^{\left( t \right)}} \right)}^{H}}{{{\mathbf{\bar{H}}}}_{k}}\mathbf{w}_{{{k}'}}^{\left( t \right)} \right|}^{2}}}+{{\left| {{\left( {{\pmb{\theta }}^{\left( t \right)}} \right)}^{H}}{{{\mathbf{\bar{H}}}}_{k}}\mathbf{z}^{\left( t \right)} \right|}^{2}} + {{\bar{\sigma }}_{R,k}}, \\
    \alpha _{0k}^{\left( t \right)}=&{{\left( {{\pmb{\theta }}^{\left( t \right)}} \right)}^{H}}{{{\mathbf{\bar{H}}}}_{k}}\mathbf{w}_{0}^{\left( t \right)}, \text{~}
    {{\mathbf{\bar{A}}}_{0}}=\sum\limits_{k=0}^{K}{\mathbf{w}_{k}^{\left( t \right)}{{\left( \mathbf{w}_{k}^{\left( t \right)} \right)}^{H}}}+\mathbf{z}^{\left( t \right)}{{\left( \mathbf{z}^{\left( t \right)} \right)}^{H}}.
\end{align*}


By adopting the same analysis of Lemma \ref{lem2}, the ergodic rate of the common stream at the Eve can be formulated as 
\begin{multline} \label{Eve_comRate_RIS}
    -{{\mathbb{E}}_{{{\mathbf{h}}_{RE}}}}\left\{ {{\log }_{2}}\left( 1+{{\gamma }_{{{s}_{0}}, E}} \right) \right\}={{\varepsilon }_{14,Ek}}-\frac{{{\pmb{\theta }}^{H}}{{\mathbf{P}}_{0E}}\pmb{\theta }}{1-u_{0E}^{\left( t \right)}}-\frac{{{\pmb{\theta }}^{H}}{{\mathbf{\Gamma }}_{E}}\pmb{\theta }}{1+\mu _{E}^{\left( t \right)}}\\
    +\frac{2\Re \left\{ {{\left( {{\pmb{\theta }}^{\left( t \right)}} \right)}^{H}}{{\widehat{\mathbf{\Psi }}}_{0E}}{{\left( \mathbf{\Xi }_{0E}^{\left( t \right)} \right)}^{-1}}\widehat{\mathbf{\Psi }}_{0E}^{H}\pmb{\theta } \right\}}{1-u_{0E}^{\left( t \right)}},
\end{multline}
where the block matrix ${{\widehat{\mathbf{\Psi }}}_{0E}}=\sigma _{E}^{-1}\left[{{{\pmb{\overset{\scriptscriptstyle\frown}{E}}}}_{E}},  {{{\pmb{\overset{\scriptscriptstyle\frown}{F}}}}_{0E}},  {{{\pmb{\overset{\scriptscriptstyle\frown}{Z}}}}_{E,\mathrm{AN}}} \right]\in {{\mathbb{C}}^{N\times \tilde{N}}}$, with $\tilde{N}={{N}^{2}}+N+NK$. Moreover, we have ${{\pmb{\overset{\scriptscriptstyle\frown}{F}}}_{0E}}=\left[{{\mathbf{\tilde{D}}}_{E,1}} \mathbf{G\tilde{W}}_{0}^{\left( t \right)}, \cdots, {{\mathbf{\tilde{D}}}_{E,N}} \mathbf{G\tilde{W}}_{0}^{\left( t \right)}   \right]\in {{\mathbb{C}}^{N\times NK}}$, $\mathbf{\tilde{W}}_{0}^{\left( t \right)}=\left[\mathbf{w}_{1}^{\left( t \right)}, \cdots, \mathbf{w}_{K}^{\left( t \right)}  \right]$, and ${{\pmb{\overset{\scriptscriptstyle\frown}{Z}}}_{E,\mathrm{AN}}}=\left[{{\mathbf{\tilde{D}}}_{E,1}}\mathbf{Gz}^{\left( t \right)}, \cdots, {{\mathbf{\tilde{D}}}_{E,N}}\mathbf{Gz}^{\left( t \right)} \right]\in {{\mathbb{C}}^{N\times N}}$. Furthermore, 
\begin{align*} 
    {{\varepsilon }_{14,Ek}}=&1-\ln \left( 1-u_{0E}^{\left( t \right)} \right)-\ln \left( 1+\mu _{E}^{\left( t \right)} \right)-\frac{1}{1+\mu _{E}^{\left( t \right)}} \\
    &-\frac{\Re \left\{ {{\left( {{\pmb{\theta }}^{\left( t \right)}} \right)}^{H}}{{\widehat{\mathbf{\Psi }}}_{0E}}{{\left( \mathbf{\Xi }_{0E}^{\left( t \right)} \right)}^{-1}}\widehat{\mathbf{\Psi }}_{0E}^{H}{{\pmb{\theta }}^{\left( t \right)}} \right\}+{{\varepsilon }_{8}}}{1-u_{0E}^{\left( t \right)}}, \\
    {{\varepsilon }_{8}}=&{{\left( {{\pmb{\theta }}^{\left( t \right)}} \right)}^{H}}{{\widehat{\mathbf{\Psi }}}_{0E}}{{\left( \mathbf{\Xi }_{0E}^{\left( t \right)} \right)}^{-1}}{{\left( \mathbf{\Xi }_{0E}^{\left( t \right)} \right)}^{-1}}\widehat{\mathbf{\Psi }}_{0E}^{H}{{\pmb{\theta }}^{\left( t \right)}}, \\
    u_{0E}^{\left( t \right)}=&{{\left( {{\pmb{\theta }}^{\left( t \right)}} \right)}^{H}}{{\widehat{\mathbf{\Psi }}}_{0E}}{{\left( \mathbf{\Xi }_{0E}^{\left( t \right)} \right)}^{-1}}\widehat{\mathbf{\Psi }}_{0E}^{H}{{\pmb{\theta }}^{\left( t \right)}}, \\
    \mathbf{\Xi }_{0E}^{\left( t \right)}=&{{\mathbf{I}}_{{\tilde{N}}}}+\widehat{\mathbf{\Psi }}_{0E}^{H}{{\pmb{\theta }}^{\left( t \right)}}{{\left( {{\pmb{\theta }}^{\left( t \right)}} \right)}^{H}}{{\widehat{\mathbf{\Psi }}}_{0E}}, \\
    {{\mathbf{P}}_{0E}}=&{{\widehat{\mathbf{\Psi }}}_{0E}}{{\left( \mathbf{\Xi }_{0E}^{\left( t \right)} \right)}^{-1}}\widehat{\mathbf{\Psi }}_{0E}^{H}{{\pmb{\theta }}^{\left( t \right)}}{{\left( {{\pmb{\theta }}^{\left( t \right)}} \right)}^{H}}{{\widehat{\mathbf{\Psi }}}_{0E}}{{\left( \mathbf{\Xi }_{0E}^{\left( t \right)} \right)}^{-1}}\widehat{\mathbf{\Psi }}_{0E}^{H}.
\end{align*}

Then, we combine the results of (\ref{com_rate_userk_RIS}) and (\ref{Eve_comRate_RIS}) to re-express the constraint (\ref{RIS1_c2}) as follows 
\begin{equation} \label{common_secRate_RISConstraint_final}
    {{\varepsilon }_{7,k}}-\frac{{{\varepsilon }_{8}}}{1-u_{0E}^{\left( t \right)}}+2\Re \left\{ {{\pmb{\theta }}^{H}}{{\pmb{\upsilon }}_{0E,k}} \right\}-{{\pmb{\theta }}^{H}}{{\mathbf{V}}_{0E,k}}\pmb{\theta }\ge \sum\limits_{k=1}^{K}{{{r}_{k}}},
\end{equation}
where 
\begin{align*}
    {{\pmb{\upsilon }}_{0E,k}}=&\frac{{{{\mathbf{\bar{H}}}}_{k}}\mathbf{w}_{0}^{\left( t \right)}{{\left( \mathbf{w}_{0}^{\left( t \right)} \right)}^{H}}\mathbf{\bar{H}}_{k}^{H}{{\pmb{\theta }}^{\left( t \right)}}}{\beta _{0k}^{\left( t \right)}}+\frac{{{\widehat{\mathbf{\Psi }}}_{0E}}{{\left( \mathbf{\Xi }_{0E}^{\left( t \right)} \right)}^{-1}}\widehat{\mathbf{\Psi }}_{0E}^{H}{{\pmb{\theta }}^{\left( t \right)}}}{1-u_{0E}^{\left( t \right)}}, \\ 
    {{\mathbf{V}}_{0E,k}}=&\frac{{{\mathbf{\Gamma }}_{E}}}{1+\mu _{E}^{\left( t \right)}}+\frac{{{\left| \alpha _{0k}^{\left( t \right)} \right|}^{2}}\left( {{{\mathbf{\bar{H}}}}_{k}}{{{\mathbf{\bar{A}}}}_{0}}\mathbf{\bar{H}}_{k}^{H}+{{\mathbf{B}}_{k}} \right)}{\beta _{0k}^{\left( t \right)}\left( \beta _{0k}^{\left( t \right)}+{{\left| \alpha _{0k}^{\left( t \right)} \right|}^{2}} \right)}+\frac{{{\mathbf{P}}_{0E}}}{1-u_{0E}^{\left( t \right)}}, \\
    {{\varepsilon }_{7,k}}=&{{\varepsilon }_{6,k}}-\ln \left( 1-u_{0E}^{\left( t \right)} \right) -\ln \left( 1+\mu _{E}^{\left( t \right)} \right)-\frac{1}{1+\mu _{E}^{\left( t \right)}} \\
    &-\frac{\Re \left\{ {{\left( {{\pmb{\theta }}^{\left( t \right)}} \right)}^{H}}{{\widehat{\mathbf{\Psi }}}_{0E}}{{\left( \mathbf{\Xi }_{0E}^{\left( t \right)} \right)}^{-1}}\widehat{\mathbf{\Psi }}_{0E}^{H}{{\pmb{\theta }}^{\left( t \right)}} \right\}}{1-u_{0E}^{\left( t \right)}} + 1.
\end{align*}

Based on the preceding analysis, the active RIS and rate splitting sub-problem can be reformulated as 
\begin{subequations}\label{ActRIS_opt_final}
	\begin{align}
		\underset{\textbf{r} \ge \textbf{0},\pmb{\theta}}{\textrm{Maximize}}\hspace{0.15cm} & \tau, \label{RIS_obj_final} \\
		\textrm{s.t.: } 
        & (\ref{priv_rate_RISconstraint_final}), (\ref{common_secRate_RISConstraint_final}),  (\ref{radar_constr_RIS_final}), (\ref{RIS1_c5}), \text{ and } (\ref{RIS_pow_budget_final})   \label{RIS_constraints}. 
		\end{align}
\end{subequations}
The problem (\ref{ActRIS_opt_final}) is in convex form. Hence, it can be solved using standard optimization tools such as CVX \cite{grant2014cvx}.

\subsection{Radar receive beamforming optimization} \label{sec3-3}
Finally, we optimize the radar receive beamforming given that $\mathbf{z}$, $\mathbf{W}$, $\mathbf{\Phi }$, and $\mathbf{r}$ are fixed. Since the radar SNR, $\gamma$, only depends on $\mathbf{u}$, it can be optimized by maximizing the radar SNR as follows. 
\begin{equation}\label{RadarRX_BF_opt_final}
		\underset{\textbf{u}}{\textrm{Maximize}}\hspace{0.15cm}  \gamma =\frac{{{\mathbf{u}}^{H}}{{{\mathbf{\tilde{H}}}}_{T}}\mathbf{u}}{{{\mathbf{u}}^{H}}{{{\mathbf{\tilde{H}}}}_{0}}\mathbf{u}}, \quad \textrm{s.t.: }  \left\| \mathbf{u} \right\|=1, 
\end{equation}
where ${{\mathbf{\tilde{H}}}_{T}}={{\zeta }^{2}}{{\mathbf{H}}_{T}}\mathbf{\Pi }\mathbf{H}_{T}^{H}$ and ${{\mathbf{\tilde{H}}}_{0}}={{\zeta }^{2}}\sigma _{R}^{2}{{\mathbf{H}}_{0}}\mathbf{H}_{0}^{H}+\sigma _{R}^{2}{{\mathbf{H}}_{1}}\mathbf{H}_{1}^{H}+{{\sigma }^{2}}{{\mathbf{I}}_{M}}$. Clearly, the  problem (\ref{RadarRX_BF_opt_final}) is classified as a generalized Rayleigh quotient problem that has a closed-form solution. This solution can be found by identifying the eigenvectors corresponding to the maximum eigenvalues of $\mathbf{\tilde{H}}_{0}^{-1}{{\mathbf{\tilde{H}}}_{T}}$.

\subsection{Overall Optimization and computational complexity} \label{sec3-4}
Finally, we have transformed (\ref{opt_problem_orginal}) into a tractable problem, allowing the three subproblems to be alternatively optimized.
The complexity of the MM when applied for the problem (\ref{RIS1_opt}) is approximately $\mathcal{O}\left(T_{\mathrm{MM}}N^2 + N^3\right)$, where $T_{\mathrm{MM}}$ refers to the number of iterations for the  convergence of RIS phase shift problem  \cite{pan2020intelligent}. There, the overall complexity is given by $\mathcal{O}\left(T_{0}\left(M^2(K+1)^2N + M(K+1)^2N^2 + T_{\mathrm{MM}}N^2 + N^3\right)\right)$, where $T_{0}$ symbolizes the number of iterations for algorithm convergence \cite{hao2022securing}.

\begin{algorithm}[t]
	\caption{Joint optimization algorithm for active RIS-enabled secure RSMA-ISAC}
	\label{alg1}
	\begin{algorithmic}[1]
		\STATE{Set the initial values $\mathbf{W}^{\left( 0 \right)}$, $\mathbf{u}^{\left( 0 \right)}$ $\mathbf{z}^{\left( 0 \right)}$, $\pmb{\theta}^{\left( 0 \right)}$, and iteration index $t=1$.} 
		\REPEAT
		\STATE{Given $\pmb{\theta}^{\left( t-1 \right)}$, find $\mathbf{W}^{(t)}$, $\mathbf{z}^{\left( t \right)}$ $\mathbf{r}^{(t,1)}$ via solving (\ref{BF1_opt_final}).}
		\STATE{Optimize $\pmb{\theta}^{\left( t \right)}$, $\mathbf{r}^{(t,2)}$ via solving (\ref{ActRIS_opt_final}).}
		\STATE{Find $\mathbf{u}^{\left( t \right)}$ using (\ref{RadarRX_BF_opt_final})}
		\STATE{Update $ t=t+1 $.} 
		\UNTIL{convergence.}
	\end{algorithmic} 
\end{algorithm}

\section{Simulation Results and Discussion}\label{sec4}
In this section, we present the results of the proposed active RIS-enabled RSMA (ARIS-RSMA) for robust secure ISAC operation, where the performance is averaged over 1000 channel realizations. In the proposed system, one ISAC-BS with $M=8$ antennas and an active RIS of $N=16$ elements are employed to secure $K = 3$ users against an unknown Eve. The locations of the BS and RIS are set as $(0, 0)$, $(50, 0)$ in meters (m), respectively. The radar target is positioned $5$m away from the RIS with relative azimuth angle of ${\pi }/{4}$. The communication users are assumed uniformly distributed within a disk of radius $5$ m that is centered at $(40, 20)$ m. Moreover, the possible Eve locations are determined by ${{\mathcal{R}}_{E}}=\left\{ {d}_{RE} = [30, 35] \text{ m}, {\theta }_{E} = [{\pi }/{6}, {\pi }/{3}] \right\}$ w.r.t the RIS viewpoint. For the trapezoidal integration of (\ref{eve_esti_CH_ver3}), we set $N_{\theta_{g}}=500$. The Rician factor and reference distance path-loss are assumed  $\kappa = 3$ dB and $\ell=-30$ dB, respectively \cite{zuo2023exploiting}. In addition, the path-loss exponents are denoted as follows: BS-RIS $\alpha_{BR} = 2$, RIS-UE $\alpha_{RU} = 2.2$, RIS-target $\alpha_{RT} = 2$, RIS-Eve $\alpha_{RE} = 2.2$. The noise variance is set as $\sigma ^{2} = -120$ dBm at the BS, and $\sigma _{k}^{2} = \sigma _{E}^{2} = \sigma _{R}^{2} = -80$ dBm for UEs, EVe, and RIS. Additionally, the SNR requirement for target sensing is set as $\gamma_r=1$ dB. The active RIS power budget is fixed at $20$ dBm. The RCS parameter is set as ${\zeta }^{2}=1$.

\begin{figure}[t]
\centering{\includegraphics[width=\columnwidth]{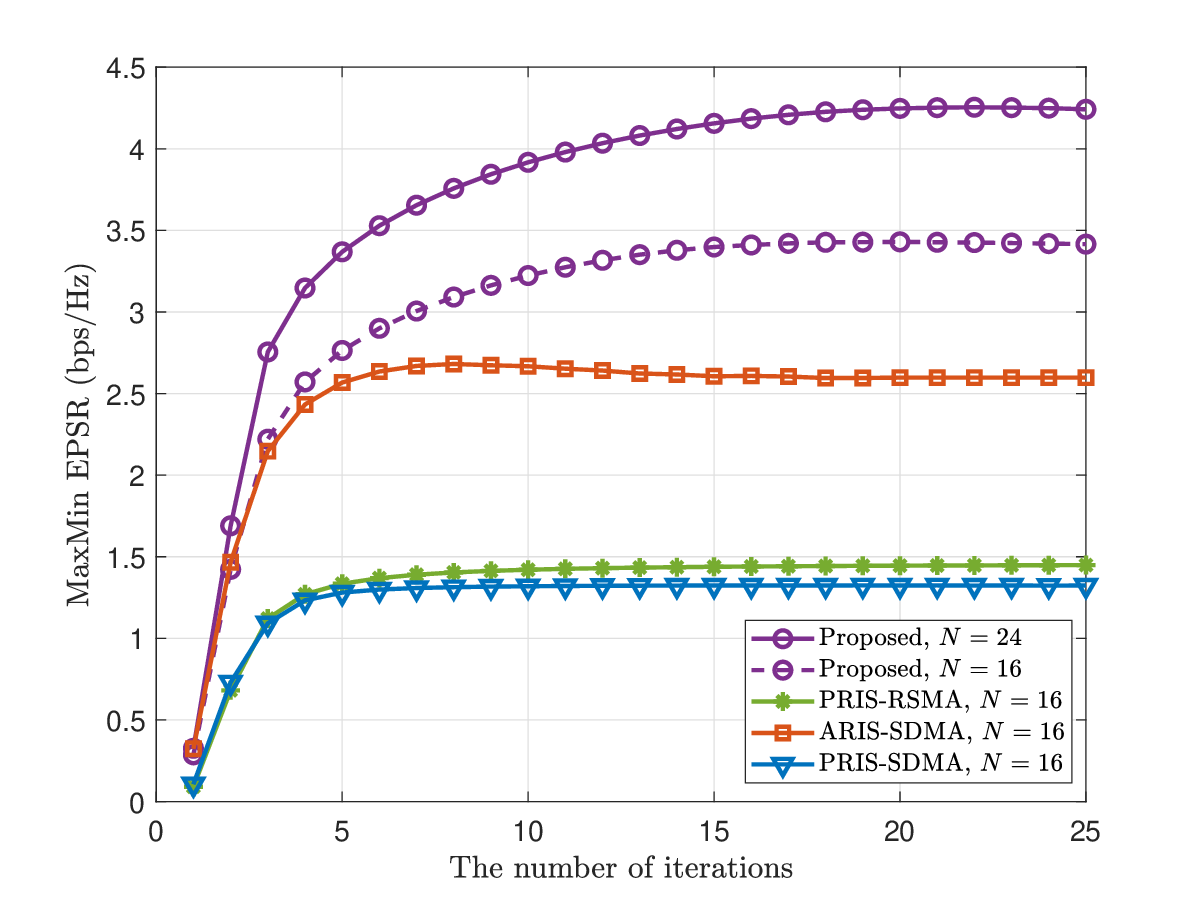}}
\caption{Convergence of the proposed solution versus the benchmarks.}\label{converg}
\end{figure}
Firstly, the MaxMin EPSR convergence of the proposed Algorithm \ref{alg1} is compared with the benchmark schemes, PRIS-RSMA, ARIS-SDMA, and PRIS-SDMA, in Fig.~\ref{converg}. It can be seen that the proposed scheme exhibits superior performance compared to the ARIS-SDMA by $30.77$\%. Additionally, it surpasses the passive RIS-based scenarios, i.e., the PRIS-RSMA and PRIS-SDMA, by $142.85$\% and $161.5$\%, respectively. However, the active RIS-based systems require $15-20$ iterations more than that needed by the passive RIS-based schemes due to the additional power budget constraint imposed by the active RIS. Furthermore, the PRIS-RSMA and PRIS-SDMA show a comparable performance, whereas a notable performance gap exists between the proposed ARIS-RSMA and the ARIS-SDMA. A possible reason is that the amplification capability of the active RIS not only boosts the ECSR and AN impact on the Eve, but it also potentially competes with the multiplicative fading effect of the passive RIS.

\begin{figure}[t]
\centering{\includegraphics[width=\columnwidth]{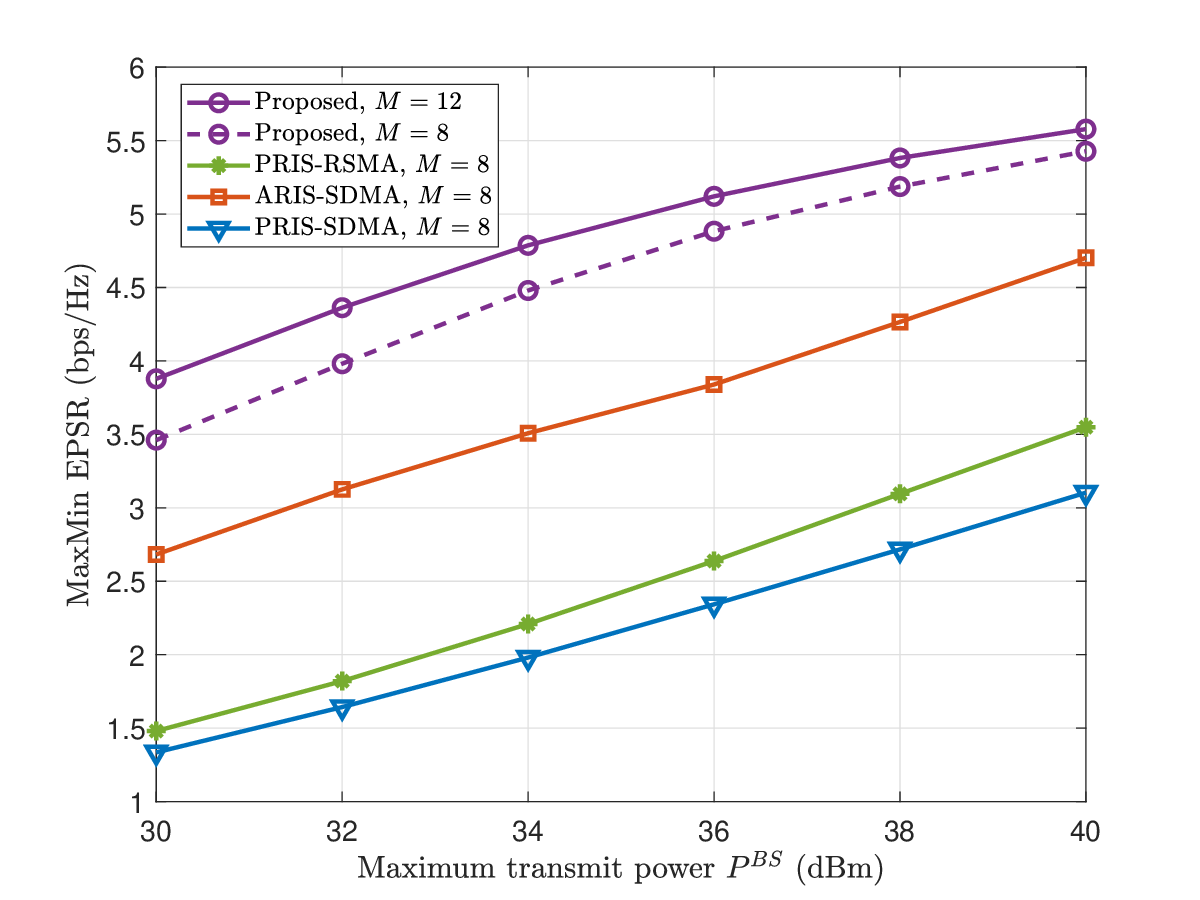}}
\caption{The achievable MaxMin EPSR versus maximum transmission power at BS.}\label{TX_power_SR}
\end{figure}
Next, in Fig.~\ref{TX_power_SR}, we examine the impact of the BS budget on the MaxMin EPSR performance. The performance improves with $P^{BS}$ across all the schemes. When $P^{BS}$ is at $36$ dB, the proposed ARIS-RSMA  exhibits a gain of $28.42$\% compared to ARIS-SDMA. Moreover, it is observed that the performance of the proposed scheme is further improved by increasing the number of transmitting antennas to 12. On the other hand, the performance gap between PRIS-RSMA and PRIS-SDMA exhibits only a limited increase from $11.28$\% to $14.2$\% when $P^{BS}$ changes from $30$ dBm to $40$ dBm, where the BS beamformers play the main role in substituting the multiplicative path-loss of the passive RIS. 

\begin{figure}[t]
\centering{\includegraphics[width=\columnwidth]{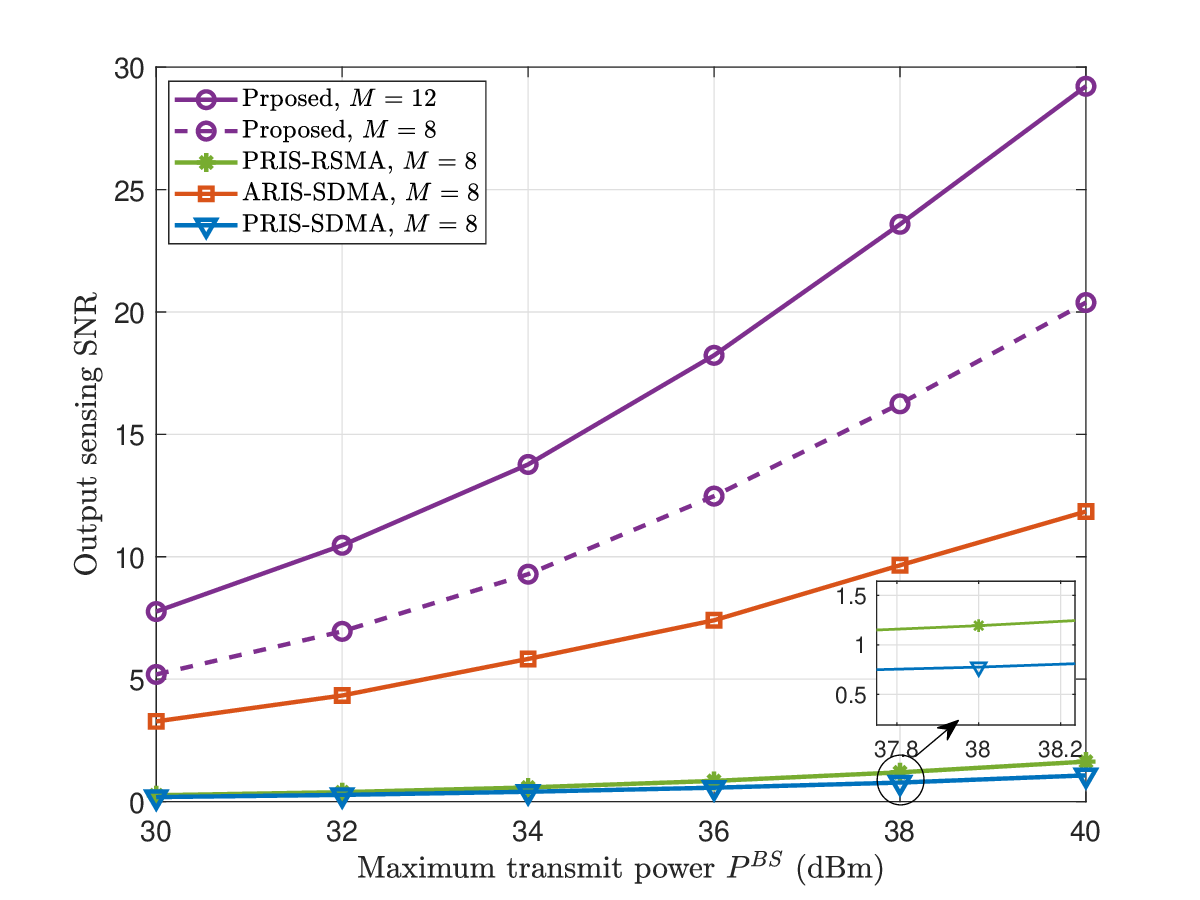}}
\caption{The achievable radar output SNR versus maximum BS power budget.}\label{TX_power_radar}
\end{figure}
We then investigate the impact of BS power budget on the radar output SNR in Fig.~\ref{TX_power_radar}, where it is observed that the sensing SNR grows exponentially with $P^{BS}$. Importantly, the gap between the proposed scheme and the benchmark schemes significantly increases with $P^{BS}$ as the actual transmission power $\sum\limits_{i=0}^{K}{{{\left\| {{\mathbf{w}}_{i}} \right\|}^{2}}}+{{\left\| \mathbf{z} \right\|}^{2}}$ monotonically enhances the sensing performance of (\ref{Radar_outSNR}). Although, the radar SNR of the active RIS-based scenarios appears significantly greater than that of the passive RIS-based systems, the enhancement rate of the radar SNR gain remains almost the same across the range of $P^{BS}$. To elaborate, the performance gain of the proposed scheme over ARIS-SDMA increases from $58.7$\% at $P^{BS}=30$ dBm to $72.13$\%  at $P^{BS}=40$, i.e., $13.43$\% of gain difference across $P^{BS}$ range. Meanwhile, the difference is amost $14.8$\% for PRIS-RSMA and PRIS-SDMA. This is explained by the fact that radar SNR values for active RIS-based schemes are scaled up by active RIS amplification coefficients.

\begin{figure}[t]
    \centering
    \subfloat[MaxMin EPSR versus $N$.]{\includegraphics[width=0.5\columnwidth, height=0.6\columnwidth]{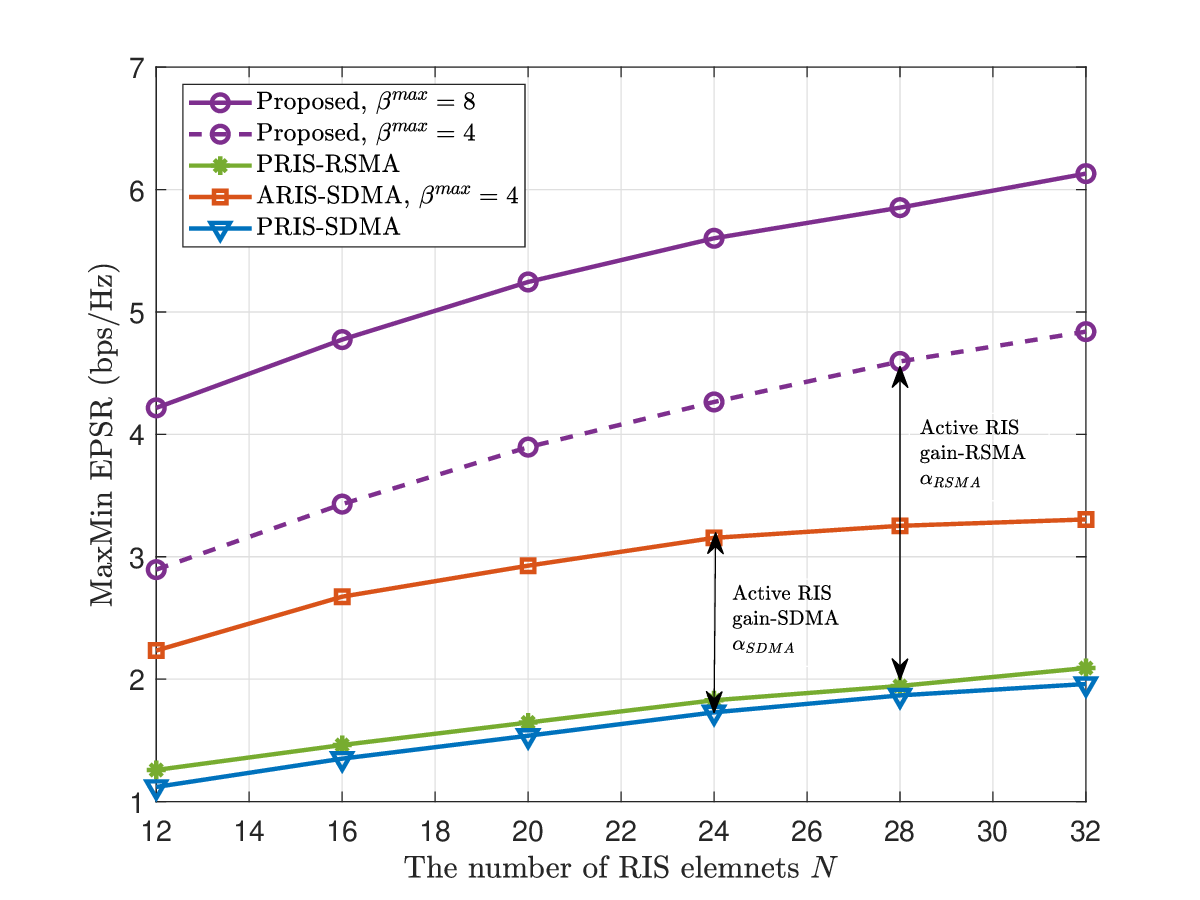}\label{RISelem_SR}}
    \hfill
    \subfloat[Radar SNR output versus $N$.]{\includegraphics[width=0.5\columnwidth, height=0.6\columnwidth]{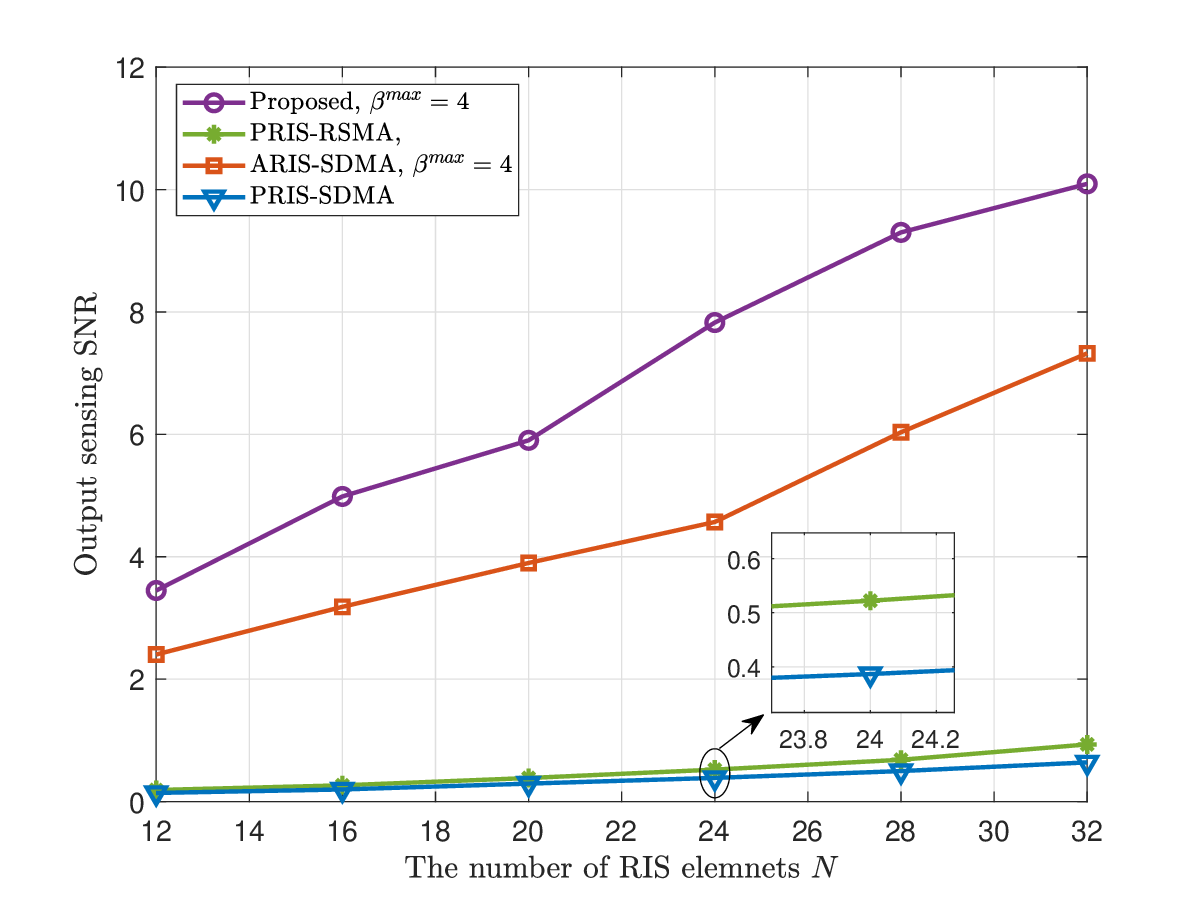}\label{RISelem_radar}}
    \caption{Number of RIS elements impact on ergodic secrecy rate and radar SNR output.}
    \label{RIS_elem}
\end{figure}
To further illustrate the effect of the RIS in improving the PLS and sensing performance, MaxMin-EPSR and radar output SNR are depicted in Fig.~\ref{RISelem_SR} and Fig.~\ref{RISelem_radar}, respectively. It is clear that MaxMin-EPSR and radar output SNR increase with $N$. The active and passive RIS gains come from the fact that with a larger number of elements, more signal enhancement can be achieved provided that the RIS reflection coefficients are optimized appropriately. In Fig.~\ref{RISelem_SR}, the active RIS performance gap of RSMA-based schemes, $\alpha_{RSMA}$, achieves a better score than its SDMA counterpart $\alpha_{SDMA}$. For instance, for $N=28$, $\alpha_{RSMA}$ is 136.6\%, which is an improvement over the 74.7\% achieved by $\alpha_{SDMA}$. An interesting observation is that the ARIS-SDMA scheme operating at $N=20$ shows a comparable performance to ARIS-RSMA at $N=12$, which implies that ARIS-SDMA could yield favorable PLS despite employing a larger RIS array. However, this advantage might come at the expense of increased power consumption. 

\begin{figure}[t]
\centering{\includegraphics[width=\columnwidth]{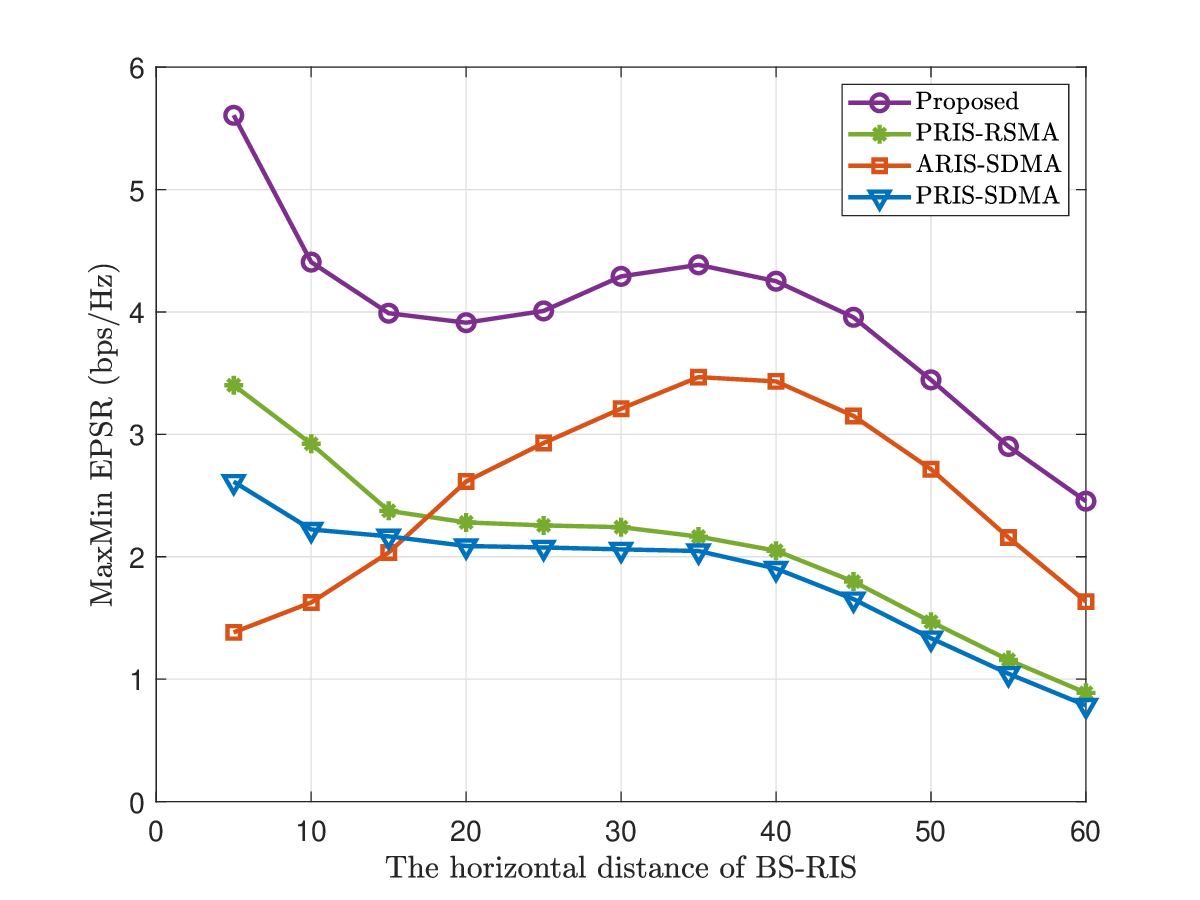}}
\caption{The achievable MaxMin EPSR versus the horizontal BS-RIS distance.}\label{RIS_pos}
\end{figure}
Building on the pivotal role of the RIS in enhancing the suggested system, Fig.~\ref{RIS_pos} illustrates the impact of varying the horizontal BS-RIS distance on the MaxMin EPSR performance. Across all the scenarios, the MaxMin EPSR performance decreases as the RIS shifts from the area of UEs towards the domain of the Eve ($d_{BR} > 35$ m), where the Eve's interception of the intended signals is enhanced by RIS array gain. For the proposed  ARIS-RSMA, the MaxMin EPSR decreases when the RIS moves away from BS ($d_{BR} > 5$ m) and increases again once it approaches the users ($d_{BR} > 20$ m). In contrast, for the ARIS-SDMA scheme, the MaxMin EPSR performance consistently increases as the RIS approaches the UE region, i.e., $5 \leq d_{BR} \leq 35$. Notably, the presence of the RIS in proximity to the BS does not yield performance enhancements comparable to those observed with alternative schemes. This can be credited to the active RIS's capability to enhance the intended streams and protect the common stream at the users' end. Moreover, the passive RIS-assisted scenarios exhibit gradual decay of the MaxMin EPSR performance when the RIS migrates away from BS, since the multiplicative path-loss effect significantly attenuates the signal propagation. Accordingly, the optimal placement of the RIS is at the user vicinity for active RIS-based systems and near the BS for the passive RIS-aided schemes.

\begin{figure}[t]
    \centering
    \subfloat[Output sensing SNR against $\Gamma_r$.]{\includegraphics[width=0.5\columnwidth, height=0.6\columnwidth]{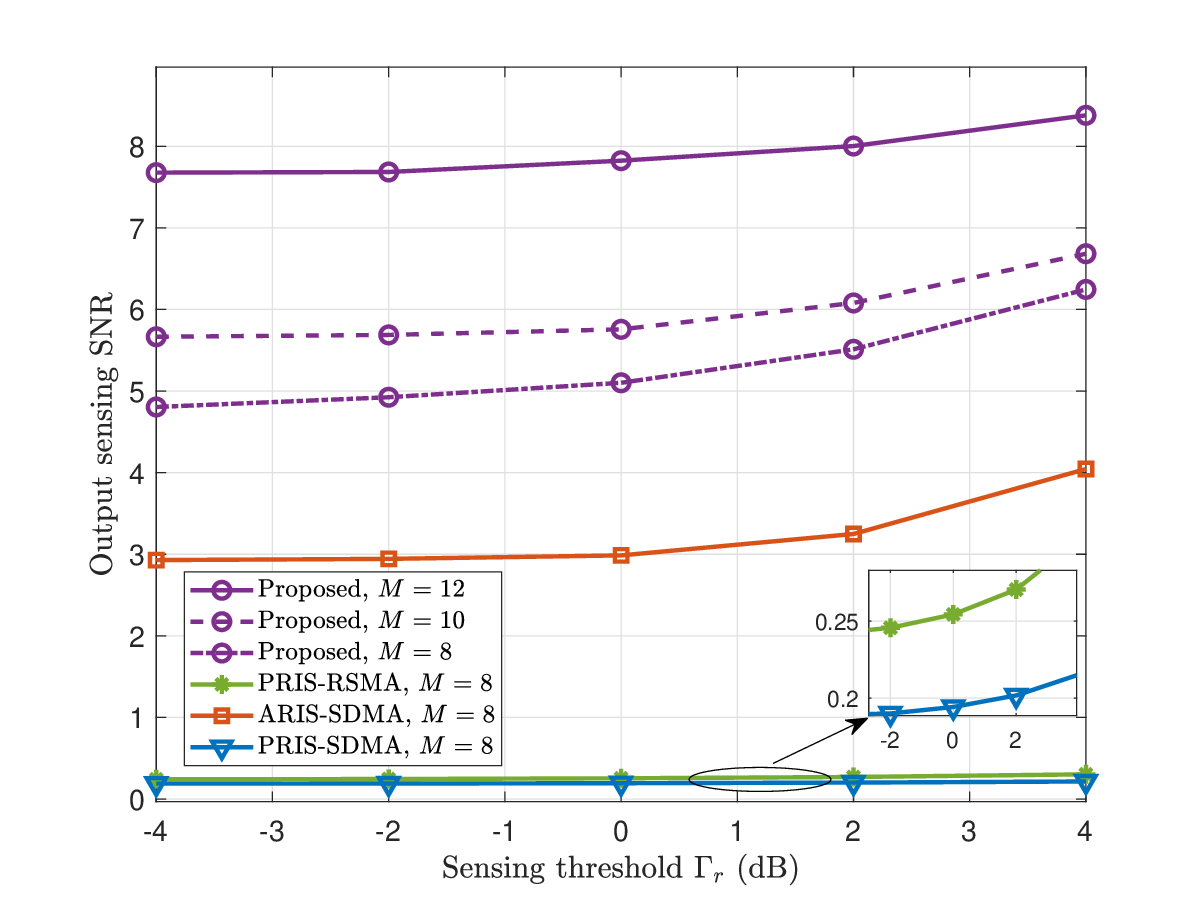}\label{radar_out}}
    \hfill
    \subfloat[The achievable ECSR against $\Gamma_r$.]{\includegraphics[width=0.5\columnwidth, height=0.6\columnwidth]{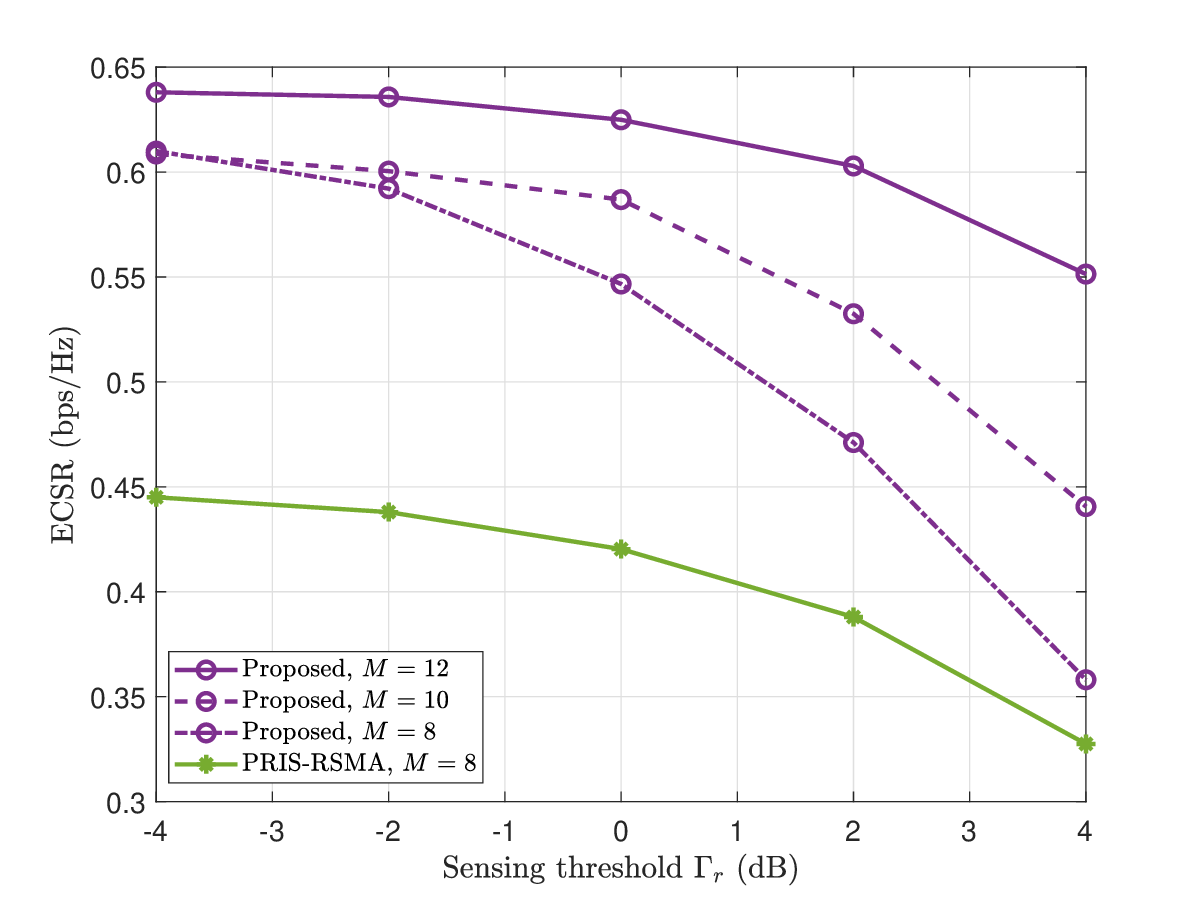}\label{radar_tradeoff}}
    \caption{The impact of radar requirement and performance trade-off.}
    \label{radar_requirement}
\end{figure}
Lastly, the impact of the radar requirement, $\Gamma_r$, on the output sensing SNR and ECSR is analyzed in Fig.~\ref{radar_out} and Fig.~\ref{radar_tradeoff}, respectively. In Fig.~\ref{radar_out}, we observe that the sensing performance improves with $\Gamma_r$, where the enhancement rate potentially grows when $\Gamma_r> 0$ for active RIS-assisted scenarios. Conversely, the improvement rate is marginal for passive RIS-enabled schemes. This is because optimizing the active RIS reflection coefficient matrix offers sufficient degrees of freedom to enhance radar returns despite the severe fades of the round trip path (BS-RIS-target-RIS-BS), which are further compromised by noise amplification. Moreover, Fig.~\ref{radar_tradeoff} shows the sensing and security trade-off in terms of the ECSR performance. The common rate splitting is exploited for sensing and securing private streams against the unknown Eve, while achieving fairness of private secure rates. Therefore, allocating higher values for the common stream rate, $r_k$, makes RSMA more prone to eavesdropping, while setting $r_k$ to zero incurs a loss in the system rate.

\section{Conclusion and future direction}\label{sec5}
In this paper, we investigated a novel secure communication and sensing approach for an active RIS-aided RSMA system in the presence of an Eve of unknown location. The EPSR and ECSR were mathematically derived based on the approximated mean of the Eve's channel gain. Accordingly, motivated by maximizing the minimum EPSR while ensuring the ECSR, sensing, active RIS power budget, and transmission power requirements, we jointly optimized the common rate splitting, transmit beamforming, AN, active RIS reflection coefficients, and radar receive beamformer via our proposed strategy. Simulation results demonstrated the effectiveness and the superiority of our proposed system against several benchmarks in terms of MaxMin EPSR, ECSR, and sensing output SNR. 

\appendices
\section{Derivation of $\mathbb{E}\left\{ {{\mathbf{h}}_{RE}}\mathbf{h}_{RE}^{H} \right\}$}\label{appendix1}
The expression ${{\widehat{\mathbf{H}}}_{RE}}=\mathbb{E}\left\{ {{\mathbf{h}}_{RE}}\mathbf{h}_{RE}^{H} \right\}$ can be expanded as 
\begin{align} \label{Eve_esti_CH}
    {{\widehat{\mathbf{H}}}_{RE}} &= \mathbb{E}\left\{ {{\chi }_{RE}}\left( \sqrt{\frac{\kappa }{1+\kappa }}\mathbf{h}_{RE}^{\text{LOS}}+\sqrt{\frac{1}{1+\kappa }}\mathbf{h}_{RE}^{\text{NLOS}} \right) \right. \nonumber \\ 
    &\quad \left. {{\left( \sqrt{\frac{\kappa }{1+\kappa }}\mathbf{h}_{RE}^{\text{LOS}}+\sqrt{\frac{1}{1+\kappa }}\mathbf{h}_{RE}^{\text{NLOS}} \right)}^{H}} \right\} \nonumber \\ 
    & =\frac{\ell }{1+\kappa }\int\limits_{{{d}_{1}}}^{{{d}_{2}}}{{{\left( {{{d}_{RE}}}/{{{d}_{0}}}\; \right)}^{-{{\alpha }_{RE}}}}f\left( {{d}_{RE}} \right)\mathrm{d}{{d}_{RE}}} \nonumber \\
    & \quad \times \left( \kappa \int\limits_{{{\theta }_{1}}}^{{{\theta }_{2}}}{\mathbf{a}\left( {{\theta }_{E}} \right){{\mathbf{a}}^{H}}\left( {{\theta }_{E}} \right)f\left( {{\theta }_{E}} \right)\mathrm{d}{{\theta }_{E}}}+{{\mathbf{I}}_{N}} \right),  
\end{align}
where $\ell$, $d_{RE}$, $\mathrm{d}{{d}_{RE}}$, and $\mathrm{d}{{\theta }_{E}}$ denote path-loss at reference distance the distance between RIS and possible Eve location, derivative operator of Eve distance and azimuth angle, respectively.
The marginal probability density functions (PDFs)  $f\left( {{d}_{RE}} \right)$ and $f\left( {{\theta }_{E}} \right)$ can be determined by evaluating the joint PDF $ f\left( {{d}_{RE}},{{\theta }_{E}} \right) $. The joint PDF can be calculated by the integral over area segment $ \mathcal{A} \triangleq \int\limits_{{{d}_{1}}}^{{{d}_{2}}}{\int\limits_{{{\theta }_{1}}}^{{{\theta }_{2}}}{{d}_{RE}} \mathrm{d}{{d}_{RE}}\mathrm{d}{{\theta }_{E}}}$ and then taking inverse.   
\begin{align}
    &f\left( {{d}_{RE}},{{\theta }_{E}} \right)=\mathcal{A}^{-1} =\frac{2}{\left( d_{2}^{2}-d_{1}^{2} \right)\left( {{\theta }_{2}}-{{\theta }_{1}} \right)},\label{joint_pdf} \\ 
    &f\left( {{d}_{RE}} \right)=\int\limits_{{{\theta }_{1}}}^{{{\theta }_{2}}}{\frac{2{{d}_{RE}}}{\left( d_{2}^{2}-d_{1}^{2} \right)\left( {{\theta }_{2}}-{{\theta }_{1}} \right)}\mathrm{d}{{\theta }_{E}}}=\frac{2{{d}_{RE}}}{d_{2}^{2}-d_{1}^{2}}, \label{dist_pdf} \\
    &f\left( {{\theta }_{E}} \right)=\int\limits_{{{d}_{1}}}^{{{d}_{2}}}{\frac{2{{d}_{RE}}}{\left( d_{2}^{2}-d_{1}^{2} \right)\left( {{\theta }_{2}}-{{\theta }_{1}} \right)}\mathrm{d}{{d}_{RE}}}=\frac{1}{{{\theta }_{2}}-{{\theta }_{1}}}. \label{ang_pdf}
\end{align}

By plugging (\ref{dist_pdf}) and (\ref{ang_pdf}) into (\ref{Eve_esti_CH}), ${{\widehat{\mathbf{H}}}_{RE}}$ can be given by 
\begin{multline} \label{eve_esti_CH_ver2}
    {{\widehat{\mathbf{H}}}_{RE}}={{\left( \frac{1}{{{d}_{0}}} \right)}^{-{{\alpha }_{RE}}}}\frac{2\ell \left( d_{1}^{2-{{\alpha }_{RE}}}-d_{2}^{2-{{\alpha }_{RE}}} \right)}{\left( 1+\kappa  \right)\left( d_{2}^{2}-d_{1}^{2} \right)\left( {{\alpha }_{RE}}-2 \right)} \\
    \times \left( \frac{\kappa }{\left( {{\theta }_{2}}-{{\theta }_{1}} \right)}\int\limits_{{{\theta }_{1}}}^{{{\theta }_{2}}}{\mathbf{a}\left( {{\theta }_{E}} \right){{\mathbf{a}}^{H}}\left( {{\theta }_{E}} \right)\mathrm{d}{{\theta }_{E}}}+{{\mathbf{I}}_{N}} \right),
\end{multline} 
where the integral over $\theta_E$ in can be evaluated via trapezoidal method. In specific, the interval of ${{\theta }_{E}}\in \left[ {{\theta }_{1}},{{\theta }_{2}} \right]$ is split into ${{N}_{{{\theta }_{E}}}}$ intervals, where the interval length is $\Delta {{\theta }_{E}}={\left( {{\theta }_{2}}-{{\theta }_{1}} \right)}/{{{N}_{{{\theta }_{E}}}}}\;$, with ${{\theta }_{E}}\in \left\{ {{\theta }_{1}}+i\Delta {{\theta }_{E}} \right\}_{i=0}^{{{N}_{{{\theta }_{E}}}}}$. Then, $\int\limits_{{{\theta }_{1}}}^{{{\theta }_{2}}}{\mathbf{a}\left( {{\theta }_{E}} \right){{\mathbf{a}}^{H}}\left( {{\theta }_{E}} \right)\mathrm{d}{{\theta }_{E}}}\triangleq \int\limits_{{{\theta }_{1}}}^{{{\theta }_{2}}}{\mathbf{A}\left( {{\theta }_{E}} \right)\mathrm{d}{{\theta }_{E}}}$ is approximated as 
\begin{multline} \label{Eve_azmith_approx}
    \int\limits_{{{\theta }_{1}}}^{{{\theta }_{2}}}{\mathbf{A}\left( {{\theta }_{E}} \right)\mathrm{d}{{\theta }_{E}}} \approx \sum\limits_{i=0}^{{{N}_{{{\theta }_{E}}}}-1}{\int\limits_{{{\theta }_{1}}+i\Delta {{\theta }_{E}}}^{{{\theta }_{1}}+\left( i+1 \right)\Delta {{\theta }_{E}}}{\mathbf{A}\left( {{\theta }_{E}} \right)\mathrm{d}{{\theta }_{E}}}}  = \frac{\Delta {{\theta }_{E}}}{2} \\
    \times \sum\limits_{i=0}^{{{N}_{{{\theta }_{E}}}}-1}{\mathbf{A}\left( {{\theta }_{1}}+i\Delta {{\theta }_{E}} \right)+\mathbf{A}\left( {{\theta }_{1}}+\left( i+1 \right)\Delta {{\theta }_{E}} \right)},
\end{multline}
where this approximation becomes more accurate as the ${{N}_{{{\theta }_{E}}}}$ is set sufficiently large. By plugging (\ref{Eve_azmith_approx}) into (\ref{eve_esti_CH_ver2}), the result  (\ref{eve_esti_CH_ver3}) can be obtained. To further simplify the subsequent analysis on the section \ref{sec3-2}, we employ decomposition  ${{\widehat{\mathbf{H}}}_{RE}}=\sum\limits_{n=1}^{N}{{{\lambda }_{E,n}}{{\mathbf{e}}_{E,n}}\mathbf{e}_{E,n}^{H}}$, where ${{\lambda }_{E,n}}$ and ${{\mathbf{e}}_{E,n}}$ represent the $n\text{-th}$ eigenvalues and eigenvectors of ${{\widehat{\mathbf{H}}}_{RE}}$, respectively.

\section{Proof of Lemma \ref{lem1}}\label{appendix2}
Here, we can re-express the ergodic rate of private stream $k$ at Eve as 
\small
\begin{align} \label{app1:ergo_rate_Eve}
&-{{\mathbb{E}}_{{{\mathbf{h}}_{RE}}}} \left\{ {{\log }_{2}}\left( 1+{{\gamma }_{k,E}} \right) \right\} = \mathbb{E}\left\{ \ln \left( {{{\bar{\sigma }}}_{E}} + {{\left| {{\left( {{\pmb{\theta }}^{\left( t \right)}} \right)}^{H}}{{{\mathbf{\bar{H}}}}_{E}}{{\mathbf{w}}_{0}} \right|}^{2}} \right. \right. \nonumber \\
& \left. \left. + \sum\limits_{i=1,i\ne k}^{K}{{{\left| {{\left( {{\pmb{\theta }}^{\left( t \right)}} \right)}^{H}}{{{\mathbf{\bar{H}}}}_{E}}{{\mathbf{w}}_{i}} \right|}^{2}}} + {{\left| {{\left( {{\pmb{\theta }}^{\left( t \right)}} \right)}^{H}}{{{\mathbf{\bar{H}}}}_{E}}\mathbf{z} \right|}^{2}} \right) \right\} \nonumber \\
& - \mathbb{E}\left\{ \ln \left( {{{\bar{\sigma }}}_{E}} + \sum\limits_{i=0}^{K}{{{\left| {{\left( {{\pmb{\theta }}^{\left( t \right)}} \right)}^{H}}{{{\mathbf{\bar{H}}}}_{E}}{{\mathbf{w}}_{i}} \right|}^{2}}} + {{\left| {{\left( {{\pmb{\theta }}^{\left( t \right)}} \right)}^{H}}{{{\mathbf{\bar{H}}}}_{E}}\mathbf{z} \right|}^{2}} \right) \right\}.
\end{align}
\normalsize
The first term of (\ref{app1:ergo_rate_Eve}) can be rewritten as 
\begin{equation} \label{first_term_eve_priv}
    \tilde{R}_{1k,E} = -\ln {{\left( {{{\bar{\sigma }}}_{E}}+\pmb{\omega }_{E,k}^{H}\mathbb{E}\left\{ {{\pmb{\Omega }}_{E,k}}\pmb{\Omega }_{E,k}^{H} \right\}{{\pmb{\omega }}_{E,k}} \right)}^{-1}}, 
\end{equation}
in which ${{\pmb{\omega }}_{E,k}}\in {{\mathbb{C}}^{M\left( K+1 \right)\times 1}}$ can attained by removing the $k$-th column from the block matrix ${{\mathbf{w}}^{\text{priv}}}$ of ${{\left[ \begin{matrix}
   \mathbf{w}_{0}^{T} & \underbrace{\begin{matrix}
   \mathbf{w}_{1}^{T} & \cdots  & \mathbf{w}_{K}^{T}  \\
\end{matrix}}_{{{\mathbf{w}}^{\text{priv}}}} & \mathbf{z}^{T}  \\
\end{matrix} \right]}^{T}}\in {{\mathbb{C}}^{M\left( K+2 \right)\times 1}}$. Besides, the expression $\mathbb{E}\left\{ {{\pmb{\Omega }}_{E,k}}\pmb{\Omega }_{E,k}^{H} \right\}{{\mathbb{C}}^{M\left( K+1 \right)\times M\left( K+1 \right)}}$ is denoted by
\begin{multline} \label{notation1}
    \mathbb{E}\left\{ {{\pmb{\Omega }}_{E,k}}\pmb{\Omega }_{E,k}^{H} \right\} = {{\mathbf{I}}_{K+1}}\otimes \mathbb{E}\left\{ \mathbf{\bar{H}}_{E}^{H}{{\pmb{\theta }}^{\left( t \right)}}{{\left( {{\pmb{\theta }}^{\left( t \right)}} \right)}^{H}}{{{\mathbf{\bar{H}}}}_{E}} \right\} \\
    = {{\mathbf{I}}_{K+1}}\otimes \sigma _{E}^{-2}{{\mathbf{G}}^{H}}{{\mathbf{\Phi }}^{^{\left( t \right)}}}{{\widehat{\mathbf{H}}}_{RE}}{{\left( {{\mathbf{\Phi }}^{^{\left( t \right)}}} \right)}^{H}}\mathbf{G}. 
\end{multline}
By resorting to the Proposition \ref{prop1}, $\bar{\sigma}_{E}$ can be calculated as  
\begin{equation} \label{seg_E_bar}
\bar{\sigma}_{E}=1+{{\left( {{\pmb{\theta}}^{\left( t \right)}} \right)}^{H}}{{\widetilde{\mathbf{J}}}_{E}}{{\pmb{\theta}}^{\left( t \right)}},
\end{equation}
with ${{\widetilde{\mathbf{J}}}_{E}}=\sigma _{R}^{2}\sigma _{E}^{-2}\sum\limits_{n=1}^{N}{{{\lambda }_{E,n}}\mathrm{diag}\left( \mathbf{e}_{E,n}^{H} {{\mathbf{e}}_{E,n}} \right)}$. Moreover, the diagonal block matrix $\mathbb{E}\left\{ {{\pmb{\Omega }}_{E,k}}\pmb{\Omega }_{E,k}^{H} \right\}\triangleq {{\widehat{\pmb{\Omega }}}_{E,k}}^{H}\widehat{\pmb{\Omega }}_{E,k}$ of (\ref{notation1}) can be rewritten as 
\begin{equation}
    {{\widehat{\pmb{\Omega }}}_{E,k}}^{H}\widehat{\pmb{\Omega }}_{E,k} = {{\mathbf{I}}_{K+1}}\otimes \sigma _{E}^{-2}{{\mathbf{G}}^{H}}{{\mathbf{\Phi }}^{^{\left( t \right)}}}{{\pmb{\overset{\scriptscriptstyle\frown}{J}}}_{E}}\pmb{\overset{\scriptscriptstyle\frown}{J}}_{E}^{H}{{\left( {{\mathbf{\Phi }}^{^{\left( t \right)}}} \right)}^{H}}\mathbf{G},
\end{equation}
where ${{\pmb{\overset{\scriptscriptstyle\frown}{J}}}_{E}}=\left[ \begin{matrix}
   \sqrt{{{\lambda }_{E,1}}}{{\mathbf{e}}_{E,1}} & \cdots  & \sqrt{{{\lambda }_{E,N}}}{{\mathbf{e}}_{E,N}}  \\
\end{matrix} \right]$ and ${{\widehat{\pmb{\Omega }}}_{E,k}}=\sigma _{E}^{-1}\left( {{\mathbf{I}}_{K+1}}\otimes \pmb{\overset{\scriptscriptstyle\frown}{J}}_{E}^{H}{{\left( {{\mathbf{\Phi }}^{^{\left( t \right)}}} \right)}^{H}}\mathbf{G} \right)\in {{\mathbb{C}}^{N\left( K+1 \right)\times M\left( K+1 \right)}}$.

By employing the matrix inversion property ${{\left( \mathbf{A}+\mathbf{BCD} \right)}^{-1}}={{\mathbf{A}}^{-1}}-{{\mathbf{A}}^{-1}}\mathbf{BC}{{\left( \mathbf{I}+\mathbf{D}{{\mathbf{A}}^{-1}}\mathbf{BC} \right)}^{-1}}\mathbf{D}{{\mathbf{A}}^{-1}}$ on the expression (\ref{first_term_eve_priv}), it can be reformulated as 
\begin{equation} \label{first_term_eve_priv_v2}
    \tilde{R}_{1k,E}  = -\ln \left( \bar{\sigma }_{E}^{-1}-\bar{\sigma }_{E}^{-2}\pmb{\omega }_{E,k}^{H}\widehat{\pmb{\Omega }}_{E,k}^{H}\mathbf{Q}_{E,k}^{-1}{{\widehat{\pmb{\Omega }}}_{E,k}}{{\pmb{\omega }}_{E,k}} \right),
\end{equation}
where ${{\mathbf{Q}}_{E,k}}={{\mathbf{I}}_{N\left( K+1 \right)}}+\bar{\sigma }_{E}^{-1}{{\widehat{\pmb{\Omega }}}_{E,k}}{{\pmb{\omega }}_{E,k}}\pmb{\omega }_{E,k}^{H}\widehat{\pmb{\Omega }}_{E,k}^{H}$. Nevertheless, the expression (\ref{first_term_eve_priv_v2}) is non-convex. The concave logarithm function can be convexified using the approximation $\ln \left( \gamma  \right)\le \ln \left( {{\gamma }^{\left( t \right)}} \right)+\frac{1}{{{\gamma }^{\left( t \right)}}}\left( \gamma -{{\gamma }^{\left( t \right)}} \right)=\ln \left( {{\gamma }^{\left( t \right)}} \right)+\frac{\gamma }{{{\gamma }^{\left( t \right)}}}-1$ as 
\begin{equation} \label{first_term_eve_priv_approx}
\resizebox{\columnwidth}{!}{$
  \tilde{R}_{1k,E}  \ge -\ln \left( \frac{1-q_{E,k}^{\left( t \right)}}{{{{\bar{\sigma }}}_{E}}} \right)+\frac{\bar{\sigma }_{E}^{-1}\pmb{\omega }_{E,k}^{H}\widehat{\pmb{\Omega }}_{E}^{H}\mathbf{Q}_{E,k}^{-1}{{\widehat{\pmb{\Omega }}}_{E}}{{\pmb{\omega }}_{E,k}}-q_{E,k}^{\left( t \right)}}{1-q_{E,k}^{\left( t \right)}}.  
  $}
\end{equation}
The matrix fractional form $\pmb{\omega }_{E,k}^{H}\widehat{\pmb{\Omega }}_{E,k}^{H}\mathbf{Q}_{E,k}^{-1}{{\widehat{\pmb{\Omega }}}_{E,k}}{{\pmb{\omega }}_{E,k}}$ can be approximated using the following approximation of \cite{dong2020secure}:
\small
\begin{multline*} 
     \mathrm{Tr}\left( \mathbf{AC}{{\mathbf{B}}^{-1}}{{\mathbf{C}}^{H}} \right)\ge \mathrm{Tr}\left( \mathbf{A}{{\mathbf{C}}^{\left( t \right)}}{{\left( {{\mathbf{B}}^{\left( t \right)}} \right)}^{-1}}{{\left( {{\mathbf{C}}^{\left( t \right)}} \right)}^{H}} \right) \\
     -\mathrm{Tr}\left( \mathbf{A}{{\mathbf{C}}^{\left( t \right)}}{{\left( {{\mathbf{B}}^{\left( t \right)}} \right)}^{-1}}\left( \mathbf{B}-{{\mathbf{B}}^{\left( t \right)}} \right){{\left( {{\mathbf{B}}^{\left( t \right)}} \right)}^{-1}}{{\left( {{\mathbf{C}}^{\left( t \right)}} \right)}^{H}} \right) \\
     +\mathrm{Tr}\left( \mathbf{A}\left( \mathbf{C}-{{\mathbf{C}}^{\left( t \right)}} \right){{\left( {{\mathbf{B}}^{\left( t \right)}} \right)}^{-1}}{{\left( {{\mathbf{C}}^{\left( t \right)}} \right)}^{H}} \right) \\
     +\mathrm{Tr}\left( \mathbf{A}{{\mathbf{C}}^{\left( t \right)}}{{\left( {{\mathbf{B}}^{\left( t \right)}} \right)}^{-1}}{{\left( \mathbf{C}-{{\mathbf{C}}^{\left( t \right)}} \right)}^{H}} \right),
\end{multline*}
\normalsize
where $\mathbf{A}\in {{\mathbb{C}}^{m\times m}},\mathbf{B}\in {{\mathbb{C}}^{n\times n}}$ and $\mathbf{C}\in {{\mathbb{C}}^{m\times n}}$are positive semi-definite matrices. Therefore, we have
\begin{multline} \label{eve_frac_matrix_approx1}
  \bar{\sigma }_{E}^{-1}\pmb{\omega }_{E,k}^{H}\widehat{\pmb{\Omega }}_{E,k}^{H}\mathbf{Q}_{E,k}^{-1}{{\widehat{\pmb{\Omega }}}_{E,k}}{{\pmb{\omega }}_{E,k}}  \ge  q_{E,k}^{\left( t \right)}-\bar{\sigma }_{E}^{-1}{{\left( \pmb{\omega }_{E,k}^{\left( t \right)} \right)}^{H}}\widehat{\pmb{\Omega }}_{E,k}^{H} \\ 
  \times {{\left( \mathbf{Q}_{E,k}^{\left( t \right)} \right)}^{-1}}\left( {{\mathbf{Q}}_{E,k}}-\mathbf{Q}_{E,k}^{\left( t \right)} \right){{\left( \mathbf{Q}_{E,k}^{\left( t \right)} \right)}^{-1}}{{\widehat{\pmb{\Omega }}}_{E,k}}\pmb{\omega }_{E,k}^{\left( t \right)} \\ 
  +2\Re \left\{ \bar{\sigma }_{E}^{-1}{{\left( \pmb{\omega }_{E,k}^{\left( t \right)} \right)}^{H}}\widehat{\pmb{\Omega }}_{E,k}^{H}{{\left( \mathbf{Q}_{E,k}^{\left( t \right)} \right)}^{-1}}{{\widehat{\pmb{\Omega }}}_{E,k}}\left( {{\pmb{\omega }}_{E,k}}-\pmb{\omega }_{E,k}^{\left( t \right)} \right) \right\}  
\end{multline}
By substituting (\ref{eve_frac_matrix_approx1}) into (\ref{first_term_eve_priv_approx}) yields the final approximation of (\ref{eve_Priv_first_approx_final}) in the top of the next page. 

The second term of (\ref{app1:ergo_rate_Eve}) can be represented as 
\begin{floatEq}
    \begin{align} 
        \tilde{R}_{1k,E} =& -\ln \left( \frac{1-q_{E,k}^{\left( t \right)}}{{{{\bar{\sigma }}}_{E}}} \right)+\frac{2\Re \left\{ \bar{\sigma }_{E}^{-1}{{\left( \pmb{\omega }_{E,k}^{\left( t \right)} \right)}^{H}}\widehat{\pmb{\Omega }}_{E,k}^{H}{{\left( \mathbf{Q}_{E,k}^{\left( t \right)} \right)}^{-1}}{{\widehat{\pmb{\Omega }}}_{E,k}}\left( {{\pmb{\omega }}_{E,k}}-\pmb{\omega }_{E,k}^{\left( t \right)} \right) \right\}}{1-q_{E,k}^{\left( t \right)}} \nonumber \\ 
        & -\frac{\bar{\sigma }_{E}^{-1}{{\left( \pmb{\omega }_{E,k}^{\left( t \right)} \right)}^{H}}\widehat{\pmb{\Omega }}_{E,k}^{H}{{\left( \mathbf{Q}_{E,k}^{\left( t \right)} \right)}^{-1}}\left( {{\mathbf{Q}}_{E,k}}-\mathbf{Q}_{E,k}^{\left( t \right)} \right){{\left( \mathbf{Q}_{E,k}^{\left( t \right)} \right)}^{-1}}{{\widehat{\pmb{\Omega }}}_{E,k}}\pmb{\omega }_{E,k}^{\left( t \right)}}{1-q_{E,k}^{\left( t \right)}}.  
        \label{eve_Priv_first_approx_final}
    \end{align}
\end{floatEq} 
\small
\begin{equation} \label{seconde_term_Eve_priv}
   \tilde{R}_{2k,E} =-\ln \left( 1+{{\ell }_{E}} \right),  
\end{equation}
\normalsize
where ${{\ell }_{E}}=\sum\limits_{i=0}^{K}{\mathbf{w}_{i}^{H}{{\widehat{\mathbf{G}}}_{E}}{{\mathbf{w}}_{i}}}+\mathbf{z}^{H}{{\widehat{\mathbf{G}}}_{E}}\mathbf{z}+{{\left( {{\pmb{\theta }}^{\left( t \right)}} \right)}^{H}}{{\widetilde{\mathbf{J}}}_{E}}{{\pmb{\theta }}^{\left( t \right)}}$ and ${{\widehat{\mathbf{G}}}_{E}}=\sigma _{E}^{-2}{{\mathbf{G}}^{H}}{{\mathbf{\Phi }}^{^{\left( t \right)}}}{{\widehat{\mathbf{H}}}_{RE}}{{\left( {{\mathbf{\Phi }}^{^{\left( t \right)}}} \right)}^{H}}\mathbf{G}$. The expression in (\ref{seconde_term_Eve_priv}) can be lower bounded as 
\begin{equation} \label{seconde_term_Eve_priv_approx}
    -\ln \left( 1+{{\ell }_{E}} \right)\ge -\ln \left( 1+\ell _{E}^{\left( t \right)} \right)-\frac{1+{{\ell }_{E}}}{1+\ell _{E}^{\left( t \right)}}+1.
\end{equation}
To further simplify the third term of (\ref{eve_Priv_first_approx_final}), we employ the following expansion  
\begin{multline}\label{approx_QEk}
    \bar{\sigma }_{E}^{-1}{{\left( \pmb{\omega }_{E,k}^{\left( t \right)} \right)}^{H}}\widehat{\pmb{\Omega }}_{E,k}^{H}{{\left( \mathbf{Q}_{E,k}^{\left( t \right)} \right)}^{-1}}{{\mathbf{Q}}_{E,k}}{{\left( \mathbf{Q}_{E,k}^{\left( t \right)} \right)}^{-1}}{{\widehat{\pmb{\Omega }}}_{E,k}}\pmb{\omega }_{E,k}^{\left( t \right)}
    \\={{\varepsilon }_{E,k}}+\pmb{\omega }_{E,k}^{H}\mathbf{\bar{Q}}_{E,k}^{\left( t \right)}{{\pmb{\omega }}_{E,k}}.
\end{multline} 
By plugging the results  (\ref{eve_Priv_first_approx_final}), (\ref{seconde_term_Eve_priv_approx}), and (\ref{approx_QEk}) into (\ref{app1:ergo_rate_Eve}), the result in (\ref{ergo_PrivRate_Eve_final}) can be attained.

\section{Proof of Lemma \ref{lem2}}\label{appendix3}
\begin{floatEq}
    \begin{align} 
        -{{\mathbb{E}}_{{{\mathbf{h}}_{RE}}}}\left\{ {{\log }_{2}}\left( 1+{{\gamma }_{k, E}} \right) \right\}=&\mathbb{E}\left\{ \ln \left( {{\bar{\sigma }}_{E}}+{{\left| {{\pmb{\theta }}^{H}}{{{\mathbf{\bar{H}}}}_{E}}\mathbf{w}_{0}^{\left( t \right)} \right|}^{2}}+\sum\limits_{i=1,i\ne k}^{K}{{{\left| {{\pmb{\theta }}^{H}}{{{\mathbf{\bar{H}}}}_{E}}\mathbf{w}_{i}^{\left( t \right)} \right|}^{2}}}+{{\left| {{\pmb{\theta }}^{H}}{{{\mathbf{\bar{H}}}}_{E}}\mathbf{z}^{\left( t \right)} \right|}^{2}} \right) \right\} \nonumber \\ 
        & -\mathbb{E}\left\{ \ln \left( {{\bar{\sigma }}_{E}}+\sum\limits_{i=0}^{K}{{{\left| {{\pmb{\theta }}^{H}}{{{\mathbf{\bar{H}}}}_{E}}\mathbf{w}_{i}^{\left( t \right)} \right|}^{2}}}+{{\left| {{\pmb{\theta }}^{H}}{{{\mathbf{\bar{H}}}}_{E}}\mathbf{z}^{\left( t \right)} \right|}^{2}} \right) \right\}.   
        \label{eve_Priv_RIS1}
    \end{align}
\end{floatEq} 
First, the ergodic Eve can be rewritten in term of $\theta$ as (\ref{eve_Priv_RIS1}) at the top of the next page, in which we define   
\begin{multline*}
  {{\bar{\sigma }}_{E}}+{{\left| {{\pmb{\theta }}^{H}}{{{\mathbf{\bar{H}}}}_{E}}\mathbf{w}_{0}^{\left( t \right)} \right|}^{2}} + {{\left| {{\pmb{\theta }}^{H}}{{{\mathbf{\bar{H}}}}_{E}}\mathbf{z}^{\left( t \right)} \right|}^{2}}\\
  +\sum\limits_{i=1,i\ne k}^{K}{{{\left| {{\pmb{\theta }}^{H}}{{{\mathbf{\bar{H}}}}_{E}}\mathbf{w}_{i}^{\left( t \right)} \right|}^{2}}}  = 1+{{\pmb{\theta }}^{H}}{{\mathbf{\Psi }}_{E,k}}\mathbf{\Psi }_{E,k}^{H}\pmb{\theta },
\end{multline*}
where ${{\mathbf{\Psi }}_{E,k}}\mathbf{\Psi }_{E,k}^{H}=\sigma _{R}^{2}\mathrm{diag}\left( {{{\mathbf{\bar{h}}}}_{RE}} \mathbf{\bar{h}}_{RE}^{H} \right)+{{\mathbf{\bar{H}}}_{E}}{\mathbf{\Upsilon }^{\left( t \right)}}\mathbf{\bar{H}}_{E}^{H}$, and ${\mathbf{\Upsilon }^{\left( t \right)}}=\mathbf{w}_{0}^{\left( t \right)}{{\left( \mathbf{w}_{0}^{\left( t \right)} \right)}^{H}}+\sum\limits_{i=1,i\ne k}^{K}{\mathbf{w}_{i}^{\left( t \right)}{{\left( \mathbf{w}_{i}^{\left( t \right)} \right)}^{H}}}+{{\mathbf{z}}^{\left( t \right)}}{{\left( {{\mathbf{z}}^{\left( t \right)}} \right)}^{H}}$. Again, we re-invoke Proposition \ref{prop1} to recast the first term of (\ref{eve_Priv_RIS1}) as 
\begin{equation}\label{first_term_eve_RIS}
    \bar{R}_{1k,E} = -\ln {{\left( 1+{{\pmb{\theta }}^{H}}{{\widehat{\mathbf{\Psi }}}_{E,k}}\widehat{\mathbf{\Psi }}_{E,k}^{H}\pmb{\theta } \right)}^{-1}},
\end{equation}
where ${{\widehat{\mathbf{\Psi }}}_{E,k}}\widehat{\mathbf{\Psi }}_{E,k}^{H} \triangleq  \mathbb{E}\left\{ {{\mathbf{\Psi }}_{E,k}}\mathbf{\Psi }_{E,k}^{H} \right\}$,  ${{\widehat{\mathbf{\Psi }}}_{E,k}}=\sigma _{E}^{-1}\left[ {{{\pmb{\overset{\scriptscriptstyle\frown}{E}}}}_{E}}, {{{\pmb{\overset{\scriptscriptstyle\frown}{G}}}}_{0E}}, {{{\pmb{\overset{\scriptscriptstyle\frown}{G}}}}_{E,k}}, {{{\pmb{\overset{\scriptscriptstyle\frown}{Z}}}}_{E,\mathrm{AN}}} \right]\in {{\mathbb{C}}^{N\times \bar{N}}}$, and  $\bar{N}={{N}^{2}}+2N+NK$. Let us define ${{\mathbf{\tilde{D}}}_{E,1}}=\sqrt{{{\lambda }_{E,1}}}\mathrm{diag}\left( \mathbf{e}_{E,1}^{H} \right)$  and ${{\mathbf{\tilde{D}}}_{E,N}}=\sqrt{{{\lambda }_{E,N}}}\mathrm{diag}\left( \mathbf{e}_{E,N}^{H} \right)$, then we have 
\begin{align*}
    {{\pmb{\overset{\scriptscriptstyle\frown}{E}}}_{E}}=&\left[ 
    {{\sigma }_{R}}{{\mathbf{\tilde{D}}}_{E,1}}, \cdots, {{\sigma }_{R}}{{\mathbf{\tilde{D}}}_{E,N}} \right] \in {{\mathbb{C}}^{N\times {{N}^{2}}}}, \\
    {{\pmb{\overset{\scriptscriptstyle\frown}{G}}}_{0E}} =& \left[ 
   {{\mathbf{\tilde{D}}}_{E,1}} \mathbf{Gw}_{0}^{\left( t \right)}, \cdots, {{\mathbf{\tilde{D}}}_{E,N}} \mathbf{Gw}_{0}^{\left( t \right)} \right] \in {{\mathbb{C}}^{N\times N}}, \\
    {{\pmb{\overset{\scriptscriptstyle\frown}{G}}}_{E,k}}=&\left[ 
    {{\mathbf{\tilde{D}}}_{E,1}}\mathbf{G\tilde{W}}_{k}^{\left( t \right)}, \cdots, {{\mathbf{\tilde{D}}}_{E,N}} \mathbf{G\tilde{W}}_{k}^{\left( t \right)} \right]\in {{\mathbb{C}}^{N\times NK}},\\
   \mathbf{\tilde{W}}_{k}^{\left( t \right)}=&\left[\mathbf{w}_{1}^{\left( t \right)}, \cdots, \underbrace{{{\mathbf{0}}_{M\times 1}}}_{k-\text{th term}}, \cdots, \mathbf{w}_{K}^{\left( t \right)}  \right], \\
   {{\pmb{\overset{\scriptscriptstyle\frown}{Z}}}_{E,\mathrm{AN}}}=&\left[ 
   {{\mathbf{\tilde{D}}}_{E,1}}\mathbf{Gz}^{\left( t \right)}, \cdots, {{\mathbf{\tilde{D}}}_{E,N}} \mathbf{Gz}^{\left( t \right)}   \right]\in {{\mathbb{C}}^{N\times N}}.
\end{align*}

By applying similar steps in (\ref{first_term_eve_priv_v2}), (\ref{first_term_eve_priv_approx}), and (\ref{eve_frac_matrix_approx1}), the first term of (\ref{eve_Priv_RIS1}) can be re-expressed as in   (\ref{1stTerm_Eve_privk_fianlRIS}) at the top of the next page, where ${{\mathbf{\Xi }}_{E,k}}={{\mathbf{I}}_{{\bar{N}}}}+\widehat{\mathbf{\Psi }}_{E,k}^{H}\pmb{\theta }{{\pmb{\theta }}^{H}}{{\widehat{\mathbf{\Psi }}}_{E,k}}$.
\begin{floatEq}
	\begin{align} 
            \bar{R}_{1k,E} =&-\ln \left( 1-u_{E,k}^{\left( t \right)} \right)+\frac{2\Re \left\{ {{\left( {{\pmb{\theta }}^{\left( t \right)}} \right)}^{H}}{{\widehat{\mathbf{\Psi }}}_{E,k}}{{\left( \mathbf{\Xi }_{E,k}^{\left( t \right)} \right)}^{-1}}\widehat{\mathbf{\Psi }}_{E,k}^{H}\left( \pmb{\theta }-{{\pmb{\theta }}^{\left( t \right)}} \right) \right\}}{1-u_{E,k}^{\left( t \right)}} \nonumber \\ 
            & -\frac{{{\left( {{\pmb{\theta }}^{\left( t \right)}} \right)}^{H}}{{\widehat{\mathbf{\Psi }}}_{E,k}}{{\left( \mathbf{\Xi }_{E,k}^{\left( t \right)} \right)}^{-1}}\left( {{\mathbf{\Xi }}_{E,k}}-\mathbf{\Xi }_{E,k}^{\left( t \right)} \right){{\left( \mathbf{\Xi }_{E,k}^{\left( t \right)} \right)}^{-1}}\widehat{\mathbf{\Psi }}_{E,k}^{H}{{\pmb{\theta }}^{\left( t \right)}}}{1-u_{E,k}^{\left( t \right)}}.  \label{1stTerm_Eve_privk_fianlRIS} 
	\end{align}
\end{floatEq}

Similar to (\ref{seconde_term_Eve_priv_approx}), the second term of (\ref{eve_Priv_RIS1}) can be approximated as 
\begin{equation} \label{2ndTerm_Eve_privk_approx1}
   \bar{R}_{2k,E} \ge -\ln \left( 1+\mu _{E}^{\left( t \right)} \right)-\frac{1+{{\mu }_{E}}}{1+\mu _{E}^{\left( t \right)}}+1.  
\end{equation}
Based on the result of Proposition \ref{prop1}, we re-write $\mu_E$ as ${{\mu }_{E}} ={{\pmb{\theta }}^{H}}{{\mathbf{\Gamma }}_{E}}\pmb{\theta }$, where ${{\mathbf{\Gamma }}_{E}}=\sigma _{E}^{-2}\left( {{{\pmb{\overset{\scriptscriptstyle\frown}{E}}}}_{E}}\pmb{\overset{\scriptscriptstyle\frown}{E}}_{E}^{H}+{{{\pmb{\overset{\scriptscriptstyle\frown}{G}}}}_{E}}\pmb{\overset{\scriptscriptstyle\frown}{G}}_{E}^{H} \right)$, ${{\pmb{\overset{\scriptscriptstyle\frown}{G}}}_{E}}=\left[ {{\mathbf{\tilde{D}}}_{E,1}}\mathbf{G}{{{\mathbf{\tilde{W}}}}^{\left( t \right)}},\cdots, {{\mathbf{\tilde{D}}}_{E,N}} \mathbf{G}{{{\mathbf{\tilde{W}}}}^{\left( t \right)}}  \right]$, and ${{\mathbf{\tilde{W}}}^{\left( t \right)}}=\left[ \mathbf{w}_{0}^{\left( t \right)}, \cdots, \mathbf{w}_{K}^{\left( t \right)},  \mathbf{z}^{\left( t \right)}  \right]$. By combining (\ref{2ndTerm_Eve_privk_approx1}) and (\ref{1stTerm_Eve_privk_fianlRIS}) into (\ref{eve_Priv_RIS1}), we obtain the result of (\ref{ergo_PrivRate_Eve_finalRIS}).

\bibliographystyle{ieeetr}
\bibliography{referances}

\end{document}